%% file: main.tex
\documentclass{article}

\usepackage[margin=1in]{geometry}
\usepackage[textsize=tiny,textwidth=2cm,color=green!50!gray]{todonotes} 
\usepackage{svg}
\usepackage{makecell}
\usepackage{tikz}
\usepackage{soul}
\usetikzlibrary{trees}
\usepackage{pstricks}
\usepackage{amsthm,amsmath,amssymb,mathtools}
\usepackage{thm-restate}
\usepackage{float}
\usepackage{hyperref}
\usepackage[capitalize]{cleveref}

\input{macros}
\date{\empty}

\newtheorem{theorem}{Theorem}
\newtheorem{lemma}{Lemma}[section]
\newtheorem{claim}{Claim}[section]
\newtheorem{proposition}[lemma]{Proposition}

\newtheorem{remark}{Remark}[section]

\hypersetup{
    colorlinks=true,
    citecolor=magenta!80!black,
    urlcolor=magenta!80!black,
    linkcolor=green!40!black
}
\begin{document}

\newcommand\relatedversion{}

\title{\Large Bin Packing under Random-Order: Breaking the Barrier of 3/2\relatedversion}
\author{Anish Hebbar\thanks{Department of Computer Science, Duke University, Durham, USA. This work was done when the author was a student at Indian Institute of Science, Bengaluru. Email: \texttt{anishshripad.hebbar@duke.edu}}
\and Arindam Khan\thanks{Department of Computer Science and Automation, Indian Institute of Science, Bengaluru, India. Research partly supported by Pratiksha Trust Young Investigator Award, Google India Research Award, and SERB Core Research Grant (CRG/2022/001176) on “Optimization under Intractability and Uncertainty”. Email: \texttt{arindamkhan@iisc.ac.in}}
\and K. V. N. Sreenivas\thanks{Department of Computer Science and Automation, Indian Institute of Science, Bengaluru, India. Supported in part by Google PhD Fellowship. Email: \texttt{venkatanaga@iisc.ac.in}}}

\date{}

\maketitle







\begin{abstract} 
\small\baselineskip=9pt 
\bestfit{} is one of the most prominent and practically used algorithms for the bin packing problem, where a set of items with associated sizes needs to be packed in the minimum number of unit-capacity bins. Kenyon [SODA '96] studied online bin packing under random-order arrival, where the adversary chooses the list of items, but the items arrive one by one according to an arrival order drawn uniformly at random from the set of all permutations of the items. Kenyon's seminal result established an upper bound of $1.5$ and a lower bound of $1.08$ on the {\em random-order ratio} of \bestfit{}, and it was conjectured that the true ratio is $\approx 1.15$. The conjecture, if true, will also imply that \bestfit{} (on randomly permuted input) has the best performance guarantee among all the widely-used simple algorithms for (offline) bin packing. This conjecture has remained one of the major open problems in the area, as highlighted in the recent survey on random-order models by Gupta and Singla [Beyond the Worst-Case Analysis of Algorithms '20]. Recently, Albers et al. [Algorithmica '21] improved the upper bound to $1.25$ for the special case when all the item sizes are greater than $1/3$, and they improve the lower bound to $1.1$. Ayyadevara et al. [ICALP '22] obtained an improved result for the special case when all the item sizes lie in $(1/4, 1/2]$, which corresponds to the {\em $3$-partition} problem. The upper bound of 3/2 for the general case, however, has remained unimproved. 
This also has remained the best random-order ratio among all polynomial-time algorithms for online bin packing. 

In this paper, we make the first progress towards the conjecture, by showing that \bestfit{} achieves a random-order ratio of at most $1.5 - \epsilon$, for a small constant $\epsilon>0$.  Furthermore, we establish an improved lower bound of $1.144$ on the random-order ratio of \bestfit{}, nearly reaching the conjectured ratio. 
\end{abstract}
\input{intro}

\input{prior-work}

\input{our-contributions}
\input{related-work}
\input{organization}
\input{notation}


\input{general-case.tex}

\input{bestfit-lowerbound}
\input{conclusion}
\input{acknowledgments.tex}
\appendix
\input{appendix-omitted-proofs.tex}
\bibliographystyle{alpha}
\bibliography{ref.bib}
\end{document}

%% file: macros.tex
\newcommand{\kvnr}[1]{}
\newcommand{\arir}[1]{}
\newcommand{\ahr}[1]{}

\newcommand{\kvn}[1]{{#1}}
\newcommand{\ah}[1]{{#1}}

\newcommand{\Bf}{\mathrm{BF}}
\newcommand{\BF}{\mathrm{BF}}

\newcommand{\vol}{\mathrm{vol}}
\newcommand{\opt}{\mathrm{Opt}}
\newcommand{\Opt}{\mathrm{Opt}}

\newcommand{\Th}{^{\mathrm{th}}}
\newcommand{\floor}[1]{\left\lfloor#1\right\rfloor}
\newcommand{\ceil}[1]{\left\lceil#1\right\rceil}
\newcommand{\abs}[1]{\left|#1\right|}
\newcommand{\prob}[2][]{\mathbb{P}_{#1}\left[#2\right]}
\newcommand{\expec}[2][]{\mathbb{E}_{#1}\left[#2\right]}
\newcommand{\var}[2][]{\mathrm{Var}_{#1}\left[#2\right]}
\newcommand{\cov}[3][]{\mathrm{Cov}_{#1}\left[#2,#3\right]}

\newcommand{\rank}{\mathrm{rank}}

\newcommand{\case}[1]{\textsf{Case~#1}}
\newcommand{\Case}[1]{\textsf{Case~#1}}
\newcommand{\casess}[1]{\textsf{Cases~#1}}

\newcommand{\bestfit}{Best-Fit}
\newcommand{\nextfit}{Next-Fit}

\newcommand{\whp}{w.h.p.}

\newcommand{\T}{{T}}
\renewcommand{\S}{{S}}
\newcommand{\M}{{M}}
\renewcommand{\L}{{L}}

\newcommand{\LS}{{LS}}

\renewcommand{\epsilon}{\varepsilon}
\renewcommand{\tilde}{\widetilde}
\renewcommand{\hat}{\widehat}

\newcommand{\OPT}{\mathrm{Opt}}
\newcommand{\eps}{\varepsilon}
\newcommand{\mycomment}[1]{}

\newcolumntype{L}{>{$}l<{$}} 
\makeatletter
\def\blfootnote{\gdef\@thefnmark{}\@footnotetext}
\makeatother

\newcommand{\typeone}{\textrm{Type-}1}
\newcommand{\typetwo}{\textrm{Type-}3/2}

\newcommand{\eone}{E_1}
\newcommand{\eoneone}{E_{11}}
\newcommand{\eoneoneone}{E_{111}}
\newcommand{\eoneonetwo}{E_{112}}
\newcommand{\eonetwo}{E_{12}}
\newcommand{\etwo}{E_{2}}
\newcommand{\etwoone}{E_{21}}
\newcommand{\etwooneone}{E_{211}}
\newcommand{\etwoonetwo}{E_{212}}
\newcommand{\etwotwo}{E_{22}}

\newcommand{\poneone}{p_{11}}
\newcommand{\poneoneone}{p_{111}}
\newcommand{\poneonetwo}{p_{112}}
\newcommand{\ponetwo}{p_{12}}

\newcommand{\ptwoone}{p_{21}}
\newcommand{\ptwooneone}{p_{211}}
\newcommand{\ptwoonetwo}{p_{212}}
\newcommand{\ptwotwo}{p_{22}}

\newenvironment{proofsketch}{\begin{proof}[Proof Sketch]}{\end{proof}}

\newenvironment{deferredproof}[1]{\begin{proof}[Proof of #1]}{\end{proof}}

%% file: intro.tex
\section{Introduction}

Bin packing is a fundamental strongly NP-complete \cite{GareyJ78} problem in combinatorial optimization.
In bin packing, we are given a list $I:= (x_1, \ldots, x_n)$ of $n$ items with sizes in $(0,1]$, and the goal is to partition them into the minimum number of unit-sized bins such
that the total size of the items in each bin is at most $1$.
Unlike offline algorithms, in online algorithms, we do not have complete information about the list $I$.
In the online model, item sizes are revealed one by one: in round $i$  the item $x_i$ arrives and needs to be {\em irrevocably} assigned to a bin before the next items $(x_{i+1}, \ldots, x_n)$ are revealed.
We measure the performance of an algorithm $\mathcal A$ 
by the following quantity: 
$R_{\mathcal{A}}^{\infty} = \limsup_{m \to \infty} \left(\sup_{I: \OPT(I) = m}   \left({ \mathcal{A}(I)  } / {\OPT(I)} \right) \right)$,
where $\mathcal A(I)$ denotes the number of bins used by $\mathcal A$ to pack an input instance $I$, and $\Opt$ denotes the optimal algorithm.
If $\mathcal A$ is an offline algorithm, $R_{\mathcal{A}}^{\infty}$ is called \emph{Asymptotic Approximation Ratio (AAR)}.
On the other hand, if $\mathcal A$ is an online algorithm, $R_{\mathcal{A}}^{\infty} $ is called \emph{Competitive Ratio (CR)}.
In this paper, we mainly deal with the \emph{random-order model} \cite{Gupta020} in online algorithms. In this model, the input set of items is chosen by the adversary; however, the arrival order of the items is decided according to a permutation chosen uniformly at random from $\mathcal{S}_n$, the set of permutations of $n$ elements. 
This reshuffling of the input items often weakens the adversary and provides better performance guarantees. 
In this model, we measure the performance of an online algorithm $\mathcal A$ using the following quantity, called \emph{random-order ratio} (RR):
\begin{align*}
RR_{\mathcal{A}}^{\infty} = \limsup_{m \to \infty} \left(\sup_{I: \OPT(I) = m}     \frac{ \expec[\sigma]{\mathcal{A}(I_\sigma)}  }  {\OPT(I)} \right)
\end{align*}
Here for a given permutation $\sigma$, we define the list $I_\sigma:=\left(x_{\sigma(1)},x_{\sigma(2)},\dots,x_{\sigma(n)}\right)$ to be the list containing items in $I$ permuted according to the permutation $\sigma$,
and the expectation is taken over the uniform probability distribution wherein each permutation of $n$ items is equally likely.
Note that the random-order ratio is only concerned with the performance for the instances whose optimal value is large, that is, we only care about the asymptotic performance.

%% file: prior-work.tex

The \bestfit{} (BF) algorithm is one of the most widely-used algorithms for bin packing. \bestfit{} packs each item into the {\em fullest} bin where it fits, possibly opening a new bin if the item fits into none of the present open bins.
As was mentioned in \cite{DBLP:conf/soda/Kenyon96}: ``Best-Fit emerges as the winner among the various online algorithms: it is simple, behaves well in practice, and no algorithm is known which beats it both in the worst case and in the average uniform case". Thus, there is an extensive literature studying the behavior of \bestfit{} in various settings: asymptotic approximation \cite{ullman1971performance, DBLP:conf/stoc/GareyGU72,  johnson1974worst}, absolute approximation \cite{simchi1994new, opt-bestfit}, average-case analysis \cite{coffman1993markov}, uniform distributions \cite{coffman1993probabilistic}, etc. 

Kenyon \cite{DBLP:conf/soda/Kenyon96} first introduced the notion of random-order ratio as an alternate measure of performance for online algorithms and established that the random-order ratio of \bestfit{} is upper bounded by $3/2$ and lower bounded by 1.08. 
Kenyon also conjectured that the true random-order ratio should ``lie somewhere close to 1.15". 
Since then, both the random-order model as well as the conjecture has received significant consideration. 
As mentioned in \cite{coffman1997approximation}, this conjecture, if proven, will have implications for the offline bin packing problem as well. 
It will show that \bestfit{} (after performing a random permutation on the input list) has the best worst-case behavior among all the practical algorithms for (offline) bin packing. 
Closing the gap between the upper and lower bounds for \bestfit{} was mentioned as one of the open problems in the recent survey on {\em Random-Order Models} by Gupta and Singla \cite{Gupta020}. 

In recent years, there have been some improvements for certain special cases. 
Albers et al. \cite{albers_et_al_MFCS} proved that the random-order ratio of \bestfit{} is at most $1.25$ when all items are larger than 1/3.
They showed that, when all items are larger than 1/3, \bestfit{}  is monotone (i.e., increasing the size of one or more items can not decrease the number of bins used by the algorithm). 
This is surprising as  \bestfit{} is not monotone even in the presence of a single item of size less than 1/3 \cite{DBLP:journals/dam/Murgolo88}.
Then their analysis utilized this monotonicity property to relate bin packing with online stochastic matching. However, these properties crucially rely on the fact that at most two items can be packed in a bin, and it does not extend to the general case. 
Ayyadevara et al.~\cite{DBLP:conf/icalp/AyyadevaraD0S22} made further progress and exploited these connections to show that the random-order ratio of \bestfit{} is $1$ when all items are larger than 1/3. They also showed that the random-order ratio of \bestfit{} is $\approx 1.4941$, for the special case of 3-partition (when all the item sizes are in (1/4, 1/2]). However, their analysis breaks down in the presence of large items of size greater than 1/2.
Recently, Fischer \cite{fischer_thesis} presented a different {\em exponential-time} randomized algorithm with an RR of $(1+\eps)$.
However, for polynomial-time algorithms, the barrier of 3/2 remains unbroken in the general case. 

For the lower bound, one can generate a list of million items such that, based on a sampling of permutations, the random-order ratio empirically appears to be $\approx 1.144$ \cite{coffman1997approximation}.
The present best-known lower bound, which can be analytically determined, is 1.1 \cite{albers_et_al_MFCS}. It holds even for the i.i.d.~model (where input items come from an i.i.d.~distribution) with only two types of items.  
However, even the empirical conjectured estimate of 1.144 is still open to be proven analytically as a lower bound for \bestfit{} under random-order.

%% file: our-contributions.tex
\subsection{Our Contributions}
We improve both the upper and lower bounds of the performance of \bestfit{} in the random-order model.
\subsubsection{Upper Bound}
Our main result is breaking the barrier of 3/2 for the upper bound. 


\begin{theorem}
\label{maintheorem}
Let $\sigma:[n]\to[n]$ be a permutation chosen uniformly at random from $\mathcal{S}_n$, the set of permutations of $n$ elements, and let $I_\sigma$ denote the instance $I$ permuted according to $\sigma$. Then
\begin{align*}
\expec{\Bf(I_\sigma)}\le \left(\frac32-\eps\right){\Opt(I)}+o({\Opt(I)}),
\end{align*}
\end{theorem}
where  $\Bf(I_\sigma)$ is the number of bins that BF requires to pack $I_\sigma$ and $\eps$ is a sufficiently small constant.

Let us now briefly explain the approach in \cite{DBLP:conf/soda/Kenyon96} that was used to show that $RR_{\BF}^\infty\le 3/2$. One of the main constructs in \cite{DBLP:conf/soda/Kenyon96} is the quantity $t_\sigma$, which is the last time that $\BF$, on input $I_\sigma$, packs an item of size at most $1/3$ in a bin of load at most $1/2$.
One can show that all bins (except at most one) opened by BF to pack the first $t_\sigma$ items (i.e., $I_\sigma(1,t_\sigma)$) are filled up to the level of at least 2/3.
Thereafter, a counting argument shows that, to pack items arriving after $t_\sigma$ (i.e., $I_\sigma(t_\sigma+1, n)$), BF is within a 3/2 factor of $\Opt$. These observations result in the following two inequalities:
\begin{align}
\BF(I_\sigma(1,t_\sigma)) &\leq \frac{3}{2}  \OPT(I_\sigma(1,t_\sigma)) + 1\label{first-half}\\
\BF(I_\sigma)-\BF(I_\sigma(1,t_\sigma)) &\leq \frac{3}{2} \OPT(I_\sigma(t_\sigma + 1 ,n))  + 1\label{second-half}
\end{align}
Finally, it was shown that 
$\OPT(I_\sigma(1,u))/u$ converges  
to  $\OPT(I)/n$ for a random permutation $\sigma$. 
Combining all these facts, an upper bound of $3/2$ was achieved. 

We explain our techniques now. First, we divide the items into four categories depending on their sizes: Large ($L$), Medium ($M$), Small ($S$), and Tiny ($T$), with sizes in $(1/2, 1]$, $(1/3, 1/2]$, $(1/4, 1/3]$, and $(0, 1/4]$, respectively. \ah{  Kenyon’s \cite{DBLP:conf/soda/Kenyon96}  proof  relies on showing that Best-Fit achieves a 3/2 approximation factor separately for items appearing before $t_\sigma$ and items appearing after $t_\sigma$. Our approach is similar, but we improve the analysis to show that one of the two inequalities above can be improved further in a fruitful way. In particular, if $t_\sigma\ge n/2$, then the factor of $3/2$ in \cref{first-half} 
can be improved to $3/2-2\eps$, and if $t_\sigma<n/2$, then the factor $3/2$ in \cref{second-half} can
be improved to $3/2-2\eps$. Combining both the improved inequalities gives us Theorem \ref{maintheorem}. }

 \ah{At a high level, we do a case analysis based on $t_\sigma$ (and also consider other parameters such as the volume of tiny items and the structure of the optimal solution) and show that either a large fraction of the bins packed by BF is rather full (the load is at least 3/4) or $\BF$ performs relatively well compared to $\OPT$.
We initially obtained a factor better than 3/2 for the case where all items have size $>1/4$, and tried to apply our techniques to the general case. For example, let us suppose $t_\sigma$ is large, and consider the time segment before $t_\sigma$. If the total size of tiny items before $t_\sigma$ was large, then intuitively, Best-Fit should do well as a substantial fraction of bins have low wasted space, as tiny items can be packed efficiently.  On the other hand, if the total size of tiny items that appear before $t_\sigma$ is small, intuitively, this should be similar to the $>1/4$ case, but it is technically still difficult to account for interactions with tiny items. We thus define a construct $t'_\sigma$, which is the last time that $\BF$, on input $I_\sigma$, packs an item of size at most $1/4$ in a bin of load at most $1/2$. One can show that Best-Fit achieves a 4/3 approximation before $t'_\sigma$ as almost all bins opened before $t_\sigma'$ have load at least $3/4$,
and that tiny items do not open new bins after $t_\sigma'$. If $t_\sigma'$ is large, then we have many bins with load at least $3/4$ in BF, allowing us to beat the factor of $3/2$. On the other hand, if $t_\sigma'$ is small, our techniques from the $>1/4$ case can be applied to the relatively large interval [$t'_\sigma$,$t_\sigma$] ($t_\sigma'$ is small, $t_\sigma$ is large), allowing us to beat the factor of 3/2.
}

Now let us describe the three key ideas that we use in this work.

\textbf{Presence of a large number of `gadgets'.} 
One key contribution of our work is the usage of ‘gadgets’ in random-order arrival. For many online optimization problems, for adversarial-order arrival, the items must appear in a specific order so that the algorithm performs poorly compared to the optimal solution. However, we show that we can classify the items and then show the existence of some special gadgets or patterns that will mitigate the poor performance of the algorithm. We show that, unlike adversarial-order arrival, in random-order arrival, such patterns appear frequently, thus leading to an improved performance guarantee. Many algorithms for problems in random-order arrival classify the input items into several item classes (e.g., based on sizes), such as knapsack and GAP \cite{kesselheim2018primal, AlbersKL21}, Machine covering \cite{AlbersGJ23}, etc. Making use of frequently recurring patterns might be helpful in these problems.
Although a rudimentary form of this idea was introduced in \cite{DBLP:conf/icalp/AyyadevaraD0S22} for the special case when the input only has two types of items (medium and small), the pattern they used was restrictive and simple. For example, the items in the pattern were needed to be consecutive. Thus, their analysis cannot be extended to the case where the items in the patterns are nonconsecutive (e.g., some tiny items appear between the medium and small items) or when there are more size classes (e.g., large items) in the input. 
To circumvent this issue, we come up with more intricate gadgets---namely, $S$-triplets, fitting $ML$ triplets, and fitting $ML/SL$ triplets.
An $S$-triplet in a fixed permutation $\sigma$ is a set of three small items in $\sigma$ with only tiny items in between them. A fitting $ML$ triplet is a triplet of fitting pairs of medium and large items
(with only tiny items in between them), where a {\em fitting pair} is defined as a pair of items whose sizes
add up to at most $1$. Fitting $ML/SL$ triplets are defined in a similar way. 
Unlike in \cite{DBLP:conf/icalp/AyyadevaraD0S22}, the presence of tiny items complicates our analysis  
(See \cref{claim:ml-triplet}, \cref{claim:ml-sl-tuple-proof} in {\cref{pf:high-lm-ls-bins}}, and \cref{claim:s-triplet-good} in \cref{pf:high-mss-mms-sss-bins}). Moreover, counting the number of gadgets in the random input sequence also turns out to be harder.
For example, to count the number of fitting $ML$ triplets, we must also ensure that each $ML$ pair is fitting; see \cref{proportionality-lm}. We handle these issues with a technically involved analysis.   

\textbf{Weight functions.} 
 Another technical contribution of our work is the use of weight functions – for the first time – in the random-order model. Weight functions map item sizes to some real numbers which we refer to as weights.  Finding suitable weight functions has been helpful in bin packing and other related problems \cite{johnson1974worst,lee-lee}, as it helps us to study interactions between item types and relate optimal packing with the packing of the algorithm. 
However, none of the previous works
on bin packing under random-order arrival used this technique. The work \cite{DBLP:conf/icalp/AyyadevaraD0S22}, e.g., 
uses combinatorial techniques to analyze $\BF{}$ in the special case when all the items are either medium or small; their
techniques are difficult to extend due to the less-understood interactions between the large and tiny items.
We use weight functions to analyze $\BF$ under random-order (See Case 2 of \cref{final-case-wt-fn}).
By forgetting the actual contents of a bin and, instead, focusing on
the weight of the bin, we show that \bestfit{} `packs' more weight in a large number of bins (See, e.g., \cref{lem:high-lm-ls-bins,lem:high-mss-mms-sss-bins} for details). This leads to a better performance.

\textbf{`One good permutation suffices'.} Another idea that we use is that if there is
one ``good permutation" (i.e., satisfying certain properties), then it is possible to extract some additional information about the input and deduce  that
at least a constant fraction of the $n!$ permutations can be packed well using \bestfit{}. This idea is
the main ingredient in analyzing some bottleneck cases
(See \cref{lem:high-lm-ls-bins,lem:high-mss-mms-sss-bins}).

Now we briefly discuss the high-level proof structure of the result. 
\kvn{See \Cref{fig:organization} for an overview of the cases we consider.}
First, we consider the case when $t_\sigma>n/2$ (\case{1}). 
Then we further classify depending on the volume of tiny items among the first $t_\sigma$ items.
If it is high (\case{1.2}), then intuitively, many bins can be shown to have a load of at least $3/4$. \ah{Otherwise the volume of tiny items before $t_\sigma$ is low (\case{1.1}), and we consider cases based on the size of $t_\sigma'$.} If $t'_\sigma$ is large (\case{1.1.2}), we can again show that 
many bins have a load of at least $3/4$. Otherwise (\case{1.1.1}), we define appropriate weight functions and show the existence of many fitting $ML$/$SL$ triplets or $S$-triplets, depending on the structure of $\OPT$. This (along with the idea that `one good permutation suffices') enables establishing the presence of many ``well-packed'' bins in the packing by BF. 
In the other case, when $t_\sigma \leq n/2$ (\case{2}), we consider if the number of $LM$ bins (bins containing one $L$ and one $M$ item) in \ah{$\OPT_1'$}  is low or not.\footnote{\kvn{Please refer to the caption of \cref{fig:organization} for the definitions of $\Opt_1$, $\Opt_1'$.}} Intuitively, \ah{we can ignore the tiny items as they don't open bins after $t_\sigma$}, and the items in two $LM$ bins in $\Opt_1'$ can be suboptimally packed by $\BF$ into three bins (one $MM$ and two $L$ bins).
Thus, informally, if the number of $LM$ bins is low (\case{2.2}) in \ah{$\OPT_1'$} 
then $\BF$ does not perform too badly compared to $\Opt$. 
Otherwise, 
\ah{if $\Opt_1'$ is bounded away from $\Opt_1$ (\case{2.1.2}), then an analysis similar to \case{2.2} shows that $\BF$ does well. Finally, if $\Opt_1'$ is close to $\Opt_1$ (\case{2.1.1}), the number of $LM$ pairs is comparable to $\Opt$. Consequently, we can show that a random instance contains many fitting $ML$ triplets, implying that BF contains sufficiently many $LM$ bins}  -- showing a better performance guarantee of $\BF$.

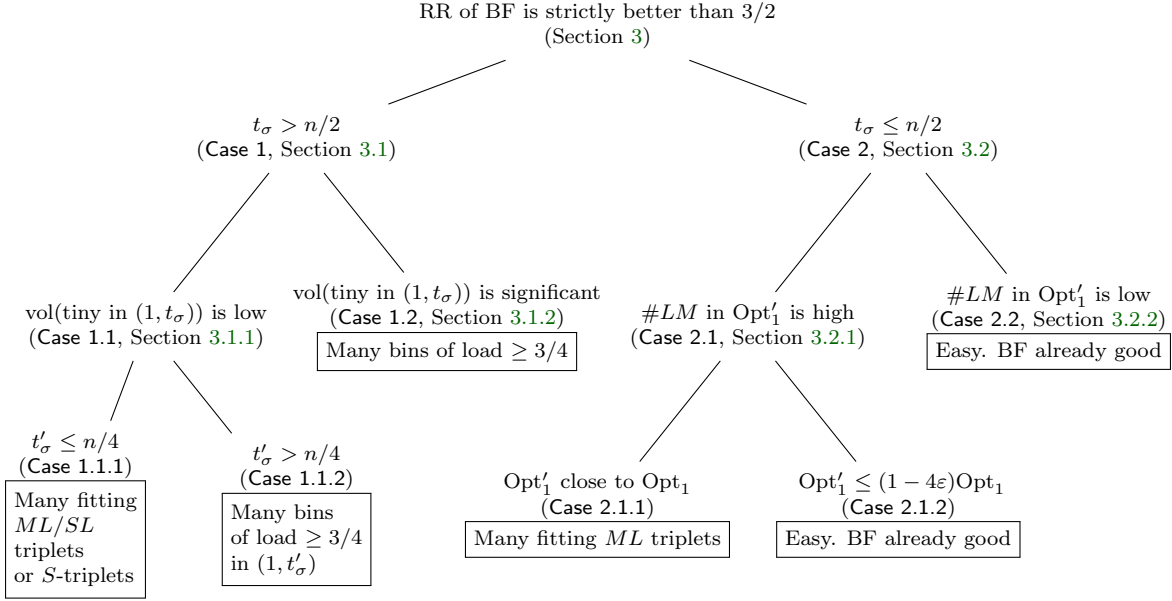
\begin{figure}
\centering
\begin{tikzpicture}
[font=\footnotesize,xscale=-1,
grow=down, 
level 1/.style={sibling distance=8cm,level distance=1.5cm},
level 2/.style={sibling distance=4cm,level distance=2.5cm}]
\node {\makecell[c]{RR of BF is strictly better than $3/2$\\(\cref{sec:general-case})}} 
child { node {\makecell[c]{$t_\sigma\le n/2$\\(\Case{2}, \cref{sec:t-sigma-small})}}
    child { node {\makecell[c]{$\#LM$ in \ah{$\Opt_1'$} is low\\(\Case{2.2}, \cref{sec:case11})\\\fbox{\makebox{Easy. BF already good}}}} }
    child { node {\makecell[c]{$\#LM$ in \ah{$\Opt_1'$} is high\\(\Case{2.1}, \cref{sec:case12})}}
        child { node {\makecell[c]{\ah{ $\Opt_1'\le(1-4\eps)\Opt_1$ }\\(\Case{2.1.2})\\\fbox{\makebox{Easy. BF already good}}}} }
        child { node {\makecell[c]{ \ah{$\Opt_1'$ close to $\Opt_1$} \\(\Case{2.1.1})\\\fbox{\makebox{Many fitting $ML$ triplets}}}} }
    }
}
child { node {\makecell[c]{$t_\sigma>n/2$\\(\Case{1}, \cref{sec:t-sigma-big})}}
    child { node {\makecell[c]{vol(tiny in $(1,t_\sigma)$) is significant\\(\Case{1.2}, \cref{sec:case21})\\\fbox{\makebox{Many bins of load $\ge3/4$}}}} }
    child { node {\makecell[c]{vol(tiny in $(1,t_\sigma)$) is low\\(\Case{1.1}, \cref{final-case-wt-fn})}}
        child { node {\makecell[c]{$t'_\sigma>n/4$\\(\Case{1.1.2})\\\fbox{\makebox{\makecell[l]{Many bins\\of load $\ge3/4$\\in $(1,t'_\sigma)$}}}}} }
        child { node [right]{\makecell[c]{$t'_\sigma\le n/4$\\(\Case{1.1.1})\\\fbox{\makebox{\makecell[l]{Many fitting \\ $ML/SL$ \\triplets\\ or $S$-triplets}}}}} }
    }
};
\end{tikzpicture}
\caption{
The overview of our case analysis.
For brevity, we write $\Opt_1$ instead of $\Opt(I_\sigma(t_\sigma+1,n))$,
$\Opt'_1$ instead of $\Opt(I'_\sigma(t_\sigma+1,n))$ (here $I'_\sigma(t_\sigma+1,n)$ denotes the list $I_\sigma(t_\sigma+1,n)$ after removing the small and tiny items), and $\#LM$ instead of ``number of $LM$-bins''.
}
\label{fig:organization}
\end{figure}
\subsubsection{Lower Bound}
We also make progress on the lower bound, arriving at the mentioned empirical estimate of 1.144 in \cite{coffman1997bin} and almost matching the conjectured ratio by Kenyon \cite{DBLP:conf/soda/Kenyon96}.
\begin{theorem}\label{aarratio}
For online bin packing under the random-order model, the random-order ratio of Best-Fit is greater than $1.144$, i.e., $RR_{\Bf}^{\infty} > 1.144$.
\end{theorem}


The main idea in the previous works on lower bounds \cite{DBLP:conf/soda/Kenyon96, albers_et_al_MFCS} is to instead consider the i.i.d.\ model to show a lower bound for \bestfit{} under random-order arrival. In the i.i.d.\ model, the input is a sequence of items drawn from a common probability distribution. This model is much easier to analyze compared to the random-order model, as the arrival of an item does not depend on the preceding input sequence. The key fact used is that the random-order ratio for any bin packing algorithm is lower bounded by the corresponding ratio in the i.i.d.\ model. 

The asymptotic performance of \bestfit{} in the i.i.d.\ model can be found exactly by computing the stationary probabilities of an underlying Markov chain. Essentially, the states are different open bin configurations, and the transitions correspond to different item arrivals. Estimating the performance of \bestfit{} thus comes down to counting the expected number of transitions where \bestfit{} opens a new bin. \cref{table:iid} summarizes our lower bounds and describes the best item list that we found and corresponding probabilities for up to seven types of items.
\begin{table}
\label{table:iid}
\begin{center}
\begin{tabular}{||L | L | L | L||} 
 \hline
 \text{\#Item }  & \text{Item Sizes}  & \text{Probabilities}  & \text{Lower}  \\ 
\text{ Types}  &   &   & \text{Bound}  \\
 \hline\hline
 2 & [1/4, 1/3] & [0.594, 0.406] & 1.1037 \\ 
 \hline
  3 & [0.25, 0.31, 0.38 ] & [0.466, 0.356, 0.178] & 1.1182 \\ 
 \hline
4 & [0.25, 0.26, 0.32, 0.44] & [0.454, 0.234, 0.195, 0.117] & 1.1334\\
 \hline
  5 & [0.25, 0.26, 0.3, 0.4, 0.46] & [0.43, 0.204, 0.176, 0.088, 0.102] & 1.1378 \\ 
 \hline
  6 & [0.245, 0.26, 0.27, 0.3, 0.38, 0.46]  & [0.35, 0.116, 0.194, 0.162, 0.081, 0.097] & 1.1419 \\ 
 \hline
  7 & [0.245, 0.25, 0.26, 0.27, 0.3, 0.38, 0.46] & [0.26, 0.13, 0.13, 0.17, 0.15, 0.075, 0.085] & 1.1440 \\ 
 \hline
\end{tabular}
\end{center}
\caption{\kvn{Different distributions and the performance of \bestfit{} when items are sampled from these distributions.}}
\end{table}

As the number of item types increased, we saw diminishing returns and an exponential increase in the size of the Markov state space and running time. While the initial example with two items discussed in \cite{albers_et_al_MFCS} has nine states in total, our example with seven items has $357$ states, making manual analysis infeasible, due to which we analyze the Markov chain with the help of a computer-assisted proof.\footnote{The code is available at: \url{https://github.com/bestfitroa/BinPackROA}.} One key difference in our example is that we use items that are not of the type $1/m$ for integral $m$, making analysis of the optimal algorithm in the i.i.d.\ model more complicated, as it often uses hybrid (consisting of multiple item types) bins. 
Thus, even though there are many possible open bin configurations, only a few of them are perfectly packed, causing \bestfit{} to  pack a large fraction of bins suboptimally. At the same time, increasing the number of item types, intuitively, increases the average load of a closed bin, resulting in less wasted space by \bestfit{}. These two conflicting factors consequently give diminishing returns with an increasing number of item types. See Section \ref{sec:Lower Bounds for Best Fit}  for a  detailed discussion on the lower bound. 

%% file: related-work.tex
\subsection{Related Work}
For offline bin packing, the present best polynomial-time approximation algorithm returns a solution using $\OPT+O(\log \OPT)$ bins \cite{DBLP:conf/soda/HobergR17}.  
However, bin packing can be solved {\em exactly} in polynomial-time  \cite{GoemansR20} when we have a constant number of item types.  
For online bin packing (under adversarial-order arrival), the present best upper and lower bounds on the CR are 1.57829 \cite{BaloghBDEL18} and 1.54278 \cite{BaloghBDEL19}, respectively.
For the i.i.d.~model, Rhee and Talagrand \cite{rhee1993lineB} exhibited an algorithm that, w.h.p., achieves a packing in  $\OPT+O(\sqrt{n}\log^{3/4}n)$ bins for any  distribution on $(0,1]$.
Ayyadevara et al.~\cite{DBLP:conf/icalp/AyyadevaraD0S22} achieved a near-optimal performance guarantee for the i.i.d.~model. For any arbitrary unknown distribution, they gave a  meta-algorithm that takes an $\alpha$-{asymptotic} approximation algorithm as input and provides a polynomial-time $(\alpha+\eps)$-competitive algorithm. 
 



Johnson et al.~\cite{johnson1974worst} studied several heuristics for bin packing such as  Best-Fit (BF), First-Fit (FF), Best-Fit-Decreasing (BFD), First-Fit-Decreasing (FFD) and showed their (asymptotic) approximation guarantees to be $17/10, 17/10, 11/9, 11/9$, respectively.
After a sequence of improvements \cite{DBLP:conf/stoc/GareyGU72, garey1976resource, simchi1994new}, the tight performance guarantee of \bestfit{} 
 (for adversarial-order) was shown to be $\lfloor 1.7 \cdot \OPT \rfloor$ \cite{opt-bestfit}.
 Another $O(n \log n)$ time algorithm Modified-First-Fit-Decreasing (MFFD) \cite{mffd} attains an AAR of $71/60\approx 1.1834$ and has the current best provable performance guarantee among all the simple and fast algorithms for offline bin packing.  
 Among all practically popular algorithms, Best-Fit (on a random permutation of the input) is conjectured to beat MFFD in terms of worst-case performance guarantee \cite{coffman1997bin}.

Note that the asymptotic polynomial-time approximation schemes (APTAS) for bin packing \cite{VegaL81, KarmarkarK82, DBLP:conf/soda/HobergR17} are theoretical in nature and seldom used in practice. 
We refer the readers to the surveys \cite{coffman2013bin, ChristensenKPT17} for a comprehensive treatment of the existing literature on bin packing and its variants.



Starting from the prototypical secretary problem \cite{freeman1983secretary}, the random-order model has been studied extensively for many optimization problems: from computational geometry \cite{clarkson1993four} to packing integer programs \cite{kesselheim2018primal}, from online matching \cite{MahdianY11} to facility location \cite{meyerson2001online}, from set cover \cite{GuptaSCR} to knapsack \cite{AlbersKL21}.
See the recent survey \cite{Gupta020} for details on random-order models.


%% file: organization.tex
\subsection{Organization of the Paper}
In \cref{sec:notation}, we discuss notations  
and introduce weight functions.
Then, in \cref{sec:general-case}, we prove the main result of the paper---RR of \bestfit{} is strictly better than $3/2$.
Our analysis is divided into multiple cases, and this organization is shown in \cref{fig:organization}.
Due to space limitations, many of the intermediate claims and lemmas have been delegated to the appendix.
Then, in \cref{sec:Lower Bounds for Best Fit}, we establish the lower bound of $1.144$ on the random-order ratio of \bestfit{}.
Finally, \cref{sec:conclusion} concludes with some remarks and open problems.

%% file: notation.tex
\section{Preliminaries}
\label{sec:notation}
We denote the size of an item $x_i$ by $s(x_i)$. Any item is categorized into one of the four different categories as follows:
(i) \emph{Large (\L):} if its size lies in the range $(1/2,1]$,
(ii) \emph{Medium (\M):} if its size lies in the range $(1/3,1/2]$,
(iii) \emph{Small (\S):} if its size lies in the range $(1/4,1/3]$,
(iv) \emph{Tiny (\T):} if its size lies in the range $(0,1/4]$. For the input sequence $I$ and two timestamps/indices
$t_1,t_2\in[n]$ such that $t_1\le t_2$, we denote by $I(t_1,t_2)$ the subsequence that arrived from time
$t_1$ to $t_2$ (including $t_1,t_2$). 

The load (or volume) 
of bin $B$ is given by  $\vol(B) := \sum_{x \in B} s(x)$. \ah{Similarly, the volume of a set of items $T$ is given by $\vol(T) := \sum_{x \in T} s(x)$.}
Observe that a bin can contain at most one large item, at most two medium items, and at most three small items.
We often indicate a bin by the items
of type $L/M/S$ it contains, e.g., an {$\LS$}-bin is a bin that
contains a large item and a small item, an {$MMS$-bin} contains two medium items
and a small item, etc. Note that we do not indicate the tiny items that a bin might contain. For any $k\in[3]$,
we say that a bin $B$ is a $k$-bin if the number of items of type $\L,\M,$ or $\S$ in it is $k$
(again, we do not indicate the tiny items, if any).
If no future items can be packed into a bin, we say it is {\em closed}, otherwise, it is {\em open}. 


We say an event occurs with high probability if its probability approaches $1$ as \ah{$\OPT(I)$} tends to infinity. For example,
an event that occurs with probability $1-1/\log(\OPT(I))$ is said to occur with high probability, or \whp{} in short.

\subsection{Weight Functions}
The concept of weight functions has been used extensively in the analysis of packing algorithms \cite{johnson1974worst,lee-lee}. It gives us a method to
upper bound the number of bins used by the algorithm that we want to analyze and lower bound the optimal solution. A weight function $W:[0,1]\to \mathbb R^+$  maps the item sizes to some rounded values, \ah{ and we generally round up the item size. 
\kvn{For brevity, we just write $W(x)$ instead of $W(s(x))$ to denote the weight of an item $x$.} The weight of a bin $B$ is given by $W(B) = \sum_{x \in B} W(x)$.}
The following lemma has been used in all the prior works which rely on weight function based analyses 
(see, e.g., \cite{johnson1974worst}). 

\ah{
\begin{lemma}[Folklore] \label{weight_function}
    Consider any given instance of items $I$ packed using an algorithm $\mathcal A$ and a weight function $W$.
    Suppose the bins $B$ in the packing $\mathcal A(I)$ satisfy the following lower bound on their total weight
    $$\sum_{B\in \mathcal A(I) }W(B)\ge\alpha_1 \mathcal A(I) - O(1)$$ 
    for some constant $\alpha_1$. Intuitively, this means that the average weight of the bins is at least $\alpha_1$, ignoring lower order terms. Further, suppose that for any set of items $C$ such that $\sum_{x\in C}s(x)\le 1$, it holds that $\sum_{x\in C}W(x)\le\alpha_2$, where $\alpha_2$ is a constant. Then we have the bound
    $$\mathcal A(I) \leq \frac{\alpha_2}{\alpha_1} \Opt(I) + O(1).$$
\end{lemma}
\begin{proof}
We compute the total weight of the items in two ways.
\begin{align*} \alpha_1 \mathcal A(I) - O(1) & \leq \sum_{B\in \mathcal A(I) }W(B)  = \sum_{B\in \mathcal \OPT(I) }W(B) \le\alpha_2 \OPT(I)
\end{align*}
which implies that
\begin{align*}
    \mathcal A(I) \leq \frac{\alpha_2}{\alpha_1} \Opt(I) + O(1)
\end{align*}
which gives us the desired bound.
\end{proof}
}

%% file: general-case.tex
\section{Upper Bound for the Random-Order Ratio of Best-Fit}
\label{sec:general-case}
In this section, we prove our main result (Theorem~\ref{maintheorem}): the RR of \bestfit{} is strictly less than $3/2$.

 Let $\sigma$ denote a  permutation of $[n]$ selected uniformly at random.
We assume $\OPT(I) \rightarrow \infty$.
Consider a run of the \bestfit{} algorithm on $I_\sigma$. Let $t_\sigma$ be the last time an item of size $ \leq 1/3$ 
(i.e., a small or tiny item) was added to
a bin of load at most $1/2$. We will break up the input instance $I$ into two parts: before and after $t_\sigma$, and analyze each time segment separately.






Kenyon \cite{DBLP:conf/soda/Kenyon96} showed that the number of bins in $\OPT(I_\sigma(1,t))$ is close to $\frac{t}{n}\OPT(I)$ with high probability. In fact, the following weaker version suffices for our result.
We give a full proof in \cref{pf:kenyonlemma}.
\begin{lemma}[\cite{DBLP:conf/soda/Kenyon96}]\label{kenyon}
Fix any two positive constants $\alpha, \delta < \frac{1}{2}$. Then, for large enough $\OPT(I)$ and all $t$ such that $\alpha n \leq t \leq (1 - \alpha) n$,   we have that with high probability:
\begin{align*}
\frac{t}{n} ( 1- \delta)\OPT(I)  &\leq \OPT(I_\sigma(1,t)) \leq \frac{t}{n}(1+ \delta) \OPT(I), \\
\left(\frac{n-t}{n}\right) ( 1- \delta)\OPT(I)  &\leq \OPT(I_\sigma(t+1,n)) \leq \left(\frac{n-t}{n}\right)(1+ \delta) \OPT(I). 
\end{align*}
\end{lemma}

\noindent We note the following, which also was proved by Kenyon \cite{DBLP:conf/soda/Kenyon96}.

\begin{lemma}[\cite{DBLP:conf/soda/Kenyon96}] \label{kenyon2}
Consider the \bestfit{} packing of $I_\sigma$.
Then, every bin in this packing, with at most one exception, opened before or at time $t_\sigma$ has a load greater than $2/3$. Moreover, we have two inequalities:
$$\BF(I_\sigma(1,t_\sigma)) \leq \frac{3}{2}  \OPT(I_\sigma(1,t_\sigma)) + 1,$$
$$\BF(I_\sigma)-\BF(I_\sigma(1,t_\sigma)) \leq \frac{3}{2} \OPT(I_\sigma(t_\sigma + 1 ,n))  + 1.$$
\end{lemma}

Before we proceed, we argue that if the number of large and medium items is at most a constant, then we are already done. \ah{Intuitively, this is because the instance contains mostly small and tiny items, so the \bestfit{} packing has low wasted space.}
The detailed proof can be found in \cref{pf:opttrunc}.
\begin{lemma}\label{opttrunc}
If the total number of large and medium items in the instance $I$ is at most $k$, where $k$ is some fixed constant, then
$\BF(I_\sigma) \leq \frac43 \OPT(I) + O(1)$ for any permutation $\sigma$.
\end{lemma}
\begin{proofsketch}
Observe that the number of bins in $\BF(I_\sigma)$ that contain a large or a medium item is at most a constant,
and thus, these bins comprise only $o(1)$ fraction of the entire packing $\BF(I_\sigma)$. The remaining bins 
only consist of small and tiny items. 
It is easy to see that these bins (except one) will have a load greater than $2/3$. 
However, with a more careful analysis, we show that, in the \bestfit{} packing of any set of tiny and small items,
almost all the bins have load greater than $3/4$.
\end{proofsketch}
Thus, we may assume that $\OPT(I') \to \infty$, 
where $I'$ consists of the list without tiny and small items (i.e., only contains large  and medium items), as otherwise we are done by \cref{opttrunc}.
To show that \bestfit{} actually achieves a random-order ratio strictly better than $3/2$,
we consider many cases where each case holds with a positive, constant probability.
In many of these cases,
we use \cref{kenyon}, using the fact that a high probability event conditioned on another event that occurs with at least constant probability, still occurs with high probability. More formally, we have the following. 
\begin{proposition}
\label{prop:probability-fact}
Consider any two events $X,Y$ in a probability space. If $\prob{X}=1-o(1)$ and $\prob{Y}\ge c$ where $c$ is a constant,
then $\prob{X|Y}=1-o(1)$.
\end{proposition}
Due to the above proposition, even if we consider only a constant fraction of all the $n!$ permutations,
\cref{kenyon} can be used. The proof of the proposition can be found in \cref{sec:other-omitted-proofs}.

\kvn{\textbf{Global Parameters:}}
In the following subsections, we will use three constant parameters \kvn{$\epsilon,\zeta,\delta$} extensively. Parameter $\eps$ is a constant
whose value is around $10^{-9}$; we will show that the random-order ratio of \bestfit{} is at most $(3/2-\eps)$. \footnote{We did not try to optimize the constants for the sake of simplicity of exposition. However, we do not expect a significant improvement just through meticulous optimization.}
Parameter $\zeta$ is a constant that we will use to analyze different cases.
For example, we first consider the case where $\prob{\eone:=(t_\sigma> n/2)}\ge\zeta$.
Since $\zeta$ is a constant, we can use the high probability guarantee provided by \cref{kenyon},
owing to \cref{prop:probability-fact}.
The closer to zero we choose $\zeta$ to be, the better our analysis. 
Finally, $\delta$ is a very small constant compared to both $\zeta$ and $\eps$; it will
be used to apply \cref{kenyon}.
\input{case-1}
\input{case-2}
\input{putting-it-together}

%% file: case-1.tex
\subsection{\texorpdfstring{$t_\sigma$}{t-sigma} is Big with Constant Probability}
\label{sec:t-sigma-big}
In this subsection, we consider \case{1}, where the event $t_\sigma> n/2$ occurs with constant probability, i.e.,
for a constant $\zeta$, 
\[
    \prob{\eone} \geq \zeta\:\:\text{where event}\:\: \eone\coloneqq \Big( t_\sigma > n/2 \Big).
\]
In this case, we will show that with probability at least $1 - \zeta$ (conditioned on $\eone$), 
the number of new bins opened by \bestfit{} up to time $t_\sigma$ is at most $(3/2 - 2\epsilon)\OPT(I_\sigma(1,t_\sigma))$.
\begin{lemma}\label{lem:case2}
Suppose the event $\eone$ occurs with a positive, constant probability.
Conditioned on $\eone$, 
we have that with probability at least $1 - \zeta$, 
\[ 
    \BF(I_\sigma(1, t_\sigma)) \leq \left(\frac{3}{2} - 2\epsilon\right)\OPT(I_\sigma(1,t_\sigma)) + o(\OPT(I)). 
\]
\end{lemma}



\noindent Depending on the volume of tiny items before $t_\sigma$, we consider two cases below, \ah{and show that as long as the considered case occurs with some constant probability, \cref{lem:case2} holds conditionally.}

\input{case-11}

\input{case-12}
We are now ready to prove \cref{lem:case2}, ending the analysis of \case{1}.
\begin{deferredproof}{\cref{lem:case2}}
Let $G$ be the event that \bestfit{} performs strictly better than $3/2$ in the time segment $(1,t_\sigma)$, i.e.,
\begin{align*}
G\coloneqq \left(\BF(I_\sigma(1,t_\sigma))  \leq \left(\frac{3}{2} - 2 \epsilon\right)\OPT(I_\sigma(1,t_\sigma)) + o(\OPT(I))\right)
\end{align*}
Define
\begin{align*}
\poneone\coloneqq\prob{\vol(T(1,t_\sigma))<12 \epsilon\vol(I_\sigma(1,t_\sigma))|\eone}\quad \text{and}\quad \ponetwo\coloneqq\prob{\vol(T(1,t_\sigma))\ge12 \epsilon\vol(I_\sigma(1,t_\sigma))|\eone}
\end{align*}
We need to show that $\prob{G|\eone}\ge 1-\zeta$ to prove the lemma.
\begin{align*} 
\prob{G|\eone}&=\prob{G|\eoneone}\prob{\eoneone|\eone}+\prob{G|\eonetwo}\prob{\eonetwo|\eone}\\
            &=\prob{G|\eoneone}\prob{\vol(T(1,t_\sigma))<12 \epsilon\vol(I_\sigma(1,t_\sigma))|\eone}\\
            &\:\:+\prob{G|\eonetwo}\prob{\vol(T(1,t_\sigma))\ge12 \epsilon\vol(I_\sigma(1,t_\sigma))|\eone}\\
            &= \prob{G|\eoneone}\poneone+\prob{G|\eonetwo}\ponetwo
\end{align*}
By \cref{lem:subcase2a}, we know that for any permutation $\sigma$ satisfying the event $\eonetwo$, the event $G$ occurs.
Hence $\prob{G|\eonetwo}$ is always $1$.
If $\poneone\ge\zeta^2$, then by \cref{lem:subcase2b}, we have that $\prob{G|\eone}=(1-\zeta)\poneone+\ponetwo\ge 1-\zeta$
since $\poneone+\ponetwo=1$ and $\poneone\le1$. On the other hand, if $\poneone<\zeta^2$, then we have
$\prob{G|\eone}\ge\ponetwo\ge1-\zeta^2\ge1-\zeta$. Hence, the lemma stands proved.
\end{deferredproof}

%% file: case-11.tex
\subsubsection{Volume of Tiny Items Before \texorpdfstring{$t_\sigma$}{t-sigma} is Low}
\label{final-case-wt-fn}
Here, we will consider \case{1.1}, where with constant probability, the fraction of the volume of tiny items in the time
segment $(1,t_\sigma)$ is small compared to the total volume in the segment $(1,t_\sigma)$. Let $T(1,t_\sigma)$ denote the set of tiny items in the sequence $I_\sigma(1,t_\sigma)$.
Formally, we assume the following condition.
\begin{align*}
\prob{\vol(T(1,t_\sigma)) < 12\epsilon \vol(I_\sigma(1,t_\sigma)) \bigg \vert t_\sigma > n/2} \ge \zeta^2.
\end{align*}
Note that, this implies
\begin{align*}
\prob{\eoneone}\ge \zeta^3 \:\:\text{where event}\:\: \eoneone\coloneqq \left (\vol(T(1,t_\sigma)) < 12\epsilon \vol(I_\sigma(1,t_\sigma))\bigwedge t_\sigma>n/2 \right ).
\end{align*}
In this case, we wish to show that \bestfit{} has a performance ratio of strictly better than $3/2$ in the time segment $(1,t_\sigma)$. More formally,
we will show the following lemma.
\begin{restatable}{lemma}{lemsubcasetwob}
\label{lem:subcase2b}
Suppose the event $\eoneone$ occurs with constant probability. Then, conditioning on $\eoneone$, we have that, with probability
at least $1-\zeta$, 
\begin{align*}
\BF(I_\sigma(1,t_\sigma)) \leq \left(\frac{3}{2} - 2\epsilon\right)\OPT(I_\sigma(1,t_\sigma)) + o(\OPT(I)).
\end{align*}
\end{restatable}
We will define a construct similar to $t_\sigma$. 
Let $t_\sigma'$ be the last time a tiny item (size $\leq 1/4)$ was added to a bin of load at most $1/2$. Note that $t_\sigma' \leq t_\sigma$, necessarily. \ah{Similar to Kenyon's proof for $t_\sigma$, one can show that the number of bins used by \bestfit{} before $t_\sigma'$ is within a factor of $4/3$ of the optimal packing $\OPT(I_\sigma(1,t_\sigma'))$. Thus, intuitively, if $t'_\sigma$ is large, we are already done as $4/3 < 3/2$. To deal with the case when $t_\sigma'$ is small, we will use weight functions.}\\

We once again consider two cases, not necessarily disjoint, that cover all possibilities
depending on the value of $t'_\sigma$. We will then combine the results to prove \cref{lem:subcase2b}.

\input{case-111}

\input{case-112}
Using \cref{lem:high-t-sigma-prime,lem:low-t-sigma-prime}, we can now prove \cref{lem:subcase2b}, which we restate below for convenience. 

\lemsubcasetwob*
\begin{proof}
From the lemma statement, we assume that the event $\eoneone$ occurs with a positive, constant probability.
Let $\poneoneone\coloneqq\prob{t'_\sigma\le n/4|\eoneone}$ and $\poneonetwo\coloneqq\prob{t'_\sigma> n/4|\eoneone}$.
\kvn{Note that since $\eoneoneone=(t'_\sigma\le n/4)\land \eoneone$, it follows that $\poneoneone=\prob{\eoneoneone|\eoneone}$.
Similarly, $\poneonetwo=\prob{\eoneonetwo|\eoneone}$.}
Let $G$ be the event that \bestfit{} performs strictly better than $3/2$ in the time segment $(1,t_\sigma)$, i.e.,
\begin{align*}
G\coloneqq \left(\BF(I_\sigma(1,t_\sigma))  \leq \left(\frac{3}{2} - 2 \epsilon\right)\OPT(I_\sigma(1,t_\sigma)) + o(\OPT(I))\right)
\end{align*}
To establish the lemma, we would like to calculate $\prob{G|\eoneone}$.
\begin{align*}
\prob{G|\eoneone}=\prob{G|\eoneoneone}\poneoneone+\prob{G|\eoneonetwo}\poneonetwo
\end{align*}
If $\poneonetwo\le \zeta^2$, then $\poneoneone\ge 1-\zeta^2$. Therefore, by \cref{lem:low-t-sigma-prime}, we have that
$\prob{G|\eoneoneone}=1-o(1)$. Hence, $\prob{G|\eoneone}\ge(1-\zeta^2)(1-o(1))\ge1-\zeta$.

On the other hand, if $\poneoneone\le \zeta^2$, then $\poneonetwo\ge 1-\zeta^2$. Then, by \cref{lem:high-t-sigma-prime}, we have that
$\prob{G|\eoneonetwo}=1-o(1)$. Hence, $\prob{G|\eoneone}\ge(1-\zeta^2)(1-o(1))\ge1-\zeta$.

Finally, if both $\poneoneone>\zeta^2$ and $\poneonetwo>\zeta^2$, then
$\prob{G|\eoneonetwo}=1-o(1)$ and $\prob{G|\eoneoneone}=1-o(1)$ by \cref{lem:low-t-sigma-prime,lem:high-t-sigma-prime}, respectively.
Hence, observing that $\poneoneone+\poneonetwo=1$, we have $\prob{G|\eoneone}=(1-o(1))\poneoneone+(1-o(1))\poneonetwo=1-o(1)$.
Overall, we have $\prob{G|\eoneone}\ge 1-\zeta$ if the event $\eoneone$ occurs with constant probability. Hence, the lemma, stands proved.
\end{proof}

%% file: case-111.tex
\noindent \Case{1.1.1:} $\prob{t_\sigma' \leq \frac{n}{4} \bigg \vert \vol(T(1,t_\sigma)) < 12\epsilon \vol(I_\sigma(1,t_\sigma))\bigwedge t_\sigma > n/2} \geq \zeta^2$.

\noindent Note that this implies 
\[
    \prob{\eoneoneone} \geq \zeta^{2 + 1 +2} = \zeta^5,\:\:\text{ where event}\:\: \eoneoneone\coloneqq \left(t'_\sigma \leq n/4\bigwedge t_\sigma > n/2\bigwedge \vol(T(1,t_\sigma)) < 12\epsilon \vol(I_\sigma(1,t_\sigma)) \right ).
\] 
    Conditioned on $\eoneoneone$,
    we will show that, with high probability,
    \[
        \BF(I_\sigma(1,t_\sigma)) \leq \left(\frac{3}{2} - 2 \epsilon\right)\OPT(I_\sigma(1,t_\sigma)) + o(\OPT(I)).
    \]
We will use a weight function approach. Let $W(x)$ denote the weight of an item $x$ according to a weight function $W$. We will set weights as follows

\[
 W(x)= 
  \begin{cases}
      1  \qquad \quad \text{ if } \frac{1}{2} <s(x) \leq 1  \, (x  \text{ is a large item} ), \\
     0.5 \qquad \, \text{ if } \frac{1}{3} <s(x) \leq \frac{1}{2}  \, (x  \text{ is a medium item} ),  \\
 0.5 \qquad \,  \text{ if } \frac{1}{4} <s(x) \leq \frac{1}{3}   \, (x  \text{ is a small item} ),  \\
    3 \cdot s(x) \,\, \hspace{0.02cm} \text{ if }\, 0 \leq s(x) \leq \frac{1}{4}  \, (x  \text{ is a tiny item} ).  \\
  \end{cases}
\]

\ah{Note that we always round up, that is, $W(x) \geq s(x)$.} For any set of items $J$, let the weight of the set $J$ be defined as $W(J)=\sum_{x\in J}W(x)$. 
The weights are chosen in a way such that in the packing $\Bf(I_\sigma(1,t_\sigma))$, 
every bin (with at most one exception) will have a weight of at least $1$. This is stated in the following claim.
The proof can be found in \cref{sec:other-omitted-proofs}.
\begin{claim}
\label{wt-bf-atleast-one}
Consider the packing $\Bf(I_\sigma(1,t_\sigma))$. With the possible exception of one bin, all the bins will have a weight of at least one.
\end{claim}
As the consequence of the above claim, we get the following claim, whose proof is again deferred to \cref{sec:other-omitted-proofs}.
\begin{claim}
\label{prop:min-wt-bf}
We have that $\BF(I_\sigma(1,t_\sigma))\le W(I_\sigma(1,t_\sigma))+1$.
\end{claim}

For the input list $I$, let $\tilde I$ be the list $I$ with tiny items deleted from it.
Similarly, let $\tilde I_\sigma(1,t_\sigma)$ be the sequence $I_\sigma(1,t_\sigma)$ with tiny items deleted from it. Since the volume of $T(1,t_\sigma)$ is very low, intuitively, the quantities $\Opt(\tilde I(1,t_\sigma))$ and $W(\tilde I(1,t_\sigma))$
must be very close to the quantities $\Opt(I(1,t_\sigma))$ and $W(I(1,t_\sigma))$, respectively.
The following two claims are based on this intuition; the proofs can be found in \cref{sec:other-omitted-proofs}.
\begin{claim}
\label{opti'lowerbound}
For any $\sigma$ satisfying $E_{111}$, we have
$\OPT(\tilde I_\sigma(1 , t_\sigma))\ge \left(1-16 \epsilon\right)\Opt(I_\sigma(1,t_\sigma))-1$.
\end{claim}
\begin{claim}
\label{weightoptnotiny}
For any $\sigma$ satisfying $E_{111}$, we have
$\displaystyle\frac{W(I_\sigma(1, t_\sigma))}{\OPT(I_\sigma(1, t_\sigma))}\leq \frac{W(\tilde I_\sigma(1, t_\sigma) )}{\OPT(\tilde I_\sigma(1, t_\sigma) )} \left (\frac{1 + 24  \epsilon}{1 - 12 \epsilon} \right  )$.
\end{claim}
Note that,
since the only possible bin configurations are $L,M,S, LM,LS,MM,SS, MS, MSS,MMS,SSS$, we can verify that
any bin in $\OPT(\tilde I_\sigma(1, t_\sigma) )$ has weight at most $3/2$. However, there can be at most $O(1)$ 
many bins in $\OPT(\tilde I_\sigma(1, t_\sigma) )$ of type $M$, $S$, $SS$. \ah{(For example, if there were $3$ bins of type $SS$, we could have repacked them into $2$ bins of type $SSS$ to get a better solution.)}

\noindent  We thus divide all but 
$O(1)$ many bins into two types: 
\begin{align*}
& \typeone{}: L , MS, MM  \\ & \typetwo{}: LM, LS, MMS, MSS, SSS.
\end{align*}
Note that if a bin $B$ is of $\typeone$, it satisfies $W(B)=1$; and if it is of $\typetwo$,
it satisfies $W(B)=3/2$. Let $\beta(\sigma)$ denote the fraction of \typeone{} bins in $\OPT(\tilde I_\sigma(1, t_\sigma) )$.
The next claim shows an upper bound on $\BF(I_\sigma(1,t_\sigma))$ in terms of $\Opt(I_\sigma(1,t_\sigma))$. \ah{We basically analyze the instance $\tilde I_\sigma(1, t_\sigma)$ using weight functions and then obtain bounds for $\BF(I_\sigma(1, t_\sigma))$ using \cref{prop:min-wt-bf} and \cref{weightoptnotiny}.}
A detailed proof can be found in \cref{sec:other-omitted-proofs}.
\begin{claim}
\label{weightfn_opt}
Conditioned on $E_{111}$, we have
\begin{align*}\BF(I_\sigma(1,t_\sigma))\le \left(\frac{3}{2} - \frac{\beta(\sigma)}{2}\right)\left(\frac{1 +24 \epsilon}{1 - 12 \epsilon} \right  ) \OPT(I_\sigma(1, t_\sigma))+O(1).
\end{align*}
\end{claim}
In words,
\cref{weightfn_opt} tells us that if there is a good fraction of $\typeone$ bins in $\OPT(\tilde I_\sigma(1, t_\sigma) )$,
then \bestfit{} packs well, i.e., has random-order ratio of strictly less than $3/2$.

Now let us give a high-level idea of the rest of the analysis for this case.  If $\beta(\sigma)$ is a constant, we obtain
from \cref{weightfn_opt} that the random-order ratio of \bestfit{} is strictly better than $3/2$. Hence, for now
assume that the packing $\Opt(\tilde I_\sigma(1,t_\sigma))$ is dominated by bins of $\typetwo$.
We further divide the bins of $\typetwo$ into those containing large items and those not containing large items.
If the number of bins of
type $LM/LS$ in $\Opt(\tilde I_\sigma(1,t_\sigma))$ is significant, then we show that,
in a random sequence $ I_\sigma(1,t_\sigma)$, there exists a good number of gadgets---we call them `fitting $ML/SL$ triplets'---that
result in many bins of weight $3/2$.
On the other hand, if the number of bins of type $MSS/MMS/SSS$ in $\Opt(\tilde I_\sigma(1,t_\sigma))$ is significant, then
we show that there exists a large number of $S$-triplets in the random input sequence and these result in the formation of many bins of weight $3/2$ in the
\bestfit{} packing.

Let $r_1(\sigma)$ denote the fraction of bins of type $LM/LS$ in $\OPT(\tilde I_\sigma(1,t_\sigma))$
and let $r_2(\sigma)$ denote the fraction of bins of type $MMS/MSS/SSS$ in $\OPT(\tilde I_\sigma(1,t_\sigma))$.
Note that, by their respective definitions, $\beta(\sigma)+r_1(\sigma)+r_2(\sigma)=1-o(1)$.
Using \cref{opti'lowerbound}, we have that with high probability,
\begin{align} 
\OPT(\tilde I_\sigma(1,t_\sigma))&\ge (1-16\eps)\OPT(I_\sigma(1,t_\sigma))-1 \nonumber\\
&\ge(1-16\eps)\frac12\left(1-\delta\right)\OPT(I)-1 \nonumber \quad \quad \text{(using \cref{kenyon} since $t_\sigma>n/2$)}\\
&\ge\frac{1-17\eps}{2}\Opt(\tilde I) \label{require-prop}
\end{align}
as we have chosen $\delta$ to be very small compared to $\epsilon$, and $\Opt(I)\ge\Opt(\tilde I)$.

Suppose, for all permutations $\sigma$ satisfying the high probability event given by \cref{require-prop},
we have that $\beta(\sigma)\ge10^{-4}$. Then, by \cref{weightfn_opt}, and using
$\left(\frac{3}{2} - \frac{\beta(\sigma)}{2}\right)\left(\frac{1 +24 \epsilon}{1 - 12 \epsilon} \right  )<\left(\frac32-2\eps\right)$, we obtain that
with high probability,
\begin{align}
\BF(I_\sigma(1,t_\sigma))\le\left(\frac32-2\eps\right) \OPT(I_\sigma(1, t_\sigma))+O(1)
\end{align}

Now, suppose that there exists a permutation $\sigma^*$ satisfying the high probability event given by \cref{require-prop}
such that $\beta(\sigma^*)< 10^{-4}$. Then, since $\beta(\sigma^*)+r_1(\sigma^*)+r_2(\sigma^*)=1-o(1)$, it must be case that
either $r_1(\sigma^*)\ge 0.91$ or $r_2(\sigma^*)\ge 0.089$ since $0.91+0.089+10^{-4}<1$.\footnote{The values $0.91,0.089$ have been obtained by optimizing $r_1(\sigma^*),r_2(\sigma^*)$, respectively, over the range $(0,1)$.} 
The former case implies that there are a large number of disjoint item pairs
of type $LM$ or $LS$ that ``fit'' together. The latter case implies that there are a large number of disjoint ``fitting'' triplets of
type $MMS$ or $MSS$ or $SSS$. The next lemmas show that in both the cases, \bestfit{} creates a large number of bins of weight $3/2$.

\begin{lemma}
\label{lem:high-lm-ls-bins}
Suppose $r_1\coloneqq r_1(\sigma^*)\ge0.91$, where $\sigma^*$ satisfies \cref{require-prop}. Consider a random permutation $\sigma$ (satisfying $\eoneoneone$).
Then, \whp{}, the number of bins of weight at least $3/2$ in the packing $\BF(I_\sigma(1,t_\sigma))$ is at least
\begin{align*}
\frac{1-16\eps}{384}\frac{(r_1-17r_1\eps)^6}{(6-r_1+17r_1\eps)^5}\Opt(I) - o(\Opt(I)).
\end{align*}
\end{lemma}
\begin{proofsketch}
The fact that $r_1(\sigma^*)$ is at least a constant implies that in the packing $\Opt(I_\sigma(1,t_\sigma))$,
there exist a good number of fitting pairs of the form $ML/SL$. Using concentration bounds, we show that,
in a random sequence $I_\sigma(1,t_\sigma)$, many disjoint consecutive triplets of pairs of type $ML/SL$ will be present with high probability. Moreover, for each of these triplets, there will be a unique corresponding bin of weight $3/2$ in the packing $\BF(I_\sigma(1,t_\sigma))$.
\end{proofsketch}
\begin{lemma}
\label{lem:high-mss-mms-sss-bins}
Suppose $r_2\coloneqq r_2(\sigma^*)\ge0.089$, where $\sigma^*$ satisfies \cref{require-prop}. Consider a random permutation $\sigma$ (satisfying $\eoneoneone$).
Then, \whp{}, the number of bins of weight at least $3/2$ in the packing $\BF(I_\sigma(1,t_\sigma))$ is at least
\begin{align*}
\frac{1 - 16 \epsilon}{24} \left(\frac{r_2-17r_2 \epsilon}{4 + r_2 -17r_2 \epsilon}\right)^3 {}\Opt(I)-o(\Opt(I)).
\end{align*}
\end{lemma}
\begin{proofsketch}
Since $r_2(\sigma^*)$ is at least a constant, we obtain that in the packing $\Opt(I_\sigma(1,t_\sigma))$,
there exist a good number of $MMS/MSS/SSS$ bins. In turn, this implies that there are a good number of small items. 
Using concentration bounds, we show that in a random sequence $I_\sigma(1,t_\sigma)$,  many disjoint consecutive $S$-triplets will be present with high probability.
Finally, we show that for every \ah{two disjoint consecutive $S$-triplets} in $I_\sigma(1,t_\sigma)$, at least one bin of weight $\ge3/2$ will be formed (with $O(1)$ many exceptions).
\end{proofsketch}

The detailed proof of \cref{lem:high-lm-ls-bins} can be found in \cref{pf:high-lm-ls-bins} and that of \cref{lem:high-mss-mms-sss-bins}
can be found in \cref{pf:high-mss-mms-sss-bins}.

To summarize, the analysis when the event $\eoneoneone$ occurs boils down to three cases.
If every permutation $\sigma$ satisfies $\beta(\sigma)\ge10^{-4}$, then \cref{weightfn_opt} ensures that \bestfit{} performs well.
Else, for one of the permutations $\sigma^*$, we have that $r_1(\sigma^*)\ge0.91$ or $r_2(\sigma^*)\ge0.089$.
\cref{lem:high-lm-ls-bins} and \cref{lem:high-mss-mms-sss-bins}, respectively, show that the existence of
$\sigma^*$ is enough to ensure that, for almost all the permutations (satisfying $\eoneoneone$), \bestfit{} creates a good number of bins of weight $3/2$.

By combining \cref{lem:high-lm-ls-bins} and \cref{lem:high-mss-mms-sss-bins}, we have the following lemma, \ah{ showing that the bound in \cref{lem:subcase2b} holds with high probability conditioned on $E_{111}$, as long as $E_{111}$ occurs with at least a constant probability.}
Its proof is delegated to \cref{pf:final-case-analysis}.
\begin{lemma}
\label{lem:low-t-sigma-prime}
Define event $\eoneoneone\coloneqq\left(t_\sigma' \leq \frac{n}{4} \bigwedge \vol(T(1,t_\sigma)) < 12\epsilon \vol(I_\sigma(1,t_\sigma)) \bigwedge t_\sigma > n/2\right)$.
Further suppose that $\eoneoneone$ occurs with constant probability.
Then
\begin{align*}
\prob{\BF(I_\sigma(1,t_\sigma))\le\left(\frac32-2\eps\right)\Opt(I_\sigma(1,t_\sigma))+o(\Opt(I))\bigg\vert \eoneoneone} = 1- o(1).
\end{align*}
\end{lemma}
That ends the analysis of \case{1.1.1}. \\

\noindent Next, we consider the case when $t'_\sigma>n/4$ (that is, relatively large) with at least constant probability.

%% file: case-112.tex
\noindent \Case{1.1.2:} $\prob{t_\sigma' > \frac{n}{4} \bigg \vert \vol(T(1,t_\sigma)) < 12\epsilon \vol(I_\sigma(1,t_\sigma))\bigwedge t_\sigma > n/2} \geq \zeta^2 $.

\noindent Note that this implies
\[
    \prob{\eoneonetwo} \geq \zeta^{2 + 1 +2} = \zeta^5\:\:\text{where event}\:\:\eoneonetwo\coloneqq \left(t'_\sigma > n/4\bigwedge t_\sigma > n/2\bigwedge \vol(T(1,t_\sigma)) < 12\epsilon \vol(I_\sigma(1,t_\sigma))\right).
\]



Recall that $t_\sigma'$ is the last time an item of size $\leq 1/4$, say $a^*$, was added into a bin with load
at most $1/2$. Since we are using Best-Fit, at any point of time, there can't be two bins with
load at most $1/2$. Hence, the only bin that has load at most $1/2$ is the bin into which $a^*$
was packed. All the other bins must have load greater than $3/4$, since \bestfit{}
would have packed $a^*$ into one of those bins otherwise.

Hence, in the \bestfit{} packing, the bins opened before time $t'_\sigma$ have a load greater than $3/4$.
And we know that the bins opened before $t_\sigma$ have a load greater than $2/3$.
Hence, if we look at the time segment $(1,t_\sigma)$, and recalling that the event $\eoneonetwo$ implies $t'_\sigma>n/4$,
we can prove that, \whp, many bins in $\BF(I_\sigma(1,t_\sigma))$ have load strictly greater than $2/3$.
We thus obtain the following lemma. Its proof is given in \cref{app:case-112}.
\begin{lemma}
\label{lem:high-t-sigma-prime}
Let the event $\eoneonetwo\coloneqq \left(t_\sigma' > \frac{n}{4} \bigwedge \vol(T(1,t_\sigma)) < 12\epsilon \vol(I_\sigma(1,t_\sigma))\bigwedge t_\sigma > n/2\right)$. 
Further suppose that $\eoneonetwo$ occurs with constant probability.
Then
\begin{align*} 
\prob{\BF(I_\sigma(1,t_\sigma))\leq \left(\frac 32-\frac{1}{36}\right) \OPT(I_\sigma(1,t_\sigma))  + \frac{17}{8}\bigg\vert \eoneonetwo}=1-o(1).
\end{align*}
\end{lemma}

%% file: case-12.tex
\subsubsection{Volume of Tiny Items Before \texorpdfstring{$t_\sigma$}{t-sigma} is Significant}
\label{sec:case21}
Here, we will consider \case{1.2} where tiny items before $t_\sigma$ contribute at least a constant fraction of the volume of all the
items before $t_\sigma$. 
More formally, we assume that
\begin{align*}
\prob{\vol(T(1,t_\sigma)) \ge 12\epsilon \vol(I_\sigma(1,t_\sigma)) \bigg \vert t_\sigma > n/2} \ge \zeta^2.
\end{align*}
Note that this implies
\begin{align*}
\prob{E_{12}} \geq \zeta^3, \:\:\text{where event}\:\: E_{12} \coloneqq \left ( \vol(T(1,t_\sigma)) \geq 12\epsilon \vol(I_\sigma(1,t_\sigma)) \bigwedge t_\sigma > n/2 \right).
\end{align*}
\begin{lemma}\label{lem:subcase2a}
Consider any arbitrary permutation $\sigma$ satisfying $\vol(T(1,t_\sigma)) \geq 12\epsilon \vol(I_\sigma(1,t_\sigma))$.
Then
\begin{align*}
\BF(I_\sigma(1,t_\sigma)) \leq \left(\frac{3}{2} - 2\epsilon\right)\OPT(I_\sigma(1,t_\sigma)) + 2.
\end{align*}
\end{lemma}
\begin{proof}
We will use volume arguments to prove the lemma. In more detail, we know from \cref{kenyon2} that all bins opened before 
$t_\sigma$, with at most one exception, have a load of at least
$2/3$. What we will show is that a constant fraction of these bins, in fact, have a load greater than $3/4$. Combining these two arguments
gives us the lemma.\\

\noindent Towards this, we will state and use the following claim. Its proof can be found in \cref{proof-largetinyvolume}.
\begin{claim}
\label{claim:largetinyvolume}
Suppose $\vol(T(1,t_\sigma)) \geq 12\epsilon \vol(I_\sigma(1,t_\sigma))$. Then at least
$\floor{12 \epsilon \vol(I_\sigma(1,t_\sigma))}$ many number of bins in $\BF(I_\sigma(1,t_\sigma))$ have a load greater than $3/4$.
\end{claim}
Now, all bins up to time $t_\sigma$ (with at most one exception) are at least $2/3$ full, and, from \Cref{claim:largetinyvolume}, 
at least $\floor{12\epsilon \vol(I_\sigma(1,t_\sigma))} \geq 12\epsilon \vol(I_\sigma(1,t_\sigma)) -1$ 
many bins in $\BF(I_\sigma(1,t_\sigma))$ are at least $3/4$ full. So, if $\mathcal B_1$ denotes the bins 
that are at least $3/4$ full, and $\mathcal B_2$ denotes the bins that are at least $2/3$ full but not $3/4$ full, we have
\begin{align*} \BF(I_\sigma(1,t_\sigma)) &\leq   |\mathcal B_1| + |\mathcal B_2| + 1 \\
& \leq \frac{4}{3}\vol(\mathcal B_1) + \frac{3}{2}\Big(\vol(I_\sigma(1,t_\sigma)) - \vol(\mathcal B_1)\Big) + 1\\
&\leq \frac{3}{2}\left(\vol(I_\sigma(1,t_\sigma)) - \frac{\vol(\mathcal B_1)}{6}\right)  + 1\\
&\leq \frac{3}{2}\left(\vol(I_\sigma(1,t_\sigma)) - \frac{1}{6}   \frac{3}{4}12\epsilon \vol(I_\sigma(1,t_\sigma)) \right)+ 2 \quad \text{(using \cref{claim:largetinyvolume} and definition of $\mathcal B_1$) }\\
&\leq \left(\frac{3}{2} - 2 \epsilon \right)\vol(I_\sigma(1,t_\sigma))  + 2\\
&\leq \left(\frac{3}{2} - 2 \epsilon \right)\OPT(I_\sigma(1,t_\sigma))  + 2.
\end{align*}
This proves the lemma and ends the analysis of \case{1.2}. 
\end{proof}

%% file: case-2.tex
\subsection{\texorpdfstring{$t_\sigma$}{t-sigma} is Small with Constant Probability}
\label{sec:t-sigma-small}
In this section, we consider \case{2}, where the event $t_\sigma\le n/2$ occurs with constant probability.
More formally, we assume that
\begin{align*}
\prob{\etwo}\ge\zeta\:\:\text{where event}\:\: \etwo\coloneqq \Big ( t_\sigma\le n/2 \Big)
\end{align*}
We will show that the number of new bins opened by \bestfit{} after time $t_\sigma$ is at most $ (3/2 - 2\epsilon)  \OPT(I_\sigma(t_\sigma + 1,n)) + o(\opt(I))$ with good probability.

\begin{lemma}
\label{lem:case1}
Suppose the event $\etwo\coloneqq\left(t_\sigma<n/2\right)$ occurs with at least constant probability.
Conditioning on $\etwo$,
we have that with probability at least $1 - \zeta$,
\begin{align*}
N_\sigma\coloneqq\BF(I_\sigma)-\BF(I_\sigma(1,t_\sigma))\le \left(\frac32-2\eps\right)\Opt(I_\sigma(t_\sigma+1,n))+o(\Opt(I))
\end{align*}
\end{lemma}

Before we proceed, we need some notation. Let $I'$ (respectively, $I'_\sigma$) denote the instance after removing small and tiny items from $I$ (respectively, $I_\sigma$). Similarly, we obtain the list $I'_\sigma(t_\sigma+1,n)$ by removing the small and tiny items from $I_\sigma(t_\sigma+1,n)$.
We use $N_\sigma$ to denote the number of bins opened by \bestfit{} after time $t_\sigma$ to pack $I_\sigma$, i.e., $N_\sigma=\Bf(I_\sigma)-\Bf(I_\sigma(1,t_\sigma))$.

Consider a permutation $\sigma$ for which $t_\sigma \leq n/2$. Let $\widehat{\ell}$ be the number of large 
items in $I_\sigma(t_\sigma + 1 ,n )$  and $\widehat m$ be the number of medium items in $I_\sigma(t_\sigma + 1 ,n )$.
Let $\widehat b$ be 
the number of $LM$ bins in $\OPT(I'_\sigma(t_\sigma + 1 ,n ))$.
Note that $\widehat{\ell},\widehat m,\widehat b$ are functions of the permutation $\sigma$.
We must have 
\begin{align}
    N_\sigma \leq \widehat{\ell} + \frac{\widehat m}{2} + 1.\label{eq:n-sigma}
\end{align}
 This is because, after $t_\sigma$, every bin must be opened by a medium or large item. Moreover, if a medium item
 opens a new bin, then the second item that is packed in this bin must be either large or medium.
Also, we have 
\begin{align}
    \OPT(I'_\sigma(t_\sigma + 1 ,n ))  = \ceil {\widehat{\ell}-\widehat b + \frac{\widehat{m}-\widehat{b}}{2} + \widehat{b}} = \ceil{\widehat{\ell} + \frac{\widehat{m}-\widehat{b}}{2} }\label{kdjfjkdkjfdkj} 
\end{align}
This is because, in the packing $\OPT(I'_\sigma(t_\sigma + 1 ,n ))$,
among the $\hat{\ell}$ large items, $\hat b$ of them are in $LM$-bins.
Therefore, the remaining $\hat\ell-\hat b$ large items must have been packed alone.
Similarly, among the $\hat{m}$ medium items, $\hat b$ of them are in $LM$-bins.
Therefore, each of the remaining $\hat\ell-\hat b$ medium items (with one possible exception)
must have been packed with another medium item.

Notice how the number of bins opened by \bestfit{} after $t_\sigma$ (given by \cref{eq:n-sigma}) and the optimal number of bins for  $I'_\sigma(t_\sigma+1,n)$ 
(given by \cref{kdjfjkdkjfdkj}) are similar in expression except for $\widehat{b}$, \ah{the number of $LM$ bins in $\OPT(I'_\sigma(t_\sigma + 1 ,n ))$.} Hence, depending on whether $\widehat{b}$ is big or small \ah{relative to $\OPT(I'_\sigma(t_\sigma + 1 ,n ))$}, we have two cases
\begin{itemize}
    \item $\prob{\widehat{b} \geq (1- 4 \epsilon)\OPT(I'_\sigma(t_\sigma + 1 ,n )) \bigg \vert t_\sigma \leq n/2 } \geq \zeta $
    \item $\prob{\widehat{b} < (1- 4 \epsilon)\OPT(I'_\sigma(t_\sigma + 1 ,n )) \bigg \vert t_\sigma \leq n/2 } \geq  \zeta $
\end{itemize}
\input{case-21}
\input{case-22}
We are now ready to end the analysis of \case{2}.
We combine \cref{lem:case11} and \cref{lem:case12} to show that in the case when $t_\sigma$ is small with
constant probability, \bestfit{} performs strictly better than $3/2$ in the time segment $(t_\sigma+1,n)$.
\begin{deferredproof}{\cref{lem:case1}}
Define the event
\begin{align*}
H\coloneqq \left(N_\sigma \leq \left(\frac{3}{2} - 2\epsilon\right)\OPT(I_\sigma(t_\sigma + 1 , n)) + o(\OPT(I))\right)
\end{align*}
Let
$$ \ptwoone = \prob{\hat b\ge(1-4\eps)\Opt(I'_\sigma(t_\sigma+1,n)) \bigg \vert \etwo}$$ 
$$ \ptwotwo = \prob{\hat b<(1-4\eps)\Opt(I'_\sigma(t_\sigma+1,n)) \bigg \vert \etwo}$$ 
and note that $\ptwoone+\ptwotwo=1$. 
\kvn{Also, note that since 
\[
\etwoone=\left(\hat b\ge(1-4\eps)\Opt(I'_\sigma(t_\sigma+1,n))\right)\land \etwo,
\] it follows that $\ptwoone=\prob{\etwoone|\etwo}$.
Similarly, $\ptwotwo=\prob{\etwotwo|\etwo}$.}
We have
\begin{align*}
\prob{H|\etwo}&=\prob{H|\etwoone}\prob{\etwoone|\etwo}+\prob{H|\etwotwo}\prob{\etwotwo|\etwo}\\
                &=\prob{H|\etwoone}\ptwoone+\prob{H|\etwotwo}\ptwotwo
\end{align*}
By \cref{lem:case11}, we have that $\prob{H|\etwotwo}=1$.
Hence, if $\ptwoone > \zeta$, then by \cref{lem:case12} (where we conditioned on the event $\etwoone$), we must have 
$\prob{H|\etwo}=\prob{H|\etwoone}\ptwoone+\prob{H|\etwotwo}\ptwotwo\ge (1-\zeta)\ptwoone+\ptwotwo\ge1-\zeta$.

On the other hand, if $\ptwoone<\zeta$, then $\ptwotwo>1-\zeta$. So, by \cref{lem:case11}, we have
$\prob{H|\etwo}\ge \prob{H|\etwotwo}(1-\zeta)= 1-\zeta$.
Thus, \cref{lem:case1} stands proved.
\end{deferredproof}

%% file: case-21.tex
\subsubsection{\texorpdfstring{$\hat b$}{b-hat} is Big with Constant Probability}
\label{sec:case12}
Here, we will consider \case{2.1}, where we assume that
\begin{align*}
\prob{\widehat{b} \geq (1- 4 \epsilon)\OPT(I'_\sigma(t_\sigma + 1 ,n )) \bigg \vert \etwo} \geq \zeta
\end{align*}
Since we assumed that $\prob{\etwo}\ge \zeta$, we have that the event
\begin{align*}
\etwoone\coloneqq \left ( \widehat{b} \geq (1- 4 \epsilon)\OPT(I'_\sigma(t_\sigma + 1 ,n )) \bigwedge t_\sigma<n/2 \right )
\end{align*}
occurs with probability at least $\zeta^2$, which is a constant.\\

\noindent Depending on how $\Opt(I'_\sigma(t_\sigma+1,n))$ compares to $\Opt(I_\sigma(t_\sigma+1,n))$ we have two cases. \ah{The high level idea is that if $\Opt(I'_\sigma(t_\sigma+1,n))$ is comparable to $\Opt(I_\sigma(t_\sigma+1,n))$, which is comparable to $\Opt(I_\sigma)$ by \cref{kenyon} when $t_\sigma \leq n/2$, then we are able to ensure a large number of `gadgets' occurs in a random instance after $t_\sigma$,  allowing us to beat the factor of $3/2$. On the other hand, if $\Opt(I'_\sigma(t_\sigma+1,n))$ is relatively small compared to $\Opt(I_\sigma(t_\sigma+1,n))$, a more refined analysis similar to the proof of \cref{kenyon2} gives us the desired bound. \\
}

\noindent \Case{2.1.1:} $\prob{\Opt(I'_\sigma(t_\sigma+1,n))\ge(1-4\eps)\Opt(I_\sigma(t_\sigma+1,n))\bigg\vert \etwoone}\ge\zeta^2$\\
Let $\etwooneone\coloneqq \left ( \Opt(I'_\sigma(t_\sigma+1,n))\ge(1-4\eps)\Opt(I_\sigma(t_\sigma+1,n))\bigwedge \etwoone \right )$.
Note that $\etwooneone$ occurs with a probability at least $\zeta^4$, which is a small, but positive constant. \ah{ In this case, we will show that the bound in \cref{lem:case1} on the number of bins opened by \bestfit{} after $t_\sigma$ holds with high probability (conditioned on $E_{211}$), that is
$$ \prob{\BF(I_\sigma)-\BF(I_\sigma(1,t_\sigma))\le \left(\frac32-2\eps\right)\Opt(I_\sigma(t_\sigma+1,n))+o(\Opt(I)) \bigg\vert E_{211}} \geq 1 - o(1) $$
}
Now we give a brief intuition for the analysis in this case. Since $\hat b$ denotes
the number of $LM$ bins in the optimal packing of $I'_\sigma(t_\sigma+1,n)$, and the event $\etwooneone$
ensures a lower bound on $\hat b$, there must be a large number of
$L,M$ items in $I$. For a moment, forget about the small and tiny items
as $\Opt(I_\sigma(t_\sigma+1,n))$ and $\Opt(I'_\sigma(t_\sigma+1,n))$ are very close.
\bestfit{} performs badly when large items are packed alone, i.e., without pairing with medium items (if at all they can be paired).
However, in the random-order model, we show that in the \bestfit{} packing, a significant number of large items pair with medium items.
To show this, we first prove that if a contiguous substring of type $MLMLML$ appears in $I'_\sigma(t_\sigma+1,n)$,
this will for sure create at least one $LM$ bin.
Finally, we show that there will be a significant number of substrings of type $MLMLML$ in $I'_\sigma(t_\sigma+1,n)$ using the randomness of $\sigma$
and concentration inequalities.

Now we proceed to formalize this intuition. First, we derive a more concrete lower bound on $\hat b$. As we have conditioned on $\etwooneone$,
we have that $\Opt(I'_\sigma(t_\sigma+1,n))\ge (1-4\eps)\Opt(I_\sigma(t_\sigma+1,n))$ and $\hat b\ge(1-4\eps)\Opt(I'_\sigma(t_\sigma+1,n))$.
Hence, we have
\begin{align}
\hat b&\ge (1-4\eps)^2\Opt(I_\sigma(t_\sigma+1,n))\nonumber\\
    &\ge (1-8\eps)\Opt(I_\sigma(n/2+1,n))\nonumber  \quad \text{(Since event $\etwooneone$ implies $\etwo$, i.e., $t_\sigma<n/2$)}\\
    &\ge (1-8\eps)\frac12(1-\delta)\Opt(I)\nonumber\quad \text{(w.h.p, using \cref{kenyon} with $t=n/2$)}\\
    &\ge \frac12(1-10\eps)\Opt(I)\nonumber\quad \text{(as $\delta$ is chosen to be very small compared to $\eps$.)}\\
    &\ge \left(\frac12-5\eps\right)\Opt(I')\label{eq:b-hat-lb}
\end{align}
For the penultimate inequality above, we used the fact that $\delta$ is very small compared to $\eps$,
and for the last inequality, we used the fact that $I'\subseteq I$.
This establishes a lower bound on $\hat b$ in terms of $\Opt(I')$.

Now, we formally define what fitting $ML$ triplets are, and show that they are good for the performance of \bestfit{}.
We say an $ML$ pair is \textit{fitting} if they both fit in one bin, i.e., 
their sizes add up to at most $1$. A sextuplet of items $(m_1,\ell_1,m_2,\ell_2,m_3,\ell_3)$ in a sequence of items $J$
is said to be a \emph{fitting $ML$ triplet} if all of the below conditions are satisfied.
\begin{itemize}
    \item For each $i\in[3]$, $m_i$ is medium and $\ell_i$ is large.
    \item For each $i\in[3]$, the $ML$ pair $(m_i,\ell_i)$ is fitting.
    \item $m_1\ell_1m_2\ell_2m_3\ell_3$ forms a substring in the sequence $J'$, where $J'$ is the sequence
    obtained after removing the small and tiny items in the sequence $J$.
\end{itemize}

\begin{claim}
\label{claim:ml-triplet}
Consider a fitting $ML$ triplet $(m_1,\ell_1,m_2,\ell_2,m_3,\ell_3)$ 
in the sequence $I_\sigma(t_\sigma+1,n)$. Then, at least 
one of the items in this $ML$ triplet will take part in creation of an $LM$ bin.
\end{claim}
\begin{proof}
    After time $t_\sigma$, note that only medium or large items can open a new bin. Moreover, 
    if a medium item opens a new bin, the next item that is packed into that bin must be either a large or medium item. 
    Thus, for each $i\in[3]$, if the medium item $m_i$ opens a new bin, $\ell_i$ must be packed with it as no intermediate item 
    can be packed on top of $m_i$ or can open a new bin, and $m_i$ did not fit into any existing bin when it arrived. If $m_i$ is packed with some
    existing large item, we are still good. Otherwise, $m_i$ is packed into a bin that does not have any large item. This bin must 
    have had load $\leq 2/3$ before $m_i$ was packed into it.

In any packing by \bestfit{}, there can be at most $2$ bins that have no large items and load at most $2/3$ at any point of time.
We defer the proof of this statement to the appendix; see \cref{bestfit23} in \cref{helper-claims}.
So, $m_i$ cannot be packed into a bin with no large item for all three of $i=1,2,3$. Thus, at least one $LM$ bin will be created.
\end{proof}

We will now use the below proportionality result, that states that the number of these fitting $ML$
triplets that occur in a time interval is proportional to the length of the interval.
Its proof is given in the appendix (see \cref{sec:proportionality}).







\begin{claim}\label{proportionality-lm}
For some constant $u>0$, let $d' \geq u \cdot \Opt(I')$ be the maximum number of disjoint fitting $ML$ pairs in $I$.
Let $n_1,n_2$ be integers such that $1\le n_1\le n_2\le n$ and $n_2-n_1=\Theta(n)$.
We have that, with high probability, the number of fitting $ML$ triplets in the sequence $I_\sigma(n_1+1,n_2)$ is at least
\begin{align*}
    \frac{u^5}{1536}\left(\frac{n_2-n_1}{n}\right) d'  - o(d')
\end{align*}
\end{claim}
Owing to \cref{eq:b-hat-lb}, we can use the above claim with $d'=\widehat b$, $u=(1/2-5\eps)$, $n_1=n/2$, and $n_2=n$.
We then obtain that the number of fitting $ML$ triplets appearing after $t_\sigma$ 
is at least
\begin{align*}
\frac{(1/2-5\eps)^5}{1536}\left(\frac12\right) \widehat b  - o(\widehat b)  &\geq \frac{(1 - 10\epsilon)^5}{1536 \cdot 64} \widehat b - o(\widehat b)\\ 
            & \ge\frac{(1 - 10\epsilon)^5(1-4\eps)}{1536 \cdot 64}\Opt(I'_\sigma(t_\sigma + 1 ,n)) - o(\Opt(I))\\
            &\qquad\qquad\qquad\qquad\qquad\text{(by event $\etwoone$ and since $\hat b\le\Opt(I)$)}\\ 
            & \geq \frac{1 -54 \epsilon}{10^5} \Opt(I'_\sigma(t_\sigma + 1 ,n))-o(\Opt(I))
\end{align*}
with high probability. So, with high probability (conditioned on $E_{211}$), \bestfit{} creates at least 
\begin{align}
\widetilde b \geq \frac{1 - 54 \epsilon}{10^5} \OPT(I'_\sigma(t_\sigma+1,n))-o(\Opt(I))\label{hgjhkjfgk}
\end{align}
many $LM$ bins after $t_\sigma$, as each such fitting $ML$ triplet creates a new $LM$ bin.
We also have the following claim whose proof can be found in \cref{sec:other-omitted-proofs}.
\begin{claim}
\label{claim:opt-prime-lb}
We have $\displaystyle\Opt(I'_\sigma(t_\sigma+1,n))\ge \frac{2\hat\ell+\hat m}{3}-\frac13$.
\end{claim}

Moreover, after $t_\sigma$, the tiny or small items cannot be packed into bins of load $\leq \frac{1}{2}$.
Hence, they can only be packed into a bin of type $L,LM, MM$. Hence, the number of bins opened by \bestfit{} after $t_\sigma$ satisfies
\begin{align}
N_\sigma &\leq \widetilde b + (\widehat{\ell} - \widetilde b) + \frac{(\widehat m-\widetilde b)}{2} + 1\nonumber \\
& \leq \widehat{\ell} + \frac{\widehat m}{2}- \frac{\widetilde b}{2}  + 1 \nonumber\\
& \leq \frac{3}{2} \OPT(I'_\sigma (t_\sigma + 1, n)) - \left( \frac{1 - 54 \epsilon}{2\cdot 10^5} \right)\OPT(I'_\sigma (t_\sigma + 1, n)) + o(\opt(I))\tag{using \cref{claim:opt-prime-lb,hgjhkjfgk}}\\
& \leq \left(\frac{3}{2} - \frac{1}{10^6}+ \frac{27\epsilon}{10^5}\right) \OPT(I'_\sigma(t_\sigma + 1 ,n ))+o(\opt(I))\nonumber\\
& \leq \left(\frac{3}{2} - 2\epsilon\right) \OPT(I_\sigma(t_\sigma + 1 ,n ))+o(\opt(I))\label{subcase1b}
\end{align}
with high probability (conditioned on $E_{211}$), where we used that $ 10^{-6}\ge3\eps$. \\

\noindent \Case{2.1.2:} $\prob{\Opt(I'_\sigma(t_\sigma+1,n))<(1-4\eps)\Opt(I_\sigma(t_\sigma+1,n))\bigg\vert \etwoone}\ge\zeta^2$\\
We define the event $\etwoonetwo\coloneqq \left ( \Opt(I'_\sigma(t_\sigma+1,n))<(1-4\eps)\Opt(I_\sigma(t_\sigma+1,n))\bigwedge \etwoone \right )$.
In this case, we will show that the bound in \cref{lem:case1} \ah{always holds} (conditioned on $E_{212}$). 
Since the event $\etwoonetwo$ implies that $\hat b\ge (1-4\eps)\Opt(I'_\sigma(t_\sigma+1,n))$, we have the following string of inequalities.
\begin{align*}
\hat b&\ge(1-4\eps)\Opt(I'_\sigma(t_\sigma+1,n))
    \ge (1-4\eps)\left(\hat{\ell}+\frac{\hat m-\hat b}{2}\right)\quad \text{(from \cref{kdjfjkdkjfdkj})}
\end{align*}
Rearranging terms and using \cref{eq:n-sigma}, \ah{ we obtain that the number of bins opened by \bestfit{} after $t_\sigma$ satisfies }
\begin{align}
N_\sigma&\le \hat{\ell}+\frac{\hat m}{2}+1\nonumber\\
        &\le \left( \frac{3/2-2\eps}{1-4\eps} \right)\hat b+1\nonumber\\
        &\le \left( \frac{3/2-2\eps}{1-4\eps}\right)\Opt(I'_\sigma(t_\sigma+1,n))+1\quad \text{(as $\hat b$ is the number of $LM$ bins in $\Opt(I'_\sigma(t_\sigma+1,n))$)}\nonumber\\
        &\le\left(\frac32-2\eps\right)\Opt(I_\sigma(t_\sigma+1,n))+1\label{djklflgjklkdglkg}
\end{align}
where the last inequality follows as we have conditioned on $\etwoonetwo$.


We combine the analyses of \casess{2.1.1, 2.1.2} to complete the analysis of the case when the event $\etwoone$ occurs,
thereby obtaining the following lemma.
\begin{lemma}
\label{lem:case12}
Suppose the event $\etwoone\coloneqq\left(\widehat{b} \geq (1- 4 \epsilon)\OPT(I'_\sigma(t_\sigma + 1 ,n )) \bigwedge t_\sigma<n/2\right)$ occurs with a constant probability. Then, we have that,
\begin{align*}
\prob{N_\sigma\le\left(\frac32-2\eps\right)\OPT(I_\sigma(t_\sigma + 1 ,n ))+o(\OPT(I))\bigg\vert \etwoone}\ge1-\zeta
\end{align*}
\end{lemma}
\begin{proof}
Define the event
\begin{align*}
H\coloneqq \left(N_\sigma \leq \left(\frac{3}{2} - 2\epsilon\right)\OPT(I_\sigma(t_\sigma + 1 , n)) + o(\OPT(I))\right)
\end{align*}
Let 
$$ \ptwooneone \coloneqq \prob{\OPT(I'_\sigma(t_\sigma + 1 ,n )) \ge (1- 4 \epsilon)\OPT(I_\sigma(t_\sigma + 1 ,n )) \bigg \vert \etwoone },$$ 
$$ \ptwoonetwo \coloneqq \prob{\OPT(I'_\sigma(t_\sigma + 1 ,n )) < (1- 4 \epsilon)\OPT(I_\sigma(t_\sigma + 1 ,n )) \bigg \vert \etwoone }$$ 
and note that $\ptwooneone + \ptwoonetwo = 1$.
\kvn{Also, note that since 
\[
\etwooneone=\Big(\OPT(I'_\sigma(t_\sigma + 1 ,n )) \ge (1- 4 \epsilon)\OPT(I_\sigma(t_\sigma + 1 ,n ))\Big)\land \etwoone,
\] it follows that $\ptwooneone=\prob{\etwooneone|\etwoone}$.
Similarly, $\ptwoonetwo=\prob{\etwoonetwo|\etwoone}$.}

\noindent Since $\etwooneone\lor \etwoonetwo=\etwoone$ and since $\etwooneone,\etwoonetwo$ are disjoint, we obtain that
\begin{align*}
\prob{H|\etwoone}&=\prob{H|\etwooneone}\prob{\etwooneone|\etwoone}+\prob{H|\etwoonetwo}\prob{\etwoonetwo|\etwoone}\\
               &=\prob{H|\etwooneone}\ptwooneone+\prob{H|\etwoonetwo}\ptwoonetwo
\end{align*}
If $\ptwoonetwo\le\zeta^2$, then $\ptwooneone\ge1-\zeta^2$. Hence, by \cref{subcase1b} (\case{2.1.1}), we have that $\prob{H|\etwooneone}=1-o(1)$.
Hence, $\prob{H|\etwoone}\ge(1-\zeta^2)(1-o(1))\ge 1-\zeta$.

On the other hand, if $\ptwooneone\le\zeta^2$, then $\ptwoonetwo\ge1-\zeta^2$. Hence, by \cref{djklflgjklkdglkg} (\case{2.1.2}), we have that
$\prob{H|\etwoonetwo}=1$. Hence, $\prob{H|\etwoone}\ge 1-\zeta^2$. 

Finally, if $\ptwooneone>\zeta^2$ and $\ptwoonetwo>\zeta^2$,
then both \cref{subcase1b} and \cref{djklflgjklkdglkg} apply. Hence, $\prob{H|\etwoone}=\ptwoonetwo+\ptwooneone(1-o(1))=1-o(1)$.
To conclude, we have $\prob{H|\etwoone}\ge1-\zeta$. 
\end{proof}

%% file: case-22.tex
\subsubsection{\texorpdfstring{$\hat b$}{b-hat} is Small with Constant Probability}
\label{sec:case11}
Here, we consider \case{2.2}, where we assume that
\begin{align*}
\prob{\widehat{b} \leq (1- 4 \epsilon)\OPT(I'_\sigma(t_\sigma + 1 ,n )) \bigg \vert \etwo} \geq  \zeta
\end{align*}
\kvn{In this case, we condition on the following event.}
\begin{align*}
\etwotwo\coloneqq \left ( \widehat{b} \leq (1- 4 \epsilon)\OPT(I'_\sigma(t_\sigma + 1 ,n )) \bigwedge t_\sigma<n/2 \right ).
\end{align*}
Conditioning on $E_{22}$, we thus have
\begin{align*}
    \widehat{b} \leq (1- 4 \epsilon)\OPT(I'_\sigma(t_\sigma + 1 ,n ))\leq (1-4\eps)\left(\widehat{\ell} + \frac{\widehat{m}-\widehat{b}}{2} +1\right)\tag{by \cref{kdjfjkdkjfdkj}}
\end{align*}
This is equivalent to saying that 
\begin{align*}
    \widehat b \leq \left(\widehat{\ell} + \frac{\widehat m}{2}+1\right)\frac{1/2 - 2\epsilon}{3/4  - \epsilon}
\end{align*}
which, in turn, is the same as
\begin{align*}
    \widehat{\ell}+\frac{\widehat m}{2}+1\le \left(\frac32-2\eps\right)\left(\widehat{\ell}+\frac{\widehat m-\widehat b}{2}+1\right)\le\left(\frac32-2\eps\right)\left(\OPT(I'_\sigma(t_\sigma + 1 ,n ))+1\right)\tag{by \cref{kdjfjkdkjfdkj}} 
\end{align*}    
Thus, due to \cref{eq:n-sigma}, we have that
\begin{align*}
N_\sigma\le\left(\frac32-2\eps\right)\left(\OPT(I'_\sigma(t_\sigma + 1 ,n ))+1\right)\le\left(\frac32-2\eps\right)\OPT(I_\sigma(t_\sigma + 1 ,n ))+2
\end{align*}
We thus have the following lemma.
\begin{lemma}
\label{lem:case11}
Let the event $\etwotwo\coloneqq\left(\widehat{b} \leq (1- 4 \epsilon)\OPT(I'_\sigma(t_\sigma + 1 ,n )) \bigwedge t_\sigma<n/2\right)$. Then, for any permutation 
$\sigma$ satisfying the event $\etwotwo$, we have
\begin{align*}
N_\sigma\le\left(\frac32-2\eps\right)\OPT(I_\sigma(t_\sigma + 1 ,n ))+2
\end{align*}
\end{lemma}

%% file: putting-it-together.tex
\subsection{Proof of Theorem~\ref{maintheorem}}
\label{sec:putting-it-together}
Here, we combine \cref{lem:case1,lem:case2} to obtain our main result, Theorem~\ref{maintheorem}.

\noindent Let $G$ be the event that \bestfit{} performs strictly better than $3/2$ in the time segment $(1,t_\sigma)$, i.e.,
\begin{align*}
G\coloneqq \left(\BF(I_\sigma(1,t_\sigma))  \leq \left(\frac{3}{2} - 2 \epsilon\right)\OPT(I_\sigma(1,t_\sigma)) + o(\OPT(I))\right)
\end{align*}
We may assume that $\OPT(\tilde I) \to \infty$, where $\tilde I$ denotes the instance $I$ with tiny items removed;
otherwise \cref{opttrunc} applies and Theorem~\ref{maintheorem} holds. Using \cref{lem:case1,lem:case2}, we show that 
$\BF(I_\sigma) \leq (\frac{3}{2} - \epsilon)\OPT(I) + o(\OPT(I))$ with high probability, i.e., at least $1 - o(1)$.

Let $N_\sigma$ be the number of new bins opened by \bestfit{} 
after time $t_\sigma$. Then \cref{kenyon2} gives the following upper bounds on $N_\sigma$ and $\BF(I_\sigma(1, t_\sigma))$:
\begin{align*}
\BF(I_\sigma(1,t_\sigma))&\le \frac32\Opt(I_\sigma(1,t_\sigma))+1\\
N_\sigma&\le\frac32\Opt(I_\sigma(t_\sigma+1,n))+1
\end{align*}

Depending on the range in which $t_\sigma$ lies, we consider four cases.
To use \cref{kenyon2}, we require a very small constant $\alpha>0$ whose value
can be chosen to be arbitrarily close to zero.
\begin{itemize}
    \item Suppose $\prob{t_\sigma\ge (1-\alpha)n}\ge \zeta$\\
    In this case, we condition on $t_\sigma\ge(1-\alpha)n$. Since, this implies 
    that $t_\sigma\ge n/2$ occurs with constant probability, we can apply \cref{lem:case2}.
    Thus, with probability at least $(1-\zeta)$, we have
    \begin{align*}
    \BF(I_\sigma(1,t_\sigma))&\le \left(\frac32-2\eps\right)\Opt(I_\sigma(1,t_\sigma))+o(\Opt(I))\\
                            &\le \left(\frac32-2\eps\right)\Opt(I)+o(\Opt(I))
    \end{align*}
    We also have
    \begin{align*}
    N_\sigma&\le \frac32\Opt(I_\sigma(t_\sigma+1,n))+1\\
            &\le \frac32\Opt(I_\sigma((1-\alpha)n+1,n))+1\\
            &\le \frac32\alpha(1+\delta)\Opt(I)+1\tag{\whp, by \cref{kenyon}}
    \end{align*}
    Hence, with probability at least $1-\zeta-o(1)\ge 1-2\zeta$, we have
    \begin{align*}
    \BF(I_\sigma)&=\BF(I_\sigma(1,t_\sigma))+N_\sigma\\
                &\le \left(\frac32-2\eps+\frac32\alpha(1+\delta)\right)\opt(I)+o(\Opt(I))\\
                &\le\left(\frac32-\eps\right)(1+\delta)\Opt(I)+o(\Opt(I))
    \end{align*}
    as we have chosen $\alpha,\delta$ to be very small compared to $\epsilon$.

    \item Suppose $\prob{\frac n2<t_\sigma<(1-\alpha)n}\ge \zeta$.\\
    In this case, we condition on $\frac n2< t_\sigma<(1-\alpha)n$. We apply
    \cref{lem:case2} to obtain that, with probability at least $1-\zeta$,
    \begin{align*}
    \BF(I_\sigma(1,t_\sigma))&\le \left(\frac32 -2\eps\right)\Opt(I_\sigma(1,t_\sigma))+o(\Opt(I))\\
                    &\le \left(\frac32 -2\eps\right)(1+\delta)\frac{t_\sigma}{n}\Opt(I)+o(\Opt(I))\tag{\whp, by \cref{kenyon}}
    \end{align*}
    We also have,
    \begin{align*}
    N_\sigma&\le \frac32\Opt(I_\sigma(t_\sigma+1,n))+1\\
            &\le \frac 32\left(\frac{n-t_\sigma}{n}\right)(1+\delta)\Opt(I)+1\tag{\whp, by \cref{kenyon}}
    \end{align*}
    Hence, with probability at least $(1-\zeta)(1-o(1))-o(1)\ge1-2\zeta$, we have
    \begin{align*}
    \BF(I_\sigma)&=\BF(I_\sigma(1,t_\sigma))+N_\sigma\\
                &\le \left(\left(\frac32-2\eps\right)\frac{t_\sigma}{n}+\frac32\left(\frac{n-t_\sigma}{n}\right)\right)(1+\delta)\Opt(I)+o(\Opt(I))\\
                &\le \left(\frac32-\eps\right)(1+\delta)\Opt(I)+o(\Opt(I))\tag{since $t_\sigma>n/2$}
    \end{align*}

    \item Suppose $\prob{\alpha n \leq t_\sigma\le n/2}\ge \zeta$.\\
    In this case, we condition on $\alpha n<t_\sigma\le n/2$. We apply
    \cref{lem:case1} to obtain that, with probability at least $1-\zeta$,
    \begin{align*}
    N_\sigma&\le \left(\frac32-2\eps\right)\Opt(I_\sigma(t_\sigma+1,n))+o(\Opt(I))\\
            &\le \left(\frac32-2\eps\right)\frac{n-t_\sigma}{n}(1+\delta)\Opt(I)+o(\Opt(I))\tag{\whp, by \cref{kenyon}}
    \end{align*}
    We also have,
    \begin{align*}
    \BF(I_\sigma(1,t_\sigma))&\le \frac32\Opt(I_\sigma(1,t_\sigma))+1\\
                    &\le\frac32(1+\delta)\frac{t_\sigma}{n}\Opt(I)+1\tag{\whp, by \cref{kenyon}}
    \end{align*}
    Hence, with probability at least $(1-\zeta)(1-o(1))-o(1)\ge1-2\zeta$, we have
    \begin{align*}
    \BF(I_\sigma)&=\BF(I_\sigma(1,t_\sigma))+N_\sigma\\
                &\le \left(\frac32\frac{t_\sigma}{n}+\left(\frac32-2\eps\right)\left(\frac{n-t_\sigma}{n}\right)\right)(1+\delta)\Opt(I)+o(\Opt(I))\\
                &\le \left(\frac32-\eps\right)(1+\delta)\Opt(I)+o(\Opt(I))\tag{since $t_\sigma\le n/2$}
    \end{align*}
    
    \item Suppose $\prob{t_\sigma < \alpha n}$.
    In this case, we condition on $t_\sigma<\alpha n$. Since this also
    implies that $t_\sigma<n/2$ holds with constant probability,
    we can use \cref{lem:case1}, to obtain, with probability at least $1-\zeta$, that
    \begin{align*}
    N_\sigma&\le \left(\frac 32-2\eps\right)\Opt(I_\sigma(t_\sigma+1,n))+o(\Opt(I))\\
            &\le \left(\frac32-2\eps\right)\Opt(I)+o(\Opt(I))
    \end{align*}
    On the other hand,
    \begin{align*}
    \BF(I_\sigma(1,t_\sigma))&\le \frac32\Opt(I_\sigma(1,t_\sigma))+1 \\ & \le \frac32\Opt(I_\sigma(1,\alpha n))+1\\
                            &\le \frac32\cdot\frac{\alpha n}{n}(1+\delta)\Opt(I)+1\tag{\whp, using \cref{kenyon}}
    \end{align*}
   
\noindent  Hence, we obtain with probability at least $1-\zeta-o(1)\ge 1-2\zeta$ that
    \begin{align*}
    \BF(I_\sigma)&\le \left(\left(\frac32-2\eps\right)+\frac32\alpha(1+\delta)\right)\Opt(I)+o(\Opt(I))\\
                &\le \left(\frac32-\eps\right)(1+\delta)\Opt(I)+o(\Opt(I))
    \end{align*}

\noindent     where the last inequality follows since $\delta,\alpha$ are very small constants compared to $\eps$.
\end{itemize}
Hence, in each of the four cases above, if the case occurs with constant probability, we have that the event
$E\coloneqq \Big ( \BF(I_\sigma)\le\left(\frac32-\eps\right)(1+\delta)\Opt(I)+o(\Opt(I))\Big )$
occurs with probability at least $1-2\zeta$. Now, consider the four events
\begin{align*}
V_1&\coloneqq t_\sigma\ge(1-\alpha)n\\
V_2&\coloneqq n/2<t_\sigma<(1-\alpha)n\\
V_3&\coloneqq \alpha n<t_\sigma\le n/2\\
V_4&\coloneqq t_\sigma\le \alpha n\\
\end{align*}
and let $v_i\coloneqq \prob{V_i}$ for each $i\in[4]$. We have
\begin{align*}
\prob{E}=\sum_{i\in 4}\prob{E|V_i}v_i\ge \sum_{i\in[4]}(1-2\zeta)(v_i-\zeta)=(1-2\zeta)(1-4\zeta)\ge 1-6\zeta
\end{align*}
Hence, to conclude, we obtain that
\begin{align*}
\BF(I_\sigma)\le \left(\frac32-2\eps\right)(1+\delta)\Opt(I)+o(\Opt(I))
\end{align*}
holds with probability at least $1-6\zeta$. Since $\delta,\zeta$ can be made arbitrarily close
to zero, while ensuring that they are constants, it follows that
\begin{align*}
\BF(I_\sigma)\le \left(\frac32-\eps\right)\Opt(I)+o(\Opt(I))
\end{align*}
with high probability.
In the remaining low probability events, we can use the worst-case ratio of $1.7$, i.e., $\BF(I_\sigma)\le 1.7\Opt(I)+O(1)$ (see \cite{johnson1974worst}).
Hence we obtain that
\begin{align*}
\expec{\BF(I_\sigma)}& \leq (1-o(1))\left(\frac32-\eps\right)\Opt(I)+o(1)\cdot(1.7\Opt(I)+O(1))+o(\Opt(I))\\
                &\le \left(\frac32-\eps\right)\Opt(I)+o(\Opt(I))
\end{align*}
concluding the proof of Theorem~\ref{maintheorem}.

%% file: bestfit-lowerbound.tex
\section{Lower Bound for the Random-order Ratio of \bestfit{}}
\label{sec:Lower Bounds for Best Fit}
In this section, we will present an improved lower bound on $RR_{\BF}^{\infty}$, the random-order ratio of \bestfit{},
using a computer-aided proof that relies on generating and analyzing the stationary distribution of a large Markov 
chain similar to \cite{AlbersKL21,DBLP:conf/soda/Kenyon96}.
We thus improve the current best lower bound of $1.1$ \cite{AlbersKL21} on the random-order ratio of \bestfit{}
to $1.144$.

We will make use of a model---namely, the i.i.d.~model---to obtain a lower bound on $RR_{\BF}^{\infty}$. In this model, the input for the bin packing algorithm is a sequence of independent, identically distributed (i.i.d.) random variables in $(0,1]$. If $F$ denotes the probability distribution these variables are drawn from, then the performance measure of an algorithm $\mathcal{A}$ is given by $\lim_{n\to\infty}\frac{\expec{\mathcal{A}(I^n(F))}}{\expec{\OPT(I^n(F))}}$, where $I^n(F) \coloneqq (X_1, \ldots ,X_n)$ is a sequence of $n$ random variables drawn i.i.d.\ from $F$.  As was shown in \cite{AlbersKL21}, this model is weaker than the random-order model.
\begin{lemma}\label{iidmodel}
Consider any online bin packing algorithm $\mathcal{A}$. Let $F$ be a discrete distribution on $(0,1]$, and $I^n(F) = (X_1, \ldots, X_n)$ be a list of i.i.d.\ samples
drawn from $F$. As $n \to \infty$, there exists  a list $J$ of $n$ items such that

\begin{align*} \frac{\expec[\sigma]{\mathcal{A}(J_\sigma)}}{\OPT(J)} \geq  \frac{\expec{\mathcal{A}(I^n(F))}}{\expec{\OPT(I^n(F)) }}\end{align*}
where $\sigma$ is a uniformly drawn random permutation of the elements in $J$. 
\end{lemma}

We prove Theorem \ref{aarratio} using \cref{iidmodel}, by exhibiting a probability distribution $F$ that 
causes \bestfit{} to perform relatively badly compared to the optimum solution in the i.i.d.\ model. 
Essentially, we will consider a distribution for which the optimal solution is almost \emph{perfect}, i.e., 
almost all bins are packed to maximum capacity, but \bestfit{} makes many mistakes on average leading to a sub-optimal 
packing. A key difference compared to \cite{AlbersKL21} is that we make 
use of item sizes that are not of the form $1/m$ for some integer $m$, which makes it more difficult to ensure that the optimal packing is almost perfect.

To illustrate the general strategy, we redo the instance used in \cite{AlbersKL21}. 
We will choose $F$ to be the distribution on the item list $\{1/4 , 1/3\}$, with the respective probabilities of item arrivals given 
by $p=0.6,q=0.4$. We say that a bin is open if it has enough space to accommodate
future items, i.e., it has a load at most $3/4$, and closed otherwise.
At any point, only the open bins are of interest to us. And in the \bestfit{} packing of any instance with item sizes in the list $\{1/4,1/3\}$,
there can be at most two open bins at any point of time. All the possibilities of these open bins are shown in the table in \cref{fig:markov}.

Consequently, we can model the behavior of \bestfit{} for this distribution by a Markov chain, where the state space corresponds to the different 
possible open bin configurations, and the transitions correspond to the arrival of different items in $\{1/4,1/3\}$, 
as illustrated on the left side of \cref{fig:markov}. 

\begin{figure}
\centering
    \begin{minipage}{0.45\textwidth}
    \includesvg[scale=0.5]{img/lower-bound.svg}
    \end{minipage}
    \hspace{1cm}
    \begin{minipage}{0.45\textwidth}
    \begin{tabular}{|c|c|}
    \hline
    State& Load of open bin(s)\\
    \hline\hline
    A&No open bins\\\hline
    B&1/4\\\hline
    C&1/3\\\hline
    D&1/2\\\hline
    E&7/12\\\hline
    F&2/3\\\hline
    G&3/4\\\hline
    H&3/4, 1/3\\\hline
    I&3/4, 2/3\\\hline
    \end{tabular}
    \end{minipage}
\caption{The Markov chain of the instance used by \cite{AlbersKL21} to prove a lower bound of $1.1$ on the random-order ratio of \bestfit{}.
The transition probabilities are $p=0.6,q=0.4$.
The bold transitions indicate those that open a new bin.}
\label{fig:markov}
\end{figure}

Consequently, the expected asymptotic behavior of \bestfit{} can be understood  by finding the expected number of transitions in which 
a bin is opened. It can be checked that the chain in \cref{fig:markov} is irreducible and aperiodic, and thus ergodic. So it has a unique stationary distribution 
$\boldsymbol\omega$, with the 
stationary probability of a state $R$ given by $\omega_R$. Let $V_R(t)$ denote the number of visits to a state $R$ of the Markov chain up to time $t$. As the Markov chain is ergodic, we know that $\lim_{t \to \infty} \frac{1}{t} \cdot V_R(t) = \omega_R$ (see \cite{walsh2012knowing}, for example). 
This means that the fraction of time 
spent by the Markov process in the state $R$ approaches its stationary probability $\omega_R$, which we can find computationally by solving a system of linear
equations. We can then find the expected performance of \bestfit{} as follows. 
Let $U$ be the set of all $(R,S)$ such that the transition $R \to S$ opens a new bin. Then, as $n\to\infty$,
\begin{align} 
\label{bfiid} 
\expec{\BF(I^n(F))} \to \sum_{(R,S) \in U} V_R(n)q_{RS} \to n \sum_{(R,S) \in U} \omega_R\cdot q_{RS}
\end{align}
where $q_{RS}$ is the probability that the Markov chain transits from state $R$ to state $S$.
For the distribution given by the list $\{1/4,1/3\}$ and their respective probabilities given by $(p=0.6,q=0.4)$, we can compute $\expec{\BF(I^n(F))}$
to be approximately $3.96n$.
On the other hand, the expected value of the optimal number of bins is given by $\expec{\opt(I^n(F))}\approx4pn+3qn=3.6n$
as the expected number of $1/4$ items is $pn$ and the expected number of $1/3$ items is $qn$.
Overall, we obtain a lower bound of $3.96/3.6=1.1$ on the performance of \bestfit{} in the i.i.d.~model.

Now, we return to our result. We come up with a more complicated distribution to achieve the following result.
However, since the Markov chain corresponding to our example has a large state space, we calculate the stationary probabilities
using a program, which is hosted at \url{https://github.com/bestfitroa/BinPackROA}.
\begin{lemma}
There exists a discrete distribution $F$ such that for $n \to \infty$, we have 
\begin{align*}
\expec{\BF(I_n(F))} > 1.144\cdot\expec{\OPT(I_n(F))}
\end{align*}
\end{lemma}
\begin{proof}
We will take $F$ to be the probability distribution on
the following item list $K$, with the probabilities of each item, respectively, given by $\boldsymbol p$.
\begin{align*}K = (0.245, 0.25, 0.26, 0.27, 0.3, 0.38, 0.46) \qquad
\boldsymbol p = (0.26, 0.13, 0.13, 0.17, 0.15, 0.075, 0.085) 
\end{align*}
It can be computationally checked that the Markov chain corresponding to the behavior of \bestfit{} for the above distribution
has a finite state space ($357$ states).
Moreover, it is irreducible because from the state of $A\coloneqq$ ``no open bins'', we can reach any other state and return back to the state
$A$. Further, state $A$ is also aperiodic because, starting from $A$, both the events ``returning to $A$ in $2$ steps'' and
``returning to $A$ in $3$ steps'' occur with positive probability. (The former event can occur due to the items $0.38,0.46$,
and the latter event can occur due to the items $0.25,0.3,0.3$.) Hence, it follows that the underlying 
Markov chain is irreducible, aperiodic, and hence, ergodic. Then, calculating the stationary distribution $\boldsymbol\omega$
using the code linked above, and using \cref{bfiid}, we obtain that
\begin{align} 
\expec{\BF(I^n(F))} \geq 0.3317621 \cdot n
\end{align}
On the other hand, note that not all items in $K$ are of the form $\frac{1}{m}$ for some integral $m$, hence there is no simple closed form for $\OPT$ simply in terms of the probability of each item in $K$ in general. 
But, we can upper bound the expected performance of the optimal algorithm by 
coming up with a good feasible packing.


\begin{claim} \label{optiid} For the distribution $F$ given by list $K$ and probabilities $\boldsymbol p$, we have $\expec{\OPT(I^n(F))} \leq 0.29 n + o(n)$
\end{claim}
\begin{proof}
We pack the items into the following $3$ bin types.
\begin{align*}
B_1 = \{0.245, 0.245,0.25,0.26 \} \qquad  B_2 = \{0.27, 0.27, 0.46\} \qquad  B_3 = \{0.3,0.3,38\} 
\end{align*}
Let $X_i$ denote the number of items of type $K_i$ (the $i\Th$ item in the list $K$) in the instance $I^n(F)$.
Then $X_i$ is a binomial random variable with mean $np_i$ ($p_i$ refers to the probability of the item $K_i$) and 
variance $np_i(1-p_i) \leq np_i$. Thus, by Chebyshev's inequality
\begin{align*}
\prob{|X_i - np_i| < n^{2/3} }  \leq \frac{np_i}{n^{4/3}} = O\left(\frac{1}{n^{1/3}}\right) 
\end{align*}
Thus, by using a union bound, each $K_i$ appears at most $np_i + n^{2/3}$ times in $I^n(F)$ with high probability. 
When this high probability event occurs, we take $0.13n+n^{2/3},  0.085n+n^{2/3}, 0.075n+n^{2/3}$ number of bins of type $B_1$, $B_2$, $B_3$, respectively,
and it can then be verified that up to $np_i+n^{2/3}$ number of items of type $K_i$ (i.e., all of them) can be packed for all $i$. 
Thus, in this high probability event, we require at most $0.29n+o(n)$ number of bins, which also serves as an upper bound for $\Opt(I^n(F))$.
Consequently, we have the following upper bound on $\OPT(I^n(F))$ with high probability
\begin{align*}\prob{\OPT(I^n(F)) \leq 0.29n + o(n) } = 1 - o(1)
\end{align*}
In the event that occurs with $o(1)$ probability, i.e., when some $K_i$ appears more than $np_i + n^{2/3}$ number of times, 
we use  $\OPT(I^n(F)) \leq n$, to obtain the desired result.
\begin{align*}\expec{\OPT(I^n(F))} \leq \left(0.29n + o(n)\right)\left(1-o(1)\right) + o(1)n = 0.29n + o(n)
\end{align*}
\end{proof}
For the given choice of $K$ and $\boldsymbol p$, using \cref{bfiid} and \cref{optiid}, we finally obtain that
\begin{align*} \lim\limits_{n \to \infty} \frac{\expec{\BF(I^n(F)}  }{\expec{\OPT(I^n(F))}} > 1.144
\end{align*}
as desired.
\end{proof}
Combining this with an application of \cref{iidmodel}, we thus get $RR_{\BF}^{\infty}> 1.144$.

%% file: conclusion.tex
\section{Conclusion}
\label{sec:conclusion}
We have given improved lower and upper bounds on the random-order ratio of \bestfit{}.
To compare with the current best bounds, we have improved the upper bound from $1.5$ to $1.5-\eps$ (for some $\eps\approx 10^{-9}$), and the lower bound
from $1.1$ to $1.144$.
We have not tried to optimize the value of $\varepsilon$ for the sake of simplicity. Moreover, we believe that it is difficult to obtain a significantly 
better upper bound using our techniques. An interesting open question to 
consider is if the conjectured ratio of $1.15$ can be achieved for \bestfit{} in a weaker model, e.g., the i.i.d. model.
Another interesting question is to find a polynomial-time algorithm with a $(1+\varepsilon)$ random-order ratio (or show its impossibility). 

%% file: acknowledgments.tex
\section{Acknowledgments}
\kvn{We sincerely thank Mohit Singh for many helpful initial discussions. \ah{We would also like to thank Riddhipratim Basu for helpful discussions regarding the concentration bounds. Finally, we thank the anonymous reviewers for their helpful comments.}}

%% file: appendix-omitted-proofs.tex
\section{Omitted Proofs}
\label{sec:omitted-proofs}
\input{kenyonlemma-fixed}

\input{helper-claims}
\input{opttrunc}
\input{high-lm-ls}
\input{high-sss-mss-mms}
\input{final-case}
\input{appendix-case-112}
\input{large-tiny-volume}
\input{gen-proportionality}
\input{num-consec-triplets}

\input{s-triplet-good}

\input{other-omitted-proofs}

%% file: kenyonlemma-fixed.tex
\subsection{Proof of \texorpdfstring{\cref{kenyon}}{Opt(1,t) almost equals t/n Opt}}
\label{pf:kenyonlemma}

We first discuss the upright matching problem introduced in \cite{spaccamela} and state a useful result of a stochastic version of upright matching. In the
\emph{upright matching} problem, we are given a $k$ plus $(+)$ points and $k$ minus $(-)$ points on a 2D plane. 
A plus point $(x_+,y_+)$ can be matched to a minus point $(x_-,y_-)$ only if the plus point lies ``upright'' to the minus point, i.e.,
only if $x_+\ge x_-$ and $y_+\ge y_-$. Further, no two points of the same sign can be matched with each other and a point cannot be
matched to more than one point. The objective of the upright matching problem is to match as many points as possible, or, in other words, minimize
the number of unmatched plus points. We denote this minimum possible number of unmatched plus points by the quantity $\mathcal{U}(P_+,P_-)$, where
$P_+$ denotes the set of plus points and $P_-$ denotes the set of minus points.

One can solve the upright matching problem exactly as follows. Sort all the points in non-decreasing order of their $x$-coordinates. When we encounter
a plus point $(x_+,y_+)$, we try to match it to an unmatched minus point $(x_-,y_-)$ satisfying $x_-\le x_+$ and $y_-\le y_+$, with $y_-$ being as 
large as possible. (If no such minus point exists, then the plus point remains unmatched.) 
It can be shown that this procedure gives us a maximum matching. See, e.g., \cite{spaccamela} for a proof.

When it comes to the bin packing setting, an item can be thought of corresponding to a point on a plane, with its time of arrival as the $x$-coordinate
and its size as the $y$-coordinate. To study bin packing under stochastic models, \cite{rhee1993lineA, spaccamela, fischer_thesis}
studied several stochastic variants of upright matching. For our purpose of showing that $\Opt(I_\sigma(1,t))\approx \frac tn\Opt(I)$,
we use a variant stated and proved by Fischer \cite{fischer_thesis}. 

This convergence result is derived from stochastic upright matching. An instance $\mathcal{P} = (\mathcal{P}^{+}, \mathcal{P}^{-})$ for the upright matching problem consists of two finite point sets in $\mathbb{R}^2$ labeled with a \textit{plus}, \textit{minus} respectively.  The goal is to match  as many   
points from $\mathcal{P}^{+}$ to $\mathcal{P}^{-}$    in an \textit{upright} fashion, i.e., while satisfying the constraints that  
\begin{lemma}\label{random-order-matching}\cite{fischer_thesis}
Let $k\in\mathbb N$, and $x_1,x_2,\dots,x_k,y_1,y_2,\dots,y_k$ be a set of reals in $[0,1]$ such that
$y_1\le x_1\le y_2\le x_2\le\dots\le y_k\le x_k$. Consider a random permutation $\pi$ of $[2k]$ and define a
set of plus points $P^{\pi}_+=\{(\pi(i),x_i):i\in[k]\}$ and a set of minus points $P^{\pi}_-=\{(\pi(k+i),y_i):i\in[k]\}$.
Then, there exist universal constants $\beta, C, K >0$ such that
\begin{align*}
\prob{\mathcal U(P^{\pi}_+,P^{\pi}_-) \geq   K\sqrt{k}(\log k)^{3/4}} \leq  C\exp\left(-\beta(\log k)^{3/2}\right)
\end{align*}
\end{lemma}

In fact, Fischer \cite{fischer_thesis} chose coordinates $x_i = 2i, y_i = 2i-1$, but the exact values are not relevant. Instead, the key property used for the result was that the conditions $x_i \geq y_{j} $ for all  $1 \leq j \leq i$ and $x_i < y_j $ for all $i+1 \leq j \leq k$ imply that $(\pi(i), x_i)$ can only be matched to $(\pi(k + j), y_j )$ when $i \geq j$ and $\pi(i) \geq \pi(k+j)$.  We can thus rephrase Fischer's result in the following more convenient graph theoretical form.

\begin{lemma}\label{random-order-matching2}
Let $k\in \mathbb{N}$, and let $G = (X,Y,E)$ be a bipartite graph with vertex set  $U = X \cup Y$, $X = (x_1, x_2, \ldots, x_k)$ and $ Y = (y_1, y_2, \ldots, y_k)$,
and edge set $E$ where $(x_i, y_j) \in E$ iff $1 \leq j \leq i$ for all $i \in [k]$. Furthermore, define $u_i = x_i$ for all $i \in [k], u_{k+i} = y_i$ for all $i \in [k] $. Consider a random permutation $\pi$ of $[2k]$, and randomly permute the vertex set $U$ to obtain a sequence of vertices $u_{\pi(1)}, \ldots, u_{\pi(2k)}$. 
Process the vertices in this order, and when vertex $x_i$ arrives, it is matched to a vertex $y_j$ with the largest index $j$ such that $j\le i$ and $y_j$ appears before $x_i$ in $\pi$ and $y_j$ is unmatched (if no such $y_j$ exists, $x_i$ is left unmatched). 
Let $\mathcal U_X(G^\pi)$ denote the number of unmatched vertices in $X$ that have arrived at \emph{any} intermediate step of this process.
Then, there exist universal constants $\beta, C, K >0$ such that
\begin{align*}
\prob{\mathcal U_X(G^{\pi}) \geq   K\sqrt{k}(\log k)^{3/4}} \leq  C\exp\left(-\beta(\log k)^{3/2}\right)
\end{align*}
\end{lemma}
\begin{remark}
To be precise, Fischer's result (\cref{random-order-matching}) only bounds the final number of unmatched points. But in \cref{random-order-matching2}, the same
bound applies for the number of unmatched vertices in $X$ at \emph{any intermediate step}. This is because, in the matching procedure of \cref{random-order-matching2}, a vertex in $X$ remains unmatched if it is not matched to a point in $Y$ on its arrival. Hence, the number of unmatched vertices in $X$ can only increase with time.
\end{remark}

That ends the discussion on stochastic upright matching. We will be using \cref{random-order-matching} repeatedly in the proof of \cref{kenyon}.
Before starting the proof of \cref{kenyon}, we will state and show two helper claims based on simple probabilistic arguments. These claims
show how the number of items of a particular type and how their volume are distributed in a part of the input, We will need a variant of Hoeffding's inequality that holds for sampling without replacement, mentioned in Hoeffding's original paper \cite{Hoe63}.

\begin{proposition}\label{hoef}
Let $\mathcal{X} = \{x_1, \ldots, x_n\}$ be a finite population of $n$ reals ($\mathcal{X}$ can be a multiset), and $X_1, \ldots , X_m$ be a random sample drawn without replacement from $\mathcal{X}$. Let $a := \min_{1 \leq i \leq n}x_i$ and $b := \max_{1 \leq i \leq n} x_i$. Then, for all $\lambda > 0$, 

\begin{align*}
\prob{\left|\sum_{i=1}^m X_i - \sum_{i=1}^m \expec{X_i}\right| \geq \lambda } \leq 2 \exp\left(-\frac{2\lambda^2}{m(b-a)^2}\right)
\end{align*}
\end{proposition}

\begin{claim} \label{binomperm}
Fix some $t$ such that $1 \leq t \leq n$. For any set of items $D$ in $I$, if $H_D$ is the number of items from $D$ in $I_\sigma(1,t)$,  we have that
\begin{align*}
 \frac{t}{n} |D| - |D|^{2/3}  \leq H_D  \leq \frac{t}{n} |D| + |D|^{2/3}
\end{align*}
 with probability at least $1 - 2\exp(-2|D|^{1/3}).$
\end{claim}
\begin{proof}
Use \cref{hoef}, where the population $\mathcal{X}$ consists of $|D|$ ones and $n-|D|$ zeroes, with a sample size of $m = t$. Note that $0 \leq a \leq b \leq 1$.
\begin{align*}
 \expec{X_i} = \frac{|D|}{n}  \quad \text{and} \quad  \expec{H_D} = \expec{\sum_{i=1}^t X_i} = \frac{t}{n}|D|
\end{align*}
Thus applying the inequality with $\lambda = |D|^{2/3}$ gives the desired claim.  In particular, note that if $|D| = \Omega(\OPT(I))$, then the bound holds with high probability as $\OPT(I) \to \infty$.

\end{proof}

\begin{claim} \label{binompermvol}
Fix some $t$ such that $1 \leq t \leq n$. We have that $\vol(I_\sigma(1,t))$ is at most
\begin{align*}
 \frac{t}{n} \vol(I) + \vol(I)^{2/3}
\end{align*}
 with probability at least $1 - 2\exp(-2\vol(I)^{1/3}).$
\end{claim}
\begin{proof}
Use \cref{hoef}, where the population $\mathcal{X}$ consists of the weights of the items in $I$, with a sample size of $m = t$. Note that $0 \leq a \leq b \leq 1$.
\begin{align*}
 \expec{X_i} = \frac{\vol(I)}{n}  \quad \text{and} \quad  \expec{V_D} = \expec{\sum_{i=1}^t X_i} = \frac{t}{n}\vol(I)
\end{align*}
Thus applying the inequality with $\lambda = \vol(I)^{2/3}$ gives the desired claim. In particular, note as $\vol(I) \geq \OPT(I)/2 - 1$, the bound holds with high probability as $\OPT(I) \to \infty$.

\end{proof}

We are now ready to begin the proof of \cref{kenyon}. Let $m = \OPT(I)$. Fix some small constants $\alpha\in(0,1/2),\mu, \gamma > 0$, let $v  = \ceil{1/\mu} $, and fix some integer $t$ in $[\alpha n, (1-\alpha)n]$. 

Consider an arbitrary optimal packing $\Opt(I)$.
An item is said to be of rank $j$, if it is the $j\Th$ largest item (breaking ties arbitrarily) in the bin it belongs to in $\Opt(I)$.
We call an item a \emph{master} item if its rank is $1$, i.e., it is the largest in the bin it belongs to in $\opt(I)$.
For any item $x$, we define $m(x)$ as the master item in the bin in $\Opt(I)$ that contains $x$.
Consider a master item $x^*$, and the bin $B$ it belongs to in $\Opt(I)$. The item of rank $j$ in the bin $B$ is denoted by $d_j(x^*)$.
For $i\in[v-1]$, define $C_i$ to be the collection of bins in $\Opt(I)$ that contain exactly $i$ number of items (see \cref{fig:collections}).
Let $C_v$ denote the collection of bins in $\Opt(I)$ that contain at least $v$ number of items.
For $i\in[v-1]$, let $I_i$ denote the set of items in the collection $C_i$, and let $b_i$ denote the number of bins in
the collection $C_i$. Also, define $I_v$ to be the set of items of rank at most $v$ in the collection of bins $C_v$, and let $b_v$ denote the number of bins in the collection $C_v$. Note that $\OPT(I) = \sum_{i=1}^v  b_i$.

\begin{figure}
\centering
\includegraphics{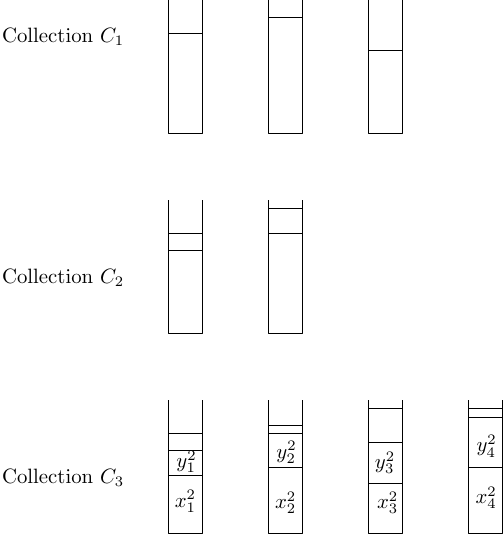}
\caption{The items $y_1^2\le y_2^2\le y_3^2\le y_4^2$
denote the items of rank $2$ in the collection $C_3$. 
The items $x_1^2,x_2^2,x_3^2,x_4^2$ denote the corresponding master items.}
\label{fig:collections}
\end{figure}

Our strategy to bound $\Opt(I_\sigma(1,t))$ is the following. We partition $I_\sigma(1,t)$ into $v+1$ sets
$I_1\cap I_\sigma(1,t),I_2\cap I_\sigma(1,t),\dots,I_{v}\cap I_\sigma(1,t)$, and $I_\sigma(1,t)\setminus (I_1\cup I_2\cup\dots I_v)$.
We consider each $i\in[v]$ separately and pack $I_i\cap I_\sigma(1,t)$ using $i-1$ applications of the procedure detailed in \cref{random-order-matching2}.
We show that, in this way, we can pack $I_i\cap I_\sigma(1,t)$ in at most $\frac tn b_i+o(b_i)$ number of bins. We then try to pack $I_\sigma(1,t)\setminus (I_1\cup I_2\cup\dots I_v)$ using a greedy algorithm like Next-Fit on top of the existing packing, and it can be shown that if we need extra bins, our packing has at approximately $\frac{t}{n} \OPT(I) )+o(\Opt(I))$ many bins with high probability.  In either case, we can compute a packing of $I_\sigma(1,t)$
in approximately $\frac tn(b_1+b_2+\dots+b_v)+o(\Opt(I))=\frac tn \Opt(I)+o(\Opt(I))$ number of bins.

We now provide the formal details. Consider any $i\in[v]$. If $b_i\le \gamma\Opt(I)$, then we trivially have that
$\Opt(I_\sigma(1,t)\cap I_i)\le \gamma\opt(I)$. Now, assume that $b_i>\gamma \Opt(I)$. For the case of $i=1$, each item in $I_1\cap I_\sigma(1,t)$
can be packed in a unique bin, and from \cref{binomperm}, we have the bound
\begin{align*}
\opt\left(I_1\cap I_\sigma(1,t)\right)\le\abs{I_1\cap I_\sigma(1,t)}\le \frac tn\abs{I_1}+\abs{I_1}^{2/3}=\frac tn b_1+b_1^{2/3}
\end{align*}
Now, suppose $i\ge 2$. We construct $i-1$ different graphs as follows. For $j$ such that $2\le j\le i$,
we define a bipartite graph $G_{ij}=(X_{ij},Y_{ij},E_{ij})$ as follows. Let $Y_{ij}=\{y_1^j,y_2^j,\dots,y_{b_i}^j\}$
denote the items of rank $j$ in the collection $C_i$ indexed such that $y_1^j\le y_2^j\le\dots\le y_{b_i}^j$. Let
$X_{ij}=\{x_1^j,x_2^j,\dots,x_{b_i}^j\}$ denote the master items in the collection $C_i$ with $x_r^i=m(y_r^j)$ for
all $r\in[b_i]$. (As as side note, the set $X_{ij}$ is the same for all $j$.)
For $p,q\in[b_i]$, draw an edge between $x_p^j$ and $y_p^j$ if and only if $y_q^j\le y_p^j$, i.e., $(x_p^j,y_1^j)\in E_{ij}$
iff $q\le p$. This graph $G_{ij}$ is exactly the graph in \cref{random-order-matching2} with $k=b_i$.
Also, since $y_p^j=d_j(x_p^j)$, we have that $x_p^j$ shares an edge with $y_q^j$ iff $y_q^j\le d_j(x_p^j)$.

We then apply the procedure in \cref{random-order-matching2} on $G_{ij}$, i.e., we permute the vertices $X_{ij}\cup Y_{ij}$
according to the random permutation $\sigma$ and whenever a vertex $x_p^j$ arrives, we match it with a vertex $y_q^j$ 
(that shares an edge with $x_p^j$ and has already arrived but is yet to be matched) such that $q$ is as large as possible.
Then, \cref{random-order-matching2} tells us that \emph{at all timesteps} in this procedure, the maximum number of unmatched points
in $X_{ij}$ is upper bounded by $O(\sqrt{b_i}(\log b_i)^{3/4})$, with high probability.
In particular, if we consider the matching until the set $I_\sigma(1,t)\cap(X_{ij}\cup Y_{ij})$ arrives, the maximum
number of unmatched points in $X_{ij}$ is at most $O(\sqrt{b_i}(\log b_i)^{3/4})$, with high probability.
Moreover, by \cref{binomperm}, there are at least $\frac tn b_i-b_i^{2/3}$
many items in $I_\sigma(1,t)\cap X_{ij}$. Therefore, with high probability, at least
$\frac tn b_i-b_i^{2/3}-O(\sqrt{b_i}(\log b_i)^{3/4})$ number of points in $I_\sigma(1,t)\cap X_{ij}$ are matched to some
point in $I_\sigma(1,t)\cap Y_{ij}$. By \cref{binomperm}, at most $\frac tn b_i+b_i^{2/3}$ items are in the set
$I_\sigma(1,t)\cap Y_{ij}$ with high probability. Hence, the maximum number of unmatched points in $I_\sigma(1,t)\cap Y_{ij}$
must be at most
\begin{align*}
\left(\frac tn b_i+b_i^{2/3}\right)-\left(\frac tn b_i-b_i^{2/3}-O(\sqrt{b_i}(\log b_i)^{3/4})\right)=2b_i^{2/3}+O(\sqrt{b_i}(\log b_i)^{3/4})
\end{align*}
Hence, overall, the number of items in $I_\sigma(1,t)\cap(X_{ij}\cap Y_{ij})$ that are unmatched is upper bounded by
$2b_i^{2/3}+O(\sqrt{b_i}(\log b_i)^{3/4})+O(\sqrt{b_i}(\log b_i)^{3/4})=o(b_i)$. Using a union bound and summing over all $j$,
which is bounded by $i$, which in turn, is bounded by $v=\ceil{1/\mu}+1$, a constant, we obtain that the number of items that
remain unmatched in $I_\sigma(1,t)\cap I_i$ is at most $o(b_i)$. 

Thus, to pack $I_\sigma(1,t)\cap I_i$, we have the following procedure.
Assign a bin for each master item in $I_1\cap I_\sigma(1,t)$. By \cref{binomperm}, the number of these bins is
at most $\frac tn b_i+b_i^{2/3}$, with high probability, since the number of master items in $I_i$ is $b_i$.
For a non-master item $y$, if it is unmatched, we pack it in a separate bin and close the bin. If it is matched, then it is packed in the
bin in which the master item to which it is matched to is packed. Many items can go into a bin but we claim that this packing is valid.
Indeed, we know that an item $y$ shares and edge with a master item $x$ iff $y\le d_{\rank(y)}(x)$.
And, moreover, no two items of the same rank can be assigned to the same master item. Hence, it follows that no bin overflows its
capacity since $x+d_{1}(x)+\dots+d_i(x)\le 1$.
Hence, the bins in which the matched items is packed is at most $\frac tn b_i+b_i^{2/3}$ in number and since the number of unmatched points
is at most $o(b_i)$, we obtain that
\begin{align*}
\Opt(I_\sigma(1,t)\cap I_i)\le \gamma \Opt(I)+\frac tn b_i+b_i^{2/3}+o(b_i)
\end{align*}
Summing over all $i\in[v]$, we obtain that
\begin{align}
\Opt(I_\sigma(1,t)\cap (I_1\cup I_2\cup \dots \cup I_v))&\le v\gamma \Opt(I)+\frac tn \sum_{i\in[v]}b_i+o(\Opt(I))\nonumber\\
&=v\gamma \Opt(I)+\frac tn \Opt(I)+o(\Opt(I))
\label{majority-packing}
\end{align}
It remains to pack $R\coloneqq I_\sigma(1,t)\setminus (I_1\cup I_2\cup \dots\cup I_v)$. Observe that this set contains items that
have a rank of at least $v+1$. Hence, each item in $R$ has a size at most $1/v\le \mu$. First, we try to pack $R$ greedily,
using \nextfit{}, in the gaps in our packing of $I_\sigma(1,t)\cap (I_1\cup I_2\cup \dots \cup I_v)$.
If we completely pack $R$ in this manner, then the bound in \cref{majority-packing} itself applies. Otherwise, we open new bins for the 
leftover items in $R$ and pack them in these new bins greedily, using \nextfit{}. Then, with an exception of one bin, every bin
must be filled up to a level of at least $1-\mu$. So, the total number of bins used is at most $\frac{1}{1 - \mu} \vol ( I_\sigma(1,t) ) \leq \frac{1}{1 - \mu}\left(\frac{t}{n} \vol(I) + o(\vol(I)) \right ) \leq (1 + 2 \mu) \frac{t}{n} \OPT(I) +o(\Opt(I))$ for small enough $\mu$ with high probability, using \cref{binompermvol}, as $\vol(I) \geq \OPT(I)/2 - 1$. Hence, if extra bins are opened by Next-Fit, we have that with high probability
\begin{align}
\Opt(I_\sigma(1,t)) \leq (1+2\mu)\frac tn \Opt(I)+o(\Opt(I)).
\label{minority-packing}
\end{align}
Combining \cref{majority-packing,minority-packing}, we obtain that with high probabiity
\begin{align*}\OPT(I_\sigma(1,t))
& \leq \frac{t}{n}\left( 1 + 2 \mu + \frac{n v \gamma}{t}  \right) \OPT(I) +  o(\OPT(I)) \\
& \leq \frac{t}{n}\left(1 + 2 \mu + \frac{2 \gamma}{\mu \alpha}\right) \OPT(I)+ o(\OPT(I)) \tag{since $t>\alpha n$}\\
& \leq \frac{t}{n}(1 + 3 \mu) \OPT(I)+ o(\OPT(I)) \\ 
& \leq \frac{t}{n}(1 + \delta) \OPT(I)  
\end{align*}
as long as $\gamma < \frac{\mu^2 \alpha}{2}$ and $4 \mu < \delta$.

Using $\OPT(I) \leq \OPT(I_\sigma(1,t)) + \OPT(I_\sigma(t+1,n)) $, we again obtain with high probability that
\begin{align*}\OPT(I_\sigma(t+1,n)) \geq \OPT(I) - \OPT(I_\sigma(1,t)) 
&\geq \Opt(I)-\frac tn(1+3\mu)\Opt(I)-o(\Opt(I))\\
&= \frac{n-t}{n}
\left(1 + 3 \mu - 3 \mu \frac{n}{n-t} \right)\OPT(I) - o(\OPT(I)) \\ 
& \geq \frac{n-t}{n}\left(1 + 3 \mu - \frac{3 \mu}{\alpha}  \right)\OPT(I) - o(\OPT(I)) \tag{since $t>(1-\alpha)n$}\\ & \geq \frac{n-t}{n}\left(1 - \frac{\delta}{2}\right)\OPT(I)  - o(\OPT(I)) \\ & \geq \frac{n-t}{n}(1 - \delta) \OPT(I)
\end{align*}
with high probability,  as long as $\mu < \frac{\delta/2}{ \frac{3}{\alpha} - 3}$.  We, now use a symmetric analysis  on the time segment $(t+1, n)$ by applying the same argument on the \textit{reverse} arrival order to obtain that with high probability,
\begin{align*}\OPT(I_\sigma(t+1,n)) \leq \frac{n-t}{n}(1 + 3 \mu) \OPT(I) + o(\OPT(I))\leq \frac{n-t}{n}(1 + \delta) \OPT(I)
\end{align*}

which shows that with high probability, we have \begin{align*}
\OPT(I_\sigma(1,t)) \geq \OPT(I) - \OPT(I_\sigma(t+1,n)) 
&\ge \Opt(I)-\frac{n-t}{n}(1+3\mu)\Opt(I)-o(\Opt(I))\\
&= \frac{t}{n}\left(1 + 3 \mu - 3 \mu \frac{n}{t} \right)\OPT(I) - o(\OPT(I))\\ & \geq \frac{t}{n}\left(1 + 3 \mu - \frac{3 \mu}{\alpha}  \right)\OPT(I) - o(\OPT(I))\tag{as $t/{n} \geq \alpha $}\\ & \geq \frac{t}{n}\left(1 - \frac{\delta}{2}\right)\OPT(I)  - o(\OPT(I)) \\ & \geq \frac{t}{n}(1 - \delta) \OPT(I) 
\end{align*}

It remains to show that these bounds hold \textit{for all} $t$ satisfying $\alpha n \leq t \leq (1 - \alpha )n$ with high probability. Note that  \cref{random-order-matching2} gives a bound on the number of unmatched points at all timesteps in the matching procedure, so we only need to show that \cref{binomperm} and \cref{binompermvol} hold for all $\alpha n \leq t \leq (1 - \alpha)n$ simultaneously with high probability, whenever they are applied.

Suppose $t^- \leq t \leq t^+$ where $t^-,t^+$ are consecutive integral multiples of $\floor{\frac{n}{\OPT(I)}}$, and the above bounds hold for both $t^-,t^+$. Then,

$$ \OPT(I_\sigma(1,t))  \geq  \OPT(I_\sigma(1,t^-)) \geq \frac{t^-}{n}(1 - \frac{\delta}{2}) \OPT(I) - o(\OPT(I)) \geq   \frac{t}{n}(1 - \delta) \OPT(I)   $$

$$ \OPT(I_\sigma(1,t))  \leq  \OPT(I_\sigma(1,t^+)) \leq \frac{t^+}{n}(1 + 3 \mu) \OPT(I) + o(\OPT(I)) \leq   \frac{t}{n}(1 + \delta) \OPT(I)   $$

$$ \OPT(I_\sigma(t+1,n))  \geq  \OPT(I_\sigma(t^+ +1 ,n )) \geq  \frac{n - t^+}{n}(1 - \frac{\delta}{2}) \OPT(I) - o(\OPT(I)) \geq   \frac{n-t}{n}(1 - \delta) \OPT(I)   $$

$$  \OPT(I_\sigma(t+1,n))  \leq  \OPT(I_\sigma(t^- + 1 ,n)) \leq \frac{n - t^-}{n}(1 + 3 \mu) \OPT(I) + o(\OPT(I)) \leq   \frac{n-t}{n}(1 + \delta) \OPT(I)   $$

For a fixed $t$, since we apply \cref{binomperm} $O(v^2)$ times and \cref{binompermvol} $O(1)$ times, the failure probability is at most $c_1\exp(-c_2OPT(I)^{1/3} )$ for some constants $c_1,c_2 >0$ as $v$ is a constant. We take a union bound over all $\alpha n \leq t \leq (1 - \alpha)n$ that are integral multiples of $\floor{\frac{n}{\OPT(I)}}$, giving a failure probability of $$O( \OPT(I) \cdot c_1\exp(-c_2OPT(I)^{1/3} )  ) $$

which goes to $0$ as $\OPT(I) \to \infty$, as desired.

%% file: helper-claims.tex
\subsection{Some Results about \bestfit{}}
\label{helper-claims}
\begin{claim}\label{bestfit23}
In any \bestfit{} packing, there can be at most $2$ bins that have no large items and load at most $2/3$ at any point of time.
\end{claim}
\begin{proof}
Assume for the sake of contradiction that at some point in time, there are three bins $B_1,B_2,B_3$ that have 
no large items but have load at most $2/3$. 
Let $x_2,x_3$ be the first items packed in $B_2,B_3$, respectively.
We have $x_2>1/3$ as otherwise $B_1$
would have had enough space to accommodate $x_2$. Similarly, we have $x_3>1/3$. 
As $B_2$ and $B_3$ do not contain large items, we have $x_2,x_3 \leq 1/2$. 
Therefore, when $x_2$ arrived, it must have been the case that $\vol(B_1)>1/2$.
When $x_3$ arrived, the bin $B_2$ must have had at least two items as otherwise, $x_3$ would fit in $B_2$. 
Say the second item packed in $B_2$ is $y_2$. But $y_2$ must be at most $1/3$ as otherwise 
$x_2+y_2>2/3$ which is a contradiction. However, if $y_2\le 1/3$, by \bestfit{} rule, $y_2$ would have been
packed in $B_1$ as at the time of arrival of $y_2$, we have $\vol(B_2)\le 1/2<\vol(B_1)$, thus arriving at
a contradiction. Hence, there can be at most two bins that do not contain large items and have load at most $2/3$
at any point of time.
\end{proof}
\begin{claim}
\label{claim:vol-gt-23}
If any bin $B$ satisfies $\vol(B)\ge 2/3$, then it also satisfies $W(B)\ge1$.
\end{claim}
\begin{proof}
If $B$ contained a large item, then $W(B)\ge 1$ holds since the 
weight of a large item is $1$. Similarly, if $B$ had two
items of type $M/S$, then $W(B)\ge 1$ since the weight of an item of type $M/S$ is $1/2$. 
If $B$ had only one item of type $M/S$ and no large items, then it must have had at least $(2/3-1/2)$ volume of tiny items.
Recalling that a tiny item of size $x$ has weight $3x$,
we obtain $W(B)\ge 0.5+(2/3-1/2)3=1$. Finally, if $B$ only had tiny items, then $W(B)\ge 3(2/3)=2$.
\end{proof}

%% file: opttrunc.tex
\subsection{Proof of \texorpdfstring{\cref{opttrunc}}{BF good when few L and M items}}
\label{pf:opttrunc}
We prove the lemma by showing that, in the \bestfit{} packing of $I_\sigma$, all but a constant number of bins have a final
load greater than $3/4$. In particular, we will show that any bin (with at most two exceptions) that does not contain an
$L$ or $M$ item will have a load greater than $3/4$. Since $k$, the number of $L,M$ items is at most a
constant, we obtain the lemma.

First, note that the number of bins that contain either $L$ or $M$ items is at most $k$, a constant.
Thus, we will only focus on the bins in which every item is either tiny or small. We prove the following claim.
\begin{claim}
For all $t\in[n]$, in the \bestfit{} packing of $I_\sigma(1,t)$, consider the set of bins in which every item is either tiny or small. 
The following properties hold about these bins.
\begin{enumerate}
    \item All of these bins, except at most two, have a load greater than $3/4$.
    \item If there are two bins of load at most $3/4$, then one of these two bins will only contain small items.
\end{enumerate}
\end{claim}

\begin{proof}
The claim follows by simple induction on $t$. Let the $t\Th$ item in the input sequence $I_\sigma$ be $x_t$.
For the base case of $t=1$, the claim trivially holds. For the induction step,
consider any $t<n$ and assume that the claim holds for $I_\sigma(1,t)$.
If all the bins have load at least $3/4$ before $x_{t+1}$ arrives, then the claim continues to hold after packing $x_{t+1}$ also.
Hence, assume that there is at least one bin of load at most $3/4$ just before $x_{t+1}$ arrives.
Now, if $x_{t+1}$ is of type $L,M$, the claim continues to hold as we are only concerned about bins
containing small or tiny items. Hence, we have two cases depending on whether $x_{t+1}$ is small or tiny.

\textbf{Case 1 - $x_{t+1}$ is small.} If there is only one bin of load at most $3/4$ in $\BF(I_\sigma(1,t))$,
then irrespective of whether $x_{t+1}$ opens a new bin or not, the claim continues to hold.
On the other hand, suppose there are two bins of load at most $3/4$. By the induction hypothesis,
one of these two bins, say $B$, only has small items. But since $\vol(B)\le3/4$, it can have at most
two small items, i.e., $\vol(B)$ is, in fact, at most $2/3$, and hence there is enough space to accommodate $x_{t+1}$. The claim thus continues to hold.

\textbf{Case 2 - $x_{t+1}$ is tiny.} If there is only one bin $B_1$ of load at most $3/4$ in $\BF(I_\sigma(1,t))$,
then $x_{t+1}$ will be packed in an already existing bin (since $B_1$ has space to accommodate $x_{t+1}$). Suppose there are two bins, $B_1,B_2$, of load at most $3/4$. One of $B_1,B_2$ must have a load greater than $2/3$ as both these bins contain items of size at most $1/3$.
Suppose $B_2$ has only small items. (This is guaranteed by the induction hypothesis.) Since $B_2$ has load at
most $3/4$, it must have at most two small items, which shows that $\vol(B_2)\le 2/3$.
Hence $\vol(B_1)>2/3$, and so, by the \bestfit{} packing rule, $x_{t+1}$ will either be packed in a bin with load $>3/4$ or into $B_1$ (as $\vol(B_1) > \vol(B_2)$). Thus the claim continues to hold after packing $x_{t+1}$.
\end{proof}
Hence, we have at most $k$ bins that contain $L,M$ items and among the remaining bins,
we have at most two bins of load at most $3/4$. Therefore, $\BF(I_\sigma)\le 4/3\vol(I) +(k+2) \leq 4/3\Opt(I)+(k+2)$.
This concludes the proof of \cref{opttrunc}.

%% file: high-lm-ls.tex
\subsection{Proof of \texorpdfstring{\cref{lem:high-lm-ls-bins}}{BF good when many LM LS bins in Opt}}
\label{pf:high-lm-ls-bins}
Let us call a pair of items \emph{fitting} if their sizes sum up to at most $1$, i.e., they fit in a bin together.
Note that $r_1=r_1(\sigma^*)\ge0.91$ indicates that we have a good number of fitting $ML/SL$ pairs.
Using this fact, we will show that \whp{}, in $\tilde I_\sigma(t'_\sigma+1,t_\sigma)$, there necessarily exist a good
number of sextuplets of the form $(q_1,\ell_1,q_2,\ell_2,q_3,\ell_3)$ where each $q_i$ is either medium or small
and each $\ell_i$ is large
and such that each pair $(q_i,\ell_i)$ is fitting. We will also prove that, in the \bestfit{} packing $\BF(I_{\sigma}(1,t_{\sigma}))$, each such sextuplet uniquely corresponds to a bin of weight $3/2$, thus improving the performance of \bestfit{}.

We now proceed to formalize the above arguments.
Consider the packing $\Opt(\tilde I_{\sigma^*}(1,t_{\sigma^*}))$ and focus on the bins of type $ML/SL$ in this packing.
Let $B_i$ be the $i\Th$ such bin and denote the items it contains by $(q_i,\ell_i)$
where $q_i$ denotes the item which is small or medium and $\ell_i$ indicates the large item.



By \cref{require-prop}, we know that
\begin{align}\label{roptbound_tdelta}
\Opt(\tilde I_{\sigma^*}(1,t_{\sigma^*}))\ge \frac{1-17 \epsilon}{2}\Opt(\tilde I)
\end{align}

Thus, there must exist at least $\frac{r_1-17r_1\epsilon}{2}\OPT(\tilde I)$ many fitting pairs of type $ML/SL$ in $I$. 
A tuple of six items $(q_1,\ell_1,q_2,\ell_2,q_3,\ell_3)$ is called a \emph{fitting $ML/SL$ triplet} in $I_\sigma$
if it satisfies the following properties.
\begin{itemize}
    \item The items $q_1,\ell_1,q_2,\ell_2,q_3,\ell_3$ occur consecutively in that order in the sequence $\tilde I_\sigma$, i.e., in the sequence $I_\sigma$,
    there can only be tiny items in between $q_1,\ell_1,q_2,\ell_2,q_3,\ell_3$.
    \item Each $q_i$ is small or medium, and each $\ell_i$ is large.
    \item Each pair $(q_i,\ell_i)$ is fitting.
\end{itemize}
We obtain the following proportionality claim. 
\begin{claim}
\label{proportionality-lm-ls}
Let $\kappa$ denote the number of disjoint fitting $ML/SL$
pairs in $I$. For some positive constant $u$, suppose $\kappa \geq u \cdot \Opt(\tilde I)$.
Let $n_1,n_2$ be two integers such that $1\le n_1\le n_2\le n$ and $n_2-n_1=\Theta(n)$.
We have that the number of fitting $ML/SL$ triplets in the sequence $I_\sigma(n_1+1,n_2)$ is at least
$$ \frac{u^5}{48(3-u)^5}\left(\frac{n_2-n_1}{n}\right) \kappa  - o(\kappa)$$
with high probability.
\end{claim}
The proof of a general version of this claim
is given in \cref{sec:proportionality}. (This version generalizes both \cref{proportionality-lm} and \cref{proportionality-lm-ls}.)


We use the above claim with $\kappa = \frac{r_1-17r_1\epsilon}{2}\Opt(\tilde  I), u=\frac{r_1-17r_1\epsilon}{2}, n_1= n/4, n_2=n/2$.
Hence, we get that the number of disjoint fitting $ML/SL$ triplets in the time segment $(t'_\sigma+1,t_\sigma)\supseteq (n/4,n/2)$ is at least
\begin{align}
&\frac{1}{48}\left(\frac{\frac{r_1-17\eps r_1}{2}}{3-\frac{r_1-17\eps r_1}{2}}\right)^5\frac14\left(\frac{r_1-17r_1\eps}{2}\right)\Opt(\tilde I)-o(\Opt(\tilde I))\nonumber\\
=&\:\:\frac{(r_1-17\eps r_1)^6}{384(6-r_1+17r_1\eps)^5}\Opt(\tilde I)-o(\Opt(\tilde I))\nonumber\\
\ge&\:\:(1-16\eps)\frac{(r_1-17\eps r_1)^6}{384(6-r_1+17r_1\eps)^5}\Opt(I)-o(\Opt(I))
\label{eq:num-ml-sl-tuplets}
\end{align}
where the last inequality is due to \cref{opti'lowerbound}.

Since $r_1\ge0.91$, we obtain that the number of disjoint fitting $ML/SL$ triplets in the input sequence $I_\sigma$ is 
at least a constant fraction of $\Opt(I)$. Next we will show that,
in the packing of \bestfit{}, each fitting $ML/SL$ triplet in $I_\sigma(t'_\sigma+1,t_\sigma)$ 
corresponds to a unique bin of weight at least $3/2$.
\begin{claim}
\label{claim:ml-sl-tuple-proof}
Suppose there are $\tau$ number of disjoint fitting $ML/SL$ triplets in $I_\sigma(t'_\sigma+1,t_\sigma)$.
Then there will be at least $\tau$ number of bins of weight at least $3/2$ in the packing $\BF(I_\sigma(1,t_\sigma))$.
\end{claim}
\begin{proof}
Consider any $ML/SL$ triplet $q_1,\ell_1,q_2,\ell_2,q_3,\ell_3$ in the time segment $(t'_\sigma+1,t_\sigma)$.
By the definition of an $ML/SL$ triplet, it must be the case that in $I_\sigma(t'_\sigma+1,t_\sigma)$,
there can only be tiny items in between $q_1,\ell_1,q_2,\ell_2,q_3,\ell_3$.
Now if $q_1$ opens a new bin, then $\ell_1$ must be packed along with $q_1$ as no tiny item in between $q_1$ and $\ell_1$ can be packed with $q_1$
or can open a new bin,
by definition of $t'_\sigma$. This leads to the creation of an $ML/SL$ bin which has a weight $3/2$, as desired, and none of the
items from the future $ML/SL$ triplets can be packed in this bin.

On the other hand, suppose $q_1$ is placed in an already existing bin $B$. 
If $B$ contained a large item before packing $q_1$, then we are done since this will result in
the formation of a bin of weight at least $3/2$ and no item from a future $ML/SL$ triplet can be packed in this bin. 

Hence, assume that $B$ did not contain any large items before packing $q_1$. We consider two sub-cases depending on
the volume of $B$ before $q_1$ is packed in it. As the first sub-case, suppose $\vol(B)\ge2/3$ before packing $q_1$.
By \cref{claim:vol-gt-23}, it must be the case that $W(B)\ge1$ before packing $q_1$. Hence, after packing $q_1$,
the bin $B$ has a weight of at least $3/2$. Moreover, since $\vol(B)\ge 2/3$ before packing $q_1$, we have that
$\vol(B)\ge11/12$ after packing $q_1$, implying that no item from a future $ML/SL$ triplet can be packed in $B$.
Finally, we look at the sub-case when $\vol(B)\le 2/3$ before packing $q_1$. We can no longer claim that packing $q_1$ makes the bin $B$
to have a weight of at least $3/2$. However, \cref{bestfit23} guarantees that at any point, and before the arrival of $q_1$ in particular,
there can be at most two bins of load at most $2/3$. Thus, if all of $q_1,q_2,q_3$ are packed in existing bins, this would mean that one of them is 
packed in a bin of load at least $2/3$, thereby resulting in the formation of bin of weight $3/2$.
\end{proof}
We can now complete the proof of \cref{lem:high-lm-ls-bins}. Inequality \ref{eq:num-ml-sl-tuplets} gives us a
lower bound on the number of disjoint fitting
$ML/SL$ triplets in the sequence $ I_\sigma(t'_\sigma+1,t_\sigma)$. \cref{claim:ml-sl-tuple-proof} tells us that,
the number of bins of weight $\ge 3/2$ in the packing $\BF(I_\sigma(1,t_\sigma))$ is at least the number of disjoint fitting $ML/SL$
triplets in the sequence $ I_\sigma(t'_\sigma+1,t_\sigma)$.
Hence, \whp, the
number of bins of weight $\ge3/2$ in $\BF(I_\sigma(1,t_\sigma))$ is at least
\begin{align*}  
(1-16\eps)\frac{(r_1-17\eps r_1)^6}{384(6-r_1+17r_1\eps)^5}\Opt(I)-o(\Opt(I))
\end{align*}

%% file: high-sss-mss-mms.tex
\subsection{Proof of \texorpdfstring{\cref{lem:high-mss-mms-sss-bins}}{BF good when high SSS MMS MSS bins in Opt}}
\label{pf:high-mss-mms-sss-bins}

Since $r_2=r_2(\sigma^*)\ge 0.089$, which is a constant, we obtain that the fraction of $SSS/MSS/MMS$ bins in $\Opt(\tilde I_{\sigma^*}(1,t_{\sigma^*}))$
is at least a constant. This, in turn, means that there are a significant number of small items. The rest of the analysis 
is as follows. First, we will show that in $\tilde I_{\sigma}(t'_\sigma+1,t_\sigma)$, there exist a good number of consecutive $S$-triplets.
Then, we will show that in the packing $\BF(I_\sigma(1,t_\sigma))$, on an average, for two $S$-triplets, there exists at least one bin of weight at least $3/2$.
We thus obtain the lemma. We will delve into the formal details now.

Since $r_2$ denotes the fraction of bins of type $SSS/MSS/MMS$ in $\Opt(\tilde I_{\sigma^*}(1,t_{\sigma^*}))$ and each of these bins contains at least 
one small item, we have that the number of small items in the instance $I$ is at least $r_2\Opt(\tilde I_{\sigma^*}(1,t_{\sigma^*}))$.
By \cref{require-prop}, we know that
\begin{align*}
\Opt(\tilde I_{\sigma^*}(1,t_{\sigma^*}))&\ge \frac{1 -17 \epsilon}{2}\Opt(\tilde I)
\end{align*}

Hence, we have that, in $\tilde I$, there are at least 
\begin{align*}
\frac{r_2 -17 r_2\epsilon}{2}\Opt(\tilde I)
\end{align*}
number of small items. On the other hand, there can be at most $2\Opt(\tilde I)$ many large or medium items in $\tilde I$ as at most $2$ such items fit into a bin. Thus, if $f_S$ denotes the fraction of small items in the instance $\tilde I$, we have
\begin{align}
\label{eq:lb-frac-small}
f_S\ge\frac{r_2 - 17 r_2 \epsilon}{4 + r_2 - 17 r_2 \epsilon}
\end{align}
We call a tuple of items $(S_1,S_2,S_3)$ in the input sequence $I_\sigma$ an $S$-triplet if the following conditions hold.
\begin{itemize}
    \item $S_1$ arrives before $S_2$ and $S_2$ arrives before $S_3$.
    \item If we consider the sequence $\tilde I_\sigma$, then $S_1,S_2,S_3$ form a substring in $\tilde I$, i.e., 
    in the original input sequence $I_\sigma$, in between $S_1,S_2,S_3$,
    there can only be tiny items.
\end{itemize}
The next claim shows that in a randomly permuted input sequence, the number of $S$-triplets in a time segment is proportional to
the length of the segment.
\begin{claim}
\label{no_of_consecutive_tripletsprop}
Suppose $f_S$, the fraction of small items in $\tilde I$, is at least some positive constant. Let $n_1,n_2$ be integers such that $1\le n_1\le n_2\le n$
and $n_2-n_1=\Theta(n)$. Then the maximum number of mutually disjoint $S$-triplets in $I(n_1+1,n_2)$ is at least
\begin{align*}
\left(\frac{n_2-n_1}{3n}\right)f_S^3\abs{\tilde I}-o\left(\abs{\tilde I}\right)
\end{align*}
with high probability.
\end{claim}
The proof of above claim mainly relies on concentration inequalities. 
However, the proof is quite long, and hence, to maintain the flow of the section,
we defer the proof to \cref{no_of_consecutive_tripletsprop-proof}.
We apply the above claim to our case by choosing $n_1=n/4$, $n_2=n/2$,
and we get that, with high probability, the maximum number of mutually disjoint $S$-triplets in
$I_\sigma(n/4,n/2)$ is at least
\begin{align*}
 \frac14\frac{f_S^3}{3}\abs{\tilde I}-o\left(\abs{\tilde I}\right)\nonumber
    &\ge \frac{f_S^3}{12}\Opt(\tilde I)-o(\Opt(\tilde I))\nonumber\\
    &\ge \left(1 - 16 \epsilon\right)\frac{f_S^3}{12}\Opt(I)-o(\Opt(I))
\end{align*}
where the last inequality follows from \cref{opti'lowerbound}.
Recall that we are conditioning on $\eoneoneone$ which implies that $t'_\sigma\le n/4$ and $t_\sigma>n/2$.
Thus, we get that, with high probability, in the random sequence $I_\sigma(t'_\sigma+1,t_\sigma)$, the number of mutually disjoint
$S$-triplets is at least
\begin{align}
\label{lb-num-s-triplets}
 \left(1 - 16 \epsilon\right)\frac{f_S^3}{12}\Opt(I)-o(\Opt(I))
\end{align}
Substituting \cref{eq:lb-frac-small} in \cref{lb-num-s-triplets}, we obtain that, with high probability, the 
number of mutually disjoint $S$-triplets in $I_\sigma(t'_\sigma+1,t_\sigma)$ is at least
\begin{align}
\label{lb-num-s-triplets-1}
\frac{1 - 16 \epsilon}{12} \left(\frac{r_2-17r_2 \epsilon}{4 + r_2 -17r_2 \epsilon}\right)^3 {}\Opt(I)-o(\Opt(I))
\end{align}
The next claim shows that the presence of $S$-triplets after $t'_\sigma$
is good for the performance of \bestfit{} as a good number of bins of weight $3/2$
will be created.
\begin{claim}
\label{claim:s-triplet-good}
If there are $\varkappa$ many mutually disjoint $S$-triplets in $I_\sigma(t'_\sigma+1,t_\sigma)$, then at least $\varkappa/2-O(1)$ number of bins will be formed in $\BF(I_\sigma(1,t_\sigma))$ that have a 
weight at least $3/2$.
\end{claim}
The proof of the above claim is by case analysis and is deferred to \cref{SSSgeneral-proof}.

Combining \cref{lb-num-s-triplets-1,claim:s-triplet-good}, we obtain that, with high probability, the number of bins of weight $3/2$
in $\BF(I_\sigma(1,t_\sigma))$ is at least
\begin{align*}
\frac{1 - 16 \epsilon}{24} \left(\frac{r_2-17r_2 \epsilon}{4 + r_2 -17r_2 \epsilon}\right)^3 {}\Opt(I)-o(\Opt(I))
\end{align*}

%% file: final-case.tex
\subsection{Proof of \texorpdfstring{\cref{lem:low-t-sigma-prime}}{BF good when t-sigma-prime is low}}
\label{pf:final-case-analysis}
We will make use of \cref{lem:high-lm-ls-bins,lem:high-mss-mms-sss-bins} and \cref{weightfn_opt} to show the desired result, conditioned on the event \begin{align*}\eoneoneone\coloneqq\left(t_\sigma' \leq \frac{n}{4} \bigwedge \vol(T(1,t_\sigma)) < 12\epsilon \vol(I_\sigma(1,t_\sigma)) \bigwedge t_\sigma > n/2\right)
\end{align*}
For simplicity, define the quantities
\begin{align}
\alpha_1 = \frac{1-16\eps}{384}\frac{(r_1-17r_1\eps)^6}{(6-r_1+17r_1\eps)^5} \qquad \qquad \alpha_2 = \frac{1 - 16 \epsilon}{24} \left(\frac{r_2-17r_2 \epsilon}{4 + r_2 -17r_2 \epsilon}\right)^3\label{alphas}
\end{align}
For any permutation $\sigma$, we know that $\beta(\sigma)+r_1(\sigma)+r_2(\sigma)=1-o(1)$.
\begin{itemize}
\item Suppose there exists a permutation $\sigma^*$ for which $r_1\coloneqq r_1(\sigma^*)\ge0.91$. 
\footnote{such that $\sigma^*$ satisfies the high probability event given by \cref{require-prop}---$\Opt(\tilde I_{\sigma^*}(1,t_{\sigma^*}))>\frac{1-17 \epsilon}{2}\Opt(\tilde I_{\sigma^*})$\label{footnote}}
Then, from \cref{lem:high-lm-ls-bins}, we get that with high probability (conditioned on $\eoneoneone$) Best-Fit creates at least
\begin{align*}
a_1 \geq \alpha_1\Opt(I) - o(\Opt(I))
\end{align*}
many bins of weight at least $3/2$, where $\alpha_1$ is given by \cref{alphas}. 
\cref{wt-bf-atleast-one} guarantees that every bin (except possibly one) in the packing of \bestfit{} has a weight at least $1$.
Consequently, we have
\begin{align*}
W(I_\sigma(1,t_\sigma))=\sum_{B \in \BF(I_\sigma(1,t_\sigma))} W(B)  &\geq \frac{3}{2} \cdot a_1 +  \BF(I_\sigma(1,t_\sigma))- a_1 -1  \\ 
& \geq \BF(I_\sigma(1,t_\sigma)) + \frac{a_1}{2} - 1 \\ 
& \geq \BF(I_\sigma(1,t_\sigma)) + \frac{\alpha_1}{2} \cdot \Opt(I) - o(\OPT(I)) \\
&\geq  \BF(I_\sigma(1,t_\sigma)) \left( 1 + \frac{\alpha_1}{3} \right) - o(\OPT(I))
\end{align*}
Combining this with \cref{prop:min-wt-bf} and using \cref{weightfn_opt}, we get that with high probability.
\begin{align*}
\BF(I_\sigma(1,t_\sigma)) & \leq \frac{ \left(\frac{3}{2}-\frac{\beta(\sigma)}{2}\right) \left(\frac{1 + 24  \epsilon}{1 - 12 \epsilon} \right  )}{ 1 + \frac{\alpha_1}{3}   } \cdot \OPT(I_\sigma(1,t_\sigma)) + o(\OPT(I)) \\ 
& \leq \frac{\frac{3}{2} \left(\frac{1 + 24  \epsilon}{1 - 12 \epsilon} \right  )}{ 1 + \frac{\alpha_1}{3}   } \cdot \OPT(I_\sigma(1,t_\sigma)) + o(\OPT(I)) \\ 
& \leq \left(\frac{3}{2} - 10^{-7}\right) \OPT(I_\sigma(1,t_\sigma)) + o(\OPT(I))\tag{substituting $r_1=0.91$ in \cref{alphas} as $\alpha_1$ is increasing in $r_1$}\\
& \leq \left(\frac{3}{2} - 2 \epsilon\right)\OPT(I_\sigma(1,t_\sigma)) + o(\OPT(I))
\end{align*}

\item Suppose there exists a permutation $\sigma^*$ for which $r_2\coloneqq r_2(\sigma^*)\ge0.089$.
\footnote{See \cref{footnote}}
Then, from \cref{lem:high-mss-mms-sss-bins}, we get that with high probability (conditioned on $\eoneoneone$) Best-Fit creates at least
\begin{align*}
a_2 \geq \alpha_2\Opt(I) - o(\Opt(I))
\end{align*}
many bins of weight at least $3/2$, where $\alpha_2$ is given by \cref{alphas}. 
\cref{wt-bf-atleast-one} guarantees that every bin (except possibly one) in the packing of \bestfit{} has a weight at least $1$.
Consequently, we have
\begin{align*}
\sum_{B \in \BF(I_\sigma(1,t_\sigma))} W(B)  &\geq \frac{3}{2} \cdot a_2 +  \BF(I_\sigma(1,t_\sigma))- a_2 -1  \\ 
& \geq \BF(I_\sigma(1,t_\sigma)) + \frac{a_2}{2} - 1 \\ 
& \geq \BF(I_\sigma(1,t_\sigma)) + \frac{\alpha_2}{2} \cdot \Opt(I) - o(\OPT(I)) \\
&\geq  \BF(I_\sigma(1,t_\sigma)) \left( 1 + \frac{\alpha_2}{3} \right) - o(\OPT(I))
\end{align*}
Combining this with \cref{prop:min-wt-bf} and using \cref{weightfn_opt}, we get that with high probability.
\begin{align*}
\BF(I_\sigma(1,t_\sigma)) & \leq \frac{ \left(\frac{3}{2}-\frac{\beta(\sigma)}{2}\right) \left(\frac{1 + 24  \epsilon}{1 - 12 \epsilon} \right  )}{ 1 + \frac{\alpha_2}{3}   } \cdot \OPT(I_\sigma(1,t_\sigma)) + o(\OPT(I)) \\ 
& \leq \frac{\frac{3}{2} \left(\frac{1 + 24  \epsilon}{1 - 12 \epsilon} \right  )}{ 1 + \frac{\alpha_2}{3}   } \cdot \OPT(I_\sigma(1,t_\sigma)) + o(\OPT(I)) \\ 
& \leq \left(\frac{3}{2} - 10^{-7}\right) \OPT(I_\sigma(1,t_\sigma)) + o(\OPT(I))\tag{substituting $r_2=0.089$ in \cref{alphas} as $\alpha_2$ is increasing in $r_2$}\\
& \leq \left(\frac{3}{2} - 2 \epsilon\right)\OPT(I_\sigma(1,t_\sigma)) + o(\OPT(I))
\end{align*}
\item Suppose for all the permutations $\sigma$ satisfying the high probability event given by \cref{require-prop}, we have $r_1(\sigma)<0.91$ and $r_2(\sigma)<0.089$. Then $\beta(\sigma)\ge 0.0001$ for each such permutation. Hence, 
by \cref{weightfn_opt}, we have that for all permutations $\sigma$,
\begin{align*}
\BF(I_\sigma(1,t_\sigma))&\le \left(\frac{3}{2} - \frac{\beta(\sigma)}{2}\right)\left(\frac{1 +24 \epsilon}{1 - 12 \epsilon} \right  )\OPT(I_\sigma(1, t_\sigma)+O(1)\\
&\le \left(\frac{3}{2} - 10^{-6}\right)\OPT(I_\sigma(1, t_\sigma)+O(1)\\
&\le \left(\frac{3}{2} - 2\eps\right)\OPT(I_\sigma(1, t_\sigma)+O(1)
\end{align*}
\end{itemize}
Thus, we have shown that, the desired bound on $\BF(I_\sigma(1,t_\sigma))$ holds for all but a negligible fraction of permutations $\sigma$, i.e.,

\begin{align*} \prob{\BF(I_\sigma(1,t_\sigma))\le\left(\frac32-2\eps\right)\Opt(I_\sigma(1,t_\sigma))+o(\Opt(I))\bigg\vert \eoneoneone}\ge 1- o(1)
\end{align*}
as desired.

%% file: appendix-case-112.tex
\subsection{Proof of \texorpdfstring{\cref{lem:high-t-sigma-prime}}{BF good when high t-sigma-prime}}
\label{app:case-112}
Using \cref{kenyon}, we have that the following is true with high probability since $t'_\sigma >n/4$.
\begin{align}
\OPT(I_\sigma(1,t'_\sigma)) \geq  \OPT(I_\sigma(1,n/4)) \geq \frac{1 - \delta}{4} \OPT(I) 
    &\geq  \frac{1 - \delta}{4} \OPT(I_\sigma(1,t_\sigma))
\label{eq:tsigmaprime}
\end{align}
Now, by definition of $t_\sigma'$, all the bins (except possibly one) in $\BF(I_\sigma(1,t_\sigma'))$ must have load greater than $3/4$. 
Hence, let $\mathcal{B}_1$ be the set of bins in $\BF(I_\sigma(1,t_\sigma))$
that have a load greater than $3/4$. We have $\abs{\mathcal{B}_1}\ge\BF(I_\sigma(1,t'_\sigma))-1\ge \OPT(I_\sigma(1,t'_\sigma))-1$.
Then, using \cref{eq:tsigmaprime}, we obtain that  
\begin{align}
\vol(\mathcal B_1) & \geq \frac{3}{4}(\OPT(I_\sigma(1, t_\sigma')) -1)\nonumber\\ 
&\geq \frac{3}{4} \frac{1-\delta}{4} \OPT(I_\sigma(1,t_\sigma))  - \frac34 \nonumber\\ 
&\geq \frac{1}{6} \OPT(I_\sigma(1,t_\sigma)) - \frac34 \label{eq:vol-b1}
\end{align}
with high probability, for small enough $\delta$. 

By definition of $t_\sigma$, all the bins (except possibly one) in $\BF(I_\sigma(1,t_\sigma))$
have a load at least $2/3$.
Let the set of bins in $\BF(I_\sigma(1,t_\sigma))$ 
with load $\geq 2/3$ but $\leq 3/4$ at time $t_\sigma$ be $\mathcal B_2$.
\begin{align*} 
\BF(I_\sigma(1,t_\sigma)) &\leq |\mathcal B_1| + |\mathcal B_2| + 2 \\
& \leq \frac{4}{3} \vol(\mathcal B_1) + \frac{3}{2}\Big( (\vol(I_\sigma(1,t_\sigma)) - \vol(\mathcal B_1) \Big) + 2 \\
&\leq \frac{3}{2}\vol(I_\sigma(1,t_\sigma)) - \frac{\vol(\mathcal B_1)}{6} + 2\\
&\leq \frac{3}{2} \OPT(I_\sigma(1,t_\sigma)) - \frac{1}{6 \cdot 6} \OPT(I_\sigma(1,t_\sigma)) + \frac18 + 2 \quad\text{(using \cref{eq:vol-b1})}\\
&\leq \left(\frac 32-\frac{1}{36}\right) \OPT(I_\sigma(1,t_\sigma))  + \frac{17}{8}
\end{align*}

%% file: large-tiny-volume.tex
\subsection{Proof of \texorpdfstring{\cref{claim:largetinyvolume}}{BF good when tiny volume is large}}
\label{proof-largetinyvolume}
Let $\mathcal B_{\le3/4}$ denote the set of bins in the packing $\Bf(I_\sigma(1,t_\sigma))$
that have a load of at most $3/4$. Let $t_1,t_2,\dots,t_r$ denote the tiny items
in the set of bins $\mathcal B_{\le3/4}$, indexed in the order of their arrival, and let
$\tau(1),\tau(2),\dots,\tau(r)$ denote their respective arrival times, i.e., their indices
in the input sequence $I_\sigma$. Also, for $i\in[r]$, denote the bin into which $t_i$ was
packed by $B_i$, and let $\vol\left(B^{(t)}\right)$ denote the volume of bin $B$
after the $t\Th$ item in the input sequence $I_\sigma$ is packed. Note that the $B_i$-s may not necessarily
be different since two tiny items can be packed into the same bin.

We claim that for all $i\in[r-1]$,
\begin{align}
    \vol\left(B_{i+1}^{\tau(i+1)}\right) \geq \vol\left(B_i^{\tau(i)}\right) + s(t_{i+1})\label{eq:vol-tiny-after-pack}
\end{align}
holds. To see why this is true, first consider the case when $B_i=B_{i+1}$. Then, the above condition holds 
since the volume of bin $B_i$ would have increased by at least $s(t_{i+1})$ after packing $t_{i+1}$ (possibly besides some items between $t_i,t_{i+1}$). 
So, suppose $B_i\ne B_{i+1}$. Since $B_i\in\mathcal B_{\le 3/4}$ and $s\left(t_{i+1}\right)\le 1/4$,
\bestfit{} must have chosen $B_{i+1}$ to pack $t_{i+1}$ because $\vol\left(B_i^{(\tau(i))}\right)\le \vol\left(B_{i+1}^{(\tau(i))}\right)$.
Since $\vol\left(B_{i+1}^{(\tau(i+1))}\right)\ge \vol\left(B_{i+1}^{(\tau(i))}\right)+s(t_{i+1})$, \cref{eq:vol-tiny-after-pack} holds.
As a consequence, combining \cref{eq:vol-tiny-after-pack} for all $i\in[r-1]$, we obtain that
\begin{align*}
\frac34\ge\vol\left(B_r^{(\tau(r))}\right)
        \ge\vol\left(B_1^{(\tau(1))}\right)+\sum_{i=2}^rs(t(i))
                                \ge \sum_{i=1}^rs(t(i))
\end{align*}
Hence, we obtain that the volume of tiny items in the set of bins $\mathcal B_{\le3/4}$ is at most $3/4$. However, recall from the lemma
statement that the total volume of tiny items in the sequence $I_\sigma(1,t_\sigma)$ is at least $12\eps \vol(I_\sigma(1,t_\sigma))$.
Hence, at least $12\eps \vol(I_\sigma(1,t_\sigma))-3/4$ volume of tiny items must be present in bins of load greater than $3/4$ 
in the packing $\Bf(I_\sigma(1,t_\sigma))$. This implies that there are at least $\floor{12\eps \vol(I_\sigma(1,t_\sigma))}$
many bins of load greater than $3/4$ in the packing $\Bf(I_\sigma(1,t_\sigma))$.

%% file: gen-proportionality.tex
\subsection{Proofs of \texorpdfstring{\cref{proportionality-lm} and \cref{proportionality-lm-ls}}{ML and ML MS triplets proportionality lemmas}}
\label{sec:proportionality}
In this section, we will prove a lemma generalizing both \cref{proportionality-lm} and \cref{proportionality-lm-ls}.

First, we define some notation. Let $P\subseteq[0,1]$ be a range of sizes and let $Q\subseteq[0,1]$ be another range of sizes such that
$Q\cap P=\emptyset$, i.e., they are disjoint. Further, we say an item is of type $P$ (respectively, type $Q$) if its size lies in the range
$P$ (respectively, $Q$). Now, consider an input sequence $I_\sigma$. Let $\hat I$ denote the list $I$ obtained after removing all the items not of
type $P/Q$. Similarly, $\hat I_\sigma$ denotes the sequence $I_\sigma$ obtained after deleting the items not of type $P/Q$.
A pair of items $(p,q)$ in $I_\sigma$ is said to be a \emph{fitting $PQ$ pair} if the item $p$ is of type $P$ and item $q$ is of type $Q$
and $p+q\le 1$. Further, a sextuplet of items $(p_1,q_1,p_2,q_2,p_3,q_3)$ in $I_\sigma$ is said to be a \emph{fitting $PQ$ triplet} if
\begin{itemize}
    \item every pair $(p_i,q_i)$ is a fitting $PQ$ pair.
    \item the items $p_1,q_1,p_2,q_2,p_3,q_3$ arrive in that order.
    \item there are no items of type $P/Q$ in between them, i.e., in the sequence $\hat I_\sigma$, the items $p_1,q_1,p_2,q_2,p_3,q_3$
    appear consecutively.
\end{itemize}
We will now state the general lemma and see how \cref{proportionality-lm} and \cref{proportionality-lm-ls} reduce to it.
\begin{lemma}
\label{gen-proportionality}
Suppose $\Opt(\hat I)\to\infty$.
Let $\Gamma$ denote a maximum cardinality set of disjoint fitting $PQ$ pairs in $I$.
Define $x\coloneqq\abs{\Gamma}$ and $y$ to be the number of items in $\hat I$ that are not part of any pair in $\Gamma$.
Suppose there exist positive constants $u,v$ such that $x\ge u\Opt(\hat I)$ and $y\le v\Opt(\hat I)$.
Then, for any two arbitrary $a,b$ such that $1\le a\le b\le n$ and $b-a=\Theta(n)$, we have that the number of disjoint
fitting $PQ$ triplets in the sequence $I_\sigma(a+1,b)$ is at least
\begin{align*}
    \frac{1}{48}\left(\frac{b-a}{n}\right)\left(\frac{1}{2+\frac yx}\right)^5x-o(x)
\end{align*}

with high probability, where $\sigma$ is a uniformly randomly chosen permutation.
\end{lemma} 
\begin{deferredproof}{\cref{proportionality-lm}}
In \cref{gen-proportionality}, substitute type $P$ with type $M$ and type $Q$ with type $L$.
Then $\hat I$ will just be $I'$, and $x$ will just be $d'\ge u\Opt(I')$. We need to calculate what the value
of $v$ will be. Since, in $\Opt(I')$, at least $u$ fraction of bins are of type $LM$,
there can be at most $(1-u)$ fraction of bins of type $L/MM$, which in turn, implies that 
there can be at most $2(1-u)\Opt(I')$ number of items in $I'$ that are not part of any fitting $ML$ pair.
Hence $y\le 2(1-u)\Opt(I')$. Finally, we substitute $b=n_2,a=n_1$ to obtain that the number of disjoint fitting
$ML$ triplets in $I_\sigma(n_1+1,n_2)$ is at least
\begin{align*}
\frac{1}{48}\left(\frac{n_2-n_1}{n}\right)\left(\frac{1}{2+\frac yx}\right)^5x-o(x)
\ge\:\: & \frac{1}{48}\left(\frac{n_2-n_1}{n}\right)\left(\frac{1}{2+\frac{2-2u}{u}}\right)^5d'-o(d')\\
\ge\:\:&\frac{u^5}{1536}\left(\frac{n_2-n_1}{n}\right)d'-o(d')
\end{align*}
with high probability.
\end{deferredproof}
\begin{deferredproof}{\cref{proportionality-lm-ls}}
In \cref{gen-proportionality}, substitute type $P$ with type $M/S$ and type $Q$ with type $L$.
Then $\hat I$ will just be $\tilde I$, and $x$ will just be $\kappa\ge u\Opt(\tilde I)$. 
Since, in $\Opt(\tilde I)$, at least $u$ fraction of bins are of type $ML/SL$,
there can be at most $(1-u)$ fraction of bins of type $L/MM/MSS/MMS/SSS$, which in turn, implies that
there can at most $3(1-u)\Opt(\tilde I)$ number of items in $\tilde I$ that are not part of any fitting $ML/SL$ pair.
Hence $y\le 3(1-u)\Opt(\tilde I)$. Finally, we substitute $b=n_2,a=n_1$ to obtain that the number of disjoint fitting
$ML/SL$ triplets in $I_\sigma(n_1+1,n_2)$ is at least
\begin{align*}
\frac{1}{48}\left(\frac{n_2-n_1}{n}\right)\left(\frac{1}{2+\frac yx}\right)^5x-o(x)\ge\:\: & \frac{1}{48}\left(\frac{n_2-n_1}{n}\right)\left(\frac{1}{2+\frac{3-3u}{u}}\right)^5\kappa-o(\kappa)\\
\ge\:\:&\frac{u^5}{48(3-u)^5}\left(\frac{n_2-n_1}{n}\right)\kappa-o(\kappa)
\end{align*}
with high probability.
\end{deferredproof}
We will now prove the general claim.
\begin{deferredproof}{\cref{gen-proportionality}}
Let the pairs in $\Gamma$ be ordered as $(p_1,q_1),(p_2,q_2),\cdots,(p_x,q_x)$
where $p_1,p_2,\dots,p_x$ are in non-decreasing order. 
At times, we will use $\Gamma$ to denote the set $\{p_1,p_2,\dots,p_x,q_1,q_2,\dots,q_x\}$.
What usage we are referring to will be clear from the context.
All the expectation, variance, and covariance
calculations will be computed over the randomness of $\sigma$. Define $z\coloneqq\abs{\hat I}$.
Observe that, by definitions of $x,y$, it follows that $z=2x+y$.

For a given index $i$, let $X_i$ be the random variable that denotes the number of items of type $P/Q$
in $I_\sigma(1,i)$. We first estimate $X_a$ and $X_b$. Let $Y_j$ be the indicator random variable that
denotes if the $j\Th$ item in $I_\sigma$ is of type $P/Q$. Then, $X_i = \sum_{j=1}^i Y_j$,
and since there are $z$ $P/Q$ items in total, we get
\begin{align*}
\prob{Y_j = 1} = \frac{z}{n}\quad \text{and} \quad \var{Y_j} = \expec{Y_j^2} - \expec{Y_j}^2 \leq  \frac{z}{n}
\end{align*}
Using linearity of expectations, we obtain 
\begin{align} 
\label{expec_X2-1} \expec{X_i} = i\frac{z}{n}
\end{align}
Next, we show that $Y_j, Y_k$ are negatively correlated for $j \neq k$.
Note that $\prob{Y_j = 1 | Y_k  = 1} = \frac{z-1}{n-1}$. This is because once the $k\Th$ position is
occupied by a $P/Q$ item, there are $z-1$ number of $P/Q$ items left to occupy the $j\Th$ position
among the remaining $n-1$ items. Since $\frac{z-1}{n-1}<\frac{z}{n}$, we have that
$\prob{Y_j = 1 | Y_k  = 1}<\prob{Y_j = 1}$. This implies that $\prob{Y_j = 1 \land Y_k  = 1}<\prob{Y_j = 1}\prob{Y_k = 1}$.
Hence,
\begin{align*}
\cov{Y_j}{Y_k} = \expec{Y_jY_k} - \expec{Y_j}\expec{Y_k} < 0
\end{align*}
This gives us the variance bound 
\begin{align} \label{var_X2-1} \var{X_i} = \sum_{j=1}^i \var{Y_j} + 2 \sum_{1 \leq j < k \leq n} \cov{Y_j}{Y_k} \leq i\frac{z}{n}
\end{align}

Hence using \cref{expec_X2-1}, \cref{var_X2-1} and Chebyshev's inequality, we obtain
\begin{align*}
\prob{ \abs{X_{a} - \frac{z}{n}a} \geq  z^{2/3}} \leq \frac{\var{X_{a}}}{z^{4/3}}
 \leq   \frac{a(z/n)}{z^{4/3}} = O\left(\frac{1}{z^{1/3}}\right)
\end{align*}
\begin{align*}
\prob{ \abs{X_{b} - \frac{z}{n}b} \geq  z^{2/3}} \leq \frac{\var{X_{b}}}{z^{4/3}}
 \leq   \frac{b(z/n)}{z^{4/3}} = O\left(\frac{1}{z^{1/3}}\right)
\end{align*}

Hence, 
\begin{align} 
\label{deletion_indices2_highprob-1} 
X_{a} \leq a\frac{z}{n} + z^{2/3} \quad \text{and} \quad X_{b} \geq b\frac{z}{n} - z^{2/3} 
\end{align} 
occur simultaneously with probability at least $1 - O(1/z^{1/3})$. 

We now argue that we have a good number of fitting $PQ$ triplets in between the above indices $a,b$ using a deletion argument.
Observe that randomly shuffling $I$ and then removing all the items not of type $P/Q$ gives us a random
permutation of $\hat I$. We group the $z$ items in $\hat I$ into $z/6$ number of sextuplets as shown below.
\begin{align*}
     \underbrace{*****\:*}_{\mathrm{Sextuplet}\:S_1} \underbrace{*****\:*}_{\mathrm{Sextuplet}\:S_2} \cdots\underbrace{*****\:*}_{\mathrm{Sextuplet}\:S_{z/6}}
\end{align*}
Let $F_i$ be the indicator random variable that takes value $1$ if the sextuplet $S_i$ is a $PQ$ triplet, where all the $6$ items belong to $\Gamma$,
and $0$ otherwise.
(It's not imperative that the items must be from $\Gamma$; they can be from $\hat I\setminus \Gamma$ too.
However, this restriction that we impose will ease the calculations in the concentration analysis that comes later.)
We calculate the probability of $F_i=1$ as follows. The first item needs to be of type $P$ and from $\Gamma$;
there are $x$ choices for this to happen among a total of $z$. Then, among the remaining $z-1$ items,
we need to select one of $x$ items of type $Y$ from $\Gamma$. Then, for the third item, we have $x-1$ choices
(as we already chose the first item to be of type $P$) among $z-2$. We continue in this manner
to obtain that
\begin{align}  \label{QPQPQP_prob}
\prob{F_i = 1}  &= \frac{x}{z} \cdot \frac{x}{z-1} \cdot  \frac{x-1}{z-2}\cdot \frac{x-1}{z-3}\cdot  \frac{x-2}{z-4}\cdot  \frac{x-2}{z-5} \nonumber\\
              &= \frac{x}{2x+y} \cdot \frac{x}{2x+y-1} \cdot  \frac{x-1}{2x+y-2}\cdot \frac{x-1}{2x+y-3}\cdot  \frac{x-2}{2x+y-4}\cdot  \frac{x-2}{2x+y-5} \nonumber\\
 &= \frac{1}{(2 + \frac{y}{x})^6} - o(1)
\end{align}
Now, we proceed to calculate the probability that these $PQ$ triplets are indeed fitting. Towards this,
we construct a bipartite graph as follows. The vertex set is given by $\Gamma_P\cup\Gamma_Q$ where
\begin{align*}
\Gamma_P=\{p_1,p_2,\dots,p_x\}\quad\text{and}\quad\Gamma_Q=\{q_1,q_2,\dots,q_x\}.
\end{align*}
For every $i\in[x]$, we draw an edge between $(p_i,q_i),(p_i,q_{i+1}),\dots,(p_i,q_x)$. Note that an
edge between $p_i$ and $q_j$ implies that the pair $(p_i,q_j)$ is fitting. This is because
$p_i+q_j\le p_j+q_j\le 1$ as $(p_j,q_j)$ is fitting.
Let us denote this graph by $G_\Gamma$.
An example of $G_\Gamma$ when $x=4$ looks like \cref{fig:bfg}.
\begin{figure}[H]
\centering
\includegraphics{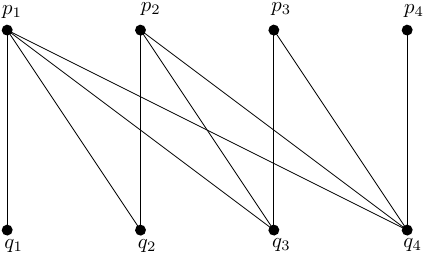}
\caption{The bipartite graph $G_\Gamma$ when $x=4$. Every edge corresponds to a fitting $PQ$ pair.
As a side note, the converse may not be true.}
\label{fig:bfg}
\end{figure}
Using the graph $G_\Gamma$, we now proceed to compute the probability that a $PQ$ triplet is indeed fitting.
Define $H_i$ to be the indicator random variable which takes value $1$ when the sextuplet $S_i$ is a fitting
$PQ$ triplet \emph{and} each of the three consecutive $PQ$ pairs in $S_i$ 
corresponds to an edge in $G_\Gamma$.

Conditioning on $F_i=1$ (i.e., $S_i$ is a $PQ$ triplet),
the probability that the first $PQ$ pair corresponds to an edge in $G_\Gamma$ is given by $\frac{x(x+1)/2}{x^2}$
as there are $x(x+1)/2$ number of edges in $G_\Gamma$ and $x^2$ number of $PQ$ pairs are possible in total.
Once the first $PQ$ pair is chosen such that it corresponds to an edge, all the edges incident on both these vertices
will be deleted, as none of these edges can be the candidates for the $PQ$ pairs chosen next. The number of
these deleted edges will be at most $2x-1$. (This worst case happens when the $PQ$ pair picked is $(p_1,q_x)$.)
Therefore, the number of remaining edges will be at least $x(x+1)/2-(2x-1)=(x-2)(x-1)/2$.
Hence, the probability that the second $PQ$ pair corresponds to an edge in $G_\Gamma$ obtained after removing the 
edges incident on the vertices corresponding to the first $PQ$ pair is at least $\frac{(x-2)(x-1)/2}{(x-1)^2}$.
Similarly, the probability that the third $PQ$ pair corresponds to an edge in $G_\Gamma$ obtained after removing the edges incident 
on the vertices corresponding to the first $PQ$ pairs is at least $\frac{(x-4)(x-3)/2}{(x-2)^2}$. Therefore,
\begin{align} \label{cond_lowerbound}\prob{H_i = 1| F_i = 1} & \geq \frac{x(x+1)/2}{x^2} \cdot  \frac{(x-2)(x-1)/2}{(x-1)^2} \cdot \frac{(x-4)(x-3)/2}{(x-2)^2}\nonumber\\
& = \left(\frac{1}{2} + O(1/x)\right)\left(\frac{1}{2} - O(1/x)\right)\left(\frac{1}{2} - O(1/x)\right)
= \frac{1}{8} - O(1/x)
\end{align}
Now, we compute an upper bound on $\prob{H_i=1|F_i=1}$ in a way similar to how we computed the lower bound.
Assuming $F_i=1$, the probability that the first $PQ$ pair corresponds to an edge in $G_\Gamma$ remains $\frac{x(x+1)/2}{x^2}$
as before. The probability that the second $PQ$ pair corresponds to an edge in $G_\Gamma$ obtained after removing the
edges incident on the vertices in the first $PQ$ pair is at most $\frac{x(x+1)/2-x}{x^2}=\frac{x-1}{2x}$ because at least $x$
edges will be lost due to the first $PQ$ pair.
(This case happens when $(p_1,q_1)$ form the first pair.)
Similarly, at least $(x-1)$ edges will be lost due to the second $PQ$ pair.
Therefore, the probability that the third $PQ$ pair corresponds to an edge in the remaining graph is
$\frac{x(x+1)/2-x-(x-1)}{(x-2)^2}=\frac{x-1}{2(x-2)}$. Therefore
\begin{align} 
\label{cond_upperbound} 
\prob{H_i = 1| F_i = 1} 
& \leq \frac{x+1}{2x} \cdot  \frac{x}{2(x-1)} \cdot \frac{x-1}{2(x-2)}
\nonumber\\
&= \frac{1}{8} + O(1/x)
\end{align}

Now, a lower bound on the number of fitting $PQ$ triplets in the sequence $I_\sigma(a+1,b)$ is given by the random variable
\begin{align*}
S_{(a,b)}=H_{(\frac a6+1)\frac zn}+\dots+H_{\frac b6\frac zn}
\end{align*}
By linearity of expectations, \cref{QPQPQP_prob}, and \cref{cond_lowerbound} we obtain 
\begin{align}
\expec{S_{(a,b)}}\ge \frac{b-a}{6}\cdot \frac{z}{n}\cdot \frac18\cdot \left(\frac{1}{2+\frac yx}\right)^6-o(z)
\label{expec_H}
\end{align}
Now, to prove concentration around the expectation, we compute $\var{S_{(a,b)}}$
and use Chebyshev's inequality. For any $i$, since $H_i$ takes values $0,1$,
\begin{align*}
\var{H_i}=\expec{H_i^2}-\expec{H_i}^2\le 1
\end{align*}
Now, consider any two sextuplets $S_j, S_k$. We claim that the events $H_j=1$ and $H_k=1$ are weakly correlated.
\begin{align*}
\prob{H_j=1| H_k=1}
&=\prob{F_j=1|H_k=1}\prob{H_j=1|F_j=1,H_k=1}\\
&\:\:+\prob{F_j=0|H_k=1}\prob{H_j=1|F_j=0,H_k=1}
\end{align*}
However, for the event $H_j=1$ to occur in the first place, $F_j=1$ must happen. Therefore,
\begin{align}
\prob{H_j = 1 | H_k = 1 } &= \prob{F_j = 1| H_k = 1 }\cdot    \prob{H_j = 1| F_j = 1, H_k = 1}\label{complicated-eq}
\end{align}

The quantity $\prob{F_j = 1| H_k = 1 }$, and an upper bound on $\prob{H_j = 1| F_j = 1, H_k = 1}$ can be calculated similar to \cref{QPQPQP_prob}
and \cref{cond_upperbound}, respectively, except that, instead of $x,z$, we substitute $x-3,z-3$, respectively. 
This is because we are conditioning on $H_k=1$, which means that we are at a loss of three items of
type $P$ and three items of type $Q$ from $\Gamma$. Therefore, we obtain that
\begin{align}
\prob{H_j = 1 | H_k = 1 }& \leq \left(\frac{x-3}{2x+y-6}\frac{x-3}{2x+y-7}\frac{x-4}{2x+y-8} \frac{x-4}{2x+y-9} \frac{x-5}{2x+y-10}  \frac{x-5}{2x+y-11}\right) \nonumber\\
&\:\:\times \left(\frac{(x-3)+1}{2(x-3)} \cdot  \frac{(x-3)}{2((x-3)-1)} \cdot \frac{(x-3)-1}{2((x-3)-2)}\right)\nonumber\\
& \leq \frac{x}{2x+y}  \frac{x}{2x+y-1}  \frac{x-1}{2x+y-2} \frac{x-1}{2x+y-3}  \frac{x-2}{2x+y-4}  \frac{x-2}{2x+y-5} \nonumber\\
&\:\:\times \left(\frac{1}{2} + O(1/x) \right)^3 \nonumber\\ 
& \leq \prob{F_j = 1} \cdot \left(\frac{1}{8} + O(1/x) \right) \label{h-upperbound}
\end{align}Using \cref{cond_lowerbound}, \cref{complicated-eq},we get the covariance estimate
\begin{align}
\cov{H_j}{H_k}  &=  \prob{H_j = 1 } \cdot (\prob{H_j = 1 | H_k = 1 }  - \prob{H_k = 1 })  \nonumber\\
                &=  \prob{H_j = 1|F_j=1}\prob{F_j=1} \cdot (\prob{H_j = 1 | H_k = 1 }  - \prob{H_k = 1|F_k=1}\prob{F_k=1})  \nonumber
\end{align}
An upper bound on $\prob{H_j = 1|F_j=1}$ is given by \cref{cond_upperbound}, an upper bound on $\prob{H_j = 1 | H_k = 1 }$
is given by \cref{h-upperbound},
and a lower bound on $\prob{H_k = 1|F_k=1}$
is given by \cref{cond_lowerbound}. Thus, we obtain
\begin{align}
 \label{cov_upperbound}
 \cov{H_j}{H_k}  
& \leq \prob{F_j = 1}\left(\frac{1}{8}  + O(1/x)\right)  \cdot \left( \prob{F_j = 1}\left(\frac{1}{8}  + O(1/x)\right) - \prob{F_k = 1}\left(\frac{1}{8}  - O(1/x)\right)  \right)   \nonumber \\
& \leq (\prob{F_j = 1})^2\left(\frac{1}{8}  + O(1/x)\right) \cdot O(1/x) \leq O(1/x)
\end{align}
where the penultimate inequality follows since $\prob{F_j=1}=\prob{F_k=1}$. 
Now using \cref{cov_upperbound}, since $z = 2x + y = (2 + \frac{y}{x} ) x = O(x)$ as $y/x$ is upper bounded by some constant as per the lemma statement, we get
\begin{align}\label{H_varbound}
\var{S_{(a,b)}}& =\sum_{i= (\frac a6 +1 )\frac zn}^{\frac b6\frac zn}\var{H_i}+2\sum_{(\frac a6 + 1)\frac zn\le j<k\le \frac a6\frac zn}\cov{H_j}{H_k} \nonumber \\ &\le\frac{z}{6}  + O(z^2/x) = O(x) 
\end{align}
Thus using Chebyshev's inequality and \cref{expec_H}, \cref{cov_upperbound}
\begin{align*}
\prob{S_{(a,b)}\le\expec{S_{(a,b)}}-\left(\expec{S_{(a,b)}}\right)^{2/3}}
                &\le\prob{\abs{S_{(a,b)}-\expec{S_{(a,b)}}}\ge\left(\expec{S_{(a,b)}}\right)^{2/3}}\\
                &\le\frac{\var{S_{(a,b)}}}{\left(\expec{S_{(a,b)}}\right)^{4/3}}\\
                &\le O\left(\frac{x}{x^{4/3}}\right)\\
                &=O\left(\frac{1}{x^{1/3}}\right)
\end{align*}
This thus gives us $S_{(a,b)} \geq \frac{b-a}{n} \cdot \frac{1}{48} \left(\frac{1}{2 + \frac{y}{x}}\right)^5  x - o(x)$ with high probability.\\

Hence, the number of disjoint fitting $PQ$ triplets in $\hat I_\sigma$ between the indices $((a/n)z, (b/n)z)$ is at least $\frac{b-a}{n} \cdot \frac{1}{48} \left(\frac{1}{2 + \frac yx}\right)^5  x - o(x)$ with high probability. As a corollary, the number of disjoint fitting $PQ$ triplets in $\hat I_\sigma$ between the indices $((a/n) z + z^{2/3}, (b/n) z - z^{2/3} )$ is at least 
$\frac{b-a}{n} \cdot \frac{1}{48} \left(\frac{1}{2 + \frac yx}\right)^5  x - o(x) - 2z^{2/3} = \frac{b-a}{n}  \cdot \frac{1}{48} \left(\frac{1}{2 + \frac yx}\right)^5  x - o(x) $ with high probability.\\

Combining this with the high probability event from \cref{deletion_indices2_highprob-1}
\begin{align*}
X_{a} \leq a\frac{z}{n} + z^{2/3} \quad \text{and} \quad X_{b} \geq b\frac{z}{n} - z^{2/3} 
\end{align*}

We obtain that the number of  disjoint fitting $PQ$ triplets in $I_\sigma(a, b)$ is at least
\begin{align*}
 \frac{b-a}{n}\frac{1}{48}\left(\frac{1}{2+\frac yx}\right)^5x-o(x)
\end{align*}
with high probability, as  $(X_{a} , X_{b}) \supseteq ((a/n)z+ z^{2/3} , (b/n)z - z^{2/3})$ with high probability. 
\end{deferredproof}

%% file: num-consec-triplets.tex
\subsection{Proof of \texorpdfstring{\cref{no_of_consecutive_tripletsprop}}{S-triplets proportionality lemma}}
\label{no_of_consecutive_tripletsprop-proof}
In the entire proof, we will implicitly refer to a uniform random permutation $\sigma$ according to which the input $I$ is permuted.
All the expectations and variances will be taken over the randomness of $\sigma$.
Also, let $m=\abs{\tilde I}$.

For a given index $i \in [n]$, let $X_i$ be the random variable that denotes the number of non-tiny items (i.e., of type $L/M/S$)
in $I_\sigma(1,i)$. We will first estimate $X_{n_1}, X_{n_2}$. Let $Y_j$ be the indicator random variable that denotes
if the $j\Th$ item in $I_\sigma$ is non-tiny. Then, $X_i = \sum_{j=1}^i Y_j$,
and since there are $m$ non-tiny items in total, we get
\begin{align*}
\prob{Y_j = 1} = \frac{m}{n}\quad \text{and} \quad \var{Y_j} = \expec{Y_j^2} - \expec{Y_j}^2 \leq  \frac{m}{n}
\end{align*}
Using linearity of expectations, we obtain 
\begin{align} 
\label{expec_X2} \expec{X_i} = i\frac{m}{n}
\end{align}
Next, we show that $Y_j, Y_k$ are negatively correlated for $j \neq k$.
Note that $\prob{Y_j = 1 | Y_k  = 1} = \frac{m-1}{n-1}$. This is because once the $k\Th$ position is
occupied by a non-tiny item, there are $m-1$ non-tiny items left to occupy the $j\Th$ position
among the remaining $n-1$ items. Since $\frac{m-1}{n-1}<\frac{m}{n}$, we have that
$\prob{Y_j = 1 | Y_k  = 1}<\prob{Y_j = 1}$. This implies that $\prob{Y_j = 1 \land Y_k  = 1}<\prob{Y_j = 1}\prob{Y_k = 1}$.
Hence,
\begin{align*}
\cov{Y_j}{Y_k} = \expec{Y_jY_k} - \expec{Y_j}\expec{Y_k} < 0
\end{align*}
This gives us the variance bound 
\begin{align} \label{var_X2} \var{X_i} = \sum_{j=1}^i \var{Y_j} + 2 \sum_{1 \leq j < k \leq n} \cov{Y_j}{Y_k} \leq i\frac{m}{n}
\end{align}

Hence, using \cref{expec_X2}, \cref{var_X2} and Chebyshev's inequality, we obtain
\begin{align*}
\prob{ \abs{X_{n_1} - \frac{m}{n}n_1} \geq  m^{2/3}} \leq \frac{\var{X_{n_1}}}{m^{4/3}}
 \leq   \frac{n_1(m/n)}{m^{4/3}} = O\left(\frac{1}{m^{1/3}}\right)
\end{align*}
\begin{align*}
\prob{ \abs{X_{n_2} - \frac{m}{n}n_2} \geq  m^{2/3}} \leq \frac{\var{X_{n_2}}}{m^{4/3}}
 \leq   \frac{n_2(m/n)}{m^{4/3}} = O\left(\frac{1}{m^{1/3}}\right)
\end{align*}

Hence, 
\begin{align} 
\label{deletion_indices2_highprob} 
X_{n_1} \leq n_1\frac{m}{n} + m^{2/3} \quad \text{and} \quad X_{n_2} \geq n_2\frac{m}{n} - m^{2/3} 
\end{align} 
occur simultaneously with probability at least $1 - O(1/m^{1/3})$. We now argue that we have many $S$-triplets in between the above indices $n_1,n_2$ using a deletion argument.

Observe that randomly shuffling $I$ and then removing all the tiny items gives us a random permutation of $\tilde I$.
We group the $m$ items in $\tilde I$ into $m/3$ number of triplets as shown below. 
\begin{align*}
     \underbrace{***}_{\mathrm{Triplet}\:T_1}\quad\underbrace{***}_{\mathrm{Triplet}\:T_2}\quad\cdots\quad\underbrace{***}_{\mathrm{Triplet}\:T_{m/3}}
\end{align*}

Let $n_S$ denote the number of small items in the input $I$ and recall that $f_S$ denotes the fraction of small items in $\tilde I$,
i.e., $f_S=n_S/m$.
Also, let $Z_i$ be the indicator random variable denoting if the triplet $T_i$ is of type $SSS$ (i.e., only small items). Then,
$$\prob{Z_i = 1} =  \frac{n_S}{m}\cdot \frac{n_S-1}{m-1} \cdot \frac{n_S-2}{m-2} $$

Now, a lower bound for the number of $S$-triplets in the time segment $(n_1+1,n_2)$
is given by the random variable $S_{(n_1,n_2)} = Z_{(\frac{n_1}{3}+1)\frac{m}{n}}  + \cdots +Z_{\frac{n_2}{3} \frac mn}$. By linearity of expectations, we obtain
\begin{align}\label{expec_S}
\expec{S_{(n_1,n_2)}} \geq \frac{n_2-n_1}{3}\cdot  \frac{m}{n}\cdot \frac{n_S}{m}\cdot\frac{n_S-1}{m-1} \cdot \frac{n_S-2}{m-2} =   \frac{n_2-n_1}{3n}f_S^3m - o(m) \end{align}
The above equality follows due to the fact that $f_S$ is a constant and $m$ is large enough; so, for all $i\in\{0,1,2\}$, $(n_S-i)/(m-i)\to f_S$.

We now compute $\var{S_{(n_1,n_2)}}$ and use Chebyshev's inequality. For any $i$, since $Z_i$ is an indicator random variable,
\begin{align*}
    \var{Z_i}&=\expec{Z_i^2}-\expec{Z_i}^2 \leq 1
\end{align*}
Now, consider any two triplets $T_j,T_k$. We claim that the events $Z_j$ and $Z_k$ are negatively correlated. Intuitively this is clear, since if $Z_j=1$, 
the number of small items available for placement in $T_k$ is fewer. Indeed, if $Z_j = 1$ there are $n_S-3$ small items available for placement at $T_k$, and we have
\begin{align*} \prob{Z_k = 1| Z_j = 1}  &= \frac{n_S-3}{m-3} \cdot \frac{n_S-4}{m-4} \cdot  \frac{n_S-5}{m-5} \leq  \frac{n_S}{m} \cdot \frac{n_S-1}{m-1} \cdot  \frac{n_S-2}{m-2}    =   \prob{Z_k = 1}
\end{align*}
Hence we obtain
\begin{align*}
\cov{Z_j}{Z_k}=\expec{Z_jZ_k} -\expec{Z_j}\expec{Z_k}\leq 0
\end{align*}
Combining the above, 
\begin{align} \label{var_S}
\var{S_{(n_1,n_2)}}&=\sum_{i= (n_1/3+1)(m/n)}^{(n_2/3)(m/n)}\var{Z_i}+2\sum_{(n_1/3+1)(m/n)\le j<k\le (n_2/3)(m/n)}\cov{Z_j}{Z_k}\nonumber\\
            &\le\frac{(n_2-n_1)}{3}\frac mn\le \frac m3
\end{align}
Now using \cref{expec_S}, \cref{var_S}  and Chebyshev's inequality,
\begin{align*}
\prob{S_{(n_1,n_2)}\le\expec{S_{(n_1,n_2)}}-\left(\expec{S_{(n_1,n_2)}}\right)^{2/3}}
                &\le\prob{\abs{S_{(n_1,n_2)}-\expec{S_{(n_1,n_2)}}}\ge\left(\expec{S_{(n_1,n_2)}}\right)^{2/3}}\\
                &\le\frac{\var{S_{(n_1,n_2)}}}{\left(\expec{S_{(n_1,n_2)}}\right)^{4/3}}\\
                &\le O\left(\frac{m}{m^{4/3}}\right)\\
                &=O\left(\frac{1}{m^{1/3}}\right)
\end{align*}
This thus gives us $S_{(n_1,n_2)} \geq \frac{n_2-n_1}{3n} f_S^3m - o(m)$ with high probability.

Hence, the number of disjoint $S$-triplets in $\tilde I_\sigma$ in the range of indices $(n_1(m/n), n_2(m/n))$ is at least $\frac{n_2-n_1}{3n} f_S^3m - o(m)$ 
with high probability. The number of disjoint $S$-triplets in $\tilde I_\sigma$ between the indices 
$(n_1(m/n) + m^{2/3}, n_2(m/n) - m^{2/3} )$ is at least $\frac{n_2-n_1}{3n} f_S^3m - o(m) - 2m^{2/3} = \frac{n_2-n_1}{3n} f_S^3m - o(m)$ with high probability.\\

Combining this with the high probability event from \cref{deletion_indices2_highprob}
\begin{align*}
X_{n_1} \leq n_1\frac{m}{n} + m^{2/3} \quad \text{and} \quad X_{n_2} \geq n_2\frac{m}{n} - m^{2/3} 
\end{align*}
we obtain that the number of  disjoint $S$-triplets in $I_\sigma(n_1,n_2)$ is at least
$\frac{n_2-n_1}{3n} f_S^3m - o(m) $ with high probability, as  $(X_{n_1} , X_{n_2}) \supseteq (n_1(m/n) + m^{2/3} , n_2(m/n) - m^{2/3})$ with high probability.

%% file: s-triplet-good.tex
\subsection{Proof of \texorpdfstring{\cref{claim:s-triplet-good}}{S-triplet is good for BF}}
\label{SSSgeneral-proof}
The following claim will be helpful.
\begin{claim}
\label{claim:medium-items-bins}
In any packing of \bestfit{}, at any point of time, there cannot be two $M$-bins both with load at most $3/4$ and both containing tiny items.
\end{claim}
\begin{proof}
Assume for the sake of contradiction that there are two $M$-bins $B_1,B_2$ with tiny items satisfying $\vol(B_1)\le 3/4$
and $\vol(B_2)\le 3/4$, where $B_1$ was opened before $B_2$.
If $B_2$ was opened by a tiny item, then $\vol(B_1) > 3/4$ at that instant, which is a contradiction. 
On the other hand, if $B_2$ was opened by a medium item, then $\vol(B_1) > 1/2$ at that instant, since medium items have size at most $1/2$.
Hence, when the first tiny item is packed in $B_2$, it must be the case that
$\vol(B_1)>3/4$ at that instant, which is a contradiction.
\end{proof}

Now, we will proceed to prove \cref{claim:s-triplet-good}.
First, we prove that every $S$-triplet arriving after $t'_\sigma$ 
(with an exception of at most $O(1)$ number of them) results in the formation of a bin of weight $3/2$ (in which future $S$-items cannot be packed) or an $SS$-bin.
For each $i\in[3]$, let $B_i$ be the bin where $S_i$ was packed. If two of the $B_i$-s are the same, this would create
an $SS$-bin and the lemma stands proved. Hence, from now on, we will assume that all the $B_i$-s are distinct.

\begin{itemize}
\item If any of the $B_i$-s is a $2$-bin before packing $S_i$, then after packing $S_i$, it becomes a $3$-bin, thus becoming a bin of 
weight $3/2$ as well as being closed for the further arriving $S$-items.
\item Suppose for some $i\in\{1,2\}$, $S_i$ opened a new bin or was packed into a bin containing only tiny items. By definition of $t'_\sigma$, no tiny item appearing in between $S_i,S_{i+1}$
can be packed on top of $S_i$, or can open a new bin. So the latter case of being packed into a bin containing only tiny items can occur at most once after $t_\sigma'$, as all bins (with at most one exception) in the packing $\BF(I_\sigma(1,t_\sigma'))$ have load greater than $3/4$.  Consequently, we consider the former case where $S_i$ opens a new bin. As $S_{i+1}$ fits in $B_i$, it must be packed in an already existing bin. Further, since $S_i$ opened a new bin, all the other bins except $B_i$ -- in particular,
$B_{i+1}$ -- must have had a load greater than $2/3$ (at the time when $S_i$ arrived). Therefore, by \cref{claim:vol-gt-23}, since $\vol(B_{i+1})> 2/3$ before the
arrival of $S_{i+1}$, we have that $W(B_{i+1})\ge 1$ before $S_{i+1}$ arrived. Hence, after $S_{i+1}$ is packed, since $W(S_{i+1})=0.5$,
we have that $W(B_{i+1})\ge3/2$. Also, after packing $S_{i+1}$, no small item can be packed in $B_{i+1}$, as it has volume $> 2/3 + 1/4 > 3/4$.
\item Next, we consider the case when each of $S_1,S_2$ is packed in a $1$-bin. 
If any $B_i$ $(i\in\{1,2\})$ was an $L$-bin, then after packing $S_i$, it would become a bin of weight $3/2$, and is closed to future $S$-items.
Similarly, if any $B_i$ $(i\in\{1,2\})$ was an $S$-bin, it would become an $SS$-bin after packing $S_i$.
The case of both $B_1,B_2$ being $M$-bins is slightly trickier. First, note that $\vol(B_1),\vol(B_2) \leq 3/4$ before the arrival of $S_1,S_2$, and that an $M$-bin with tiny items can only be created before $t'_\sigma$, since after $t_\sigma'$, tiny items cannot be added to bins with load $\leq 1/2$ or can open new bins. Moreover by \cref{claim:medium-items-bins}, there can be at most one such $M$-bin with tiny items and load $\leq 3/4$ in $\BF(I_\sigma(1,t_\sigma'))$. Hence, this case can only occur $O(1)$ many times.
\end{itemize}
Thus, we have established that, barring $O(1)$ number of $S$-triplets, every other $S$-triplet results in the formation of a bin of weight $3/2$
or an $SS$-bin. However, our aim is to prove a lower bound on the number of bins of weight $3/2$.

Consider an $SS$-bin $B$ formed in this process.
If another item of type $M/S$ is packed in the bin $B$, then it would mean that the bin $B$ has transformed into a bin of weight $3/2$,
in which case we are good. Assume otherwise, i.e., the bin $B$ continued to be an $SS$-bin. 
But, by \cref{bestfit23}, at any point in time, there can be at most two bins that do not contain a large item that have a 
load of at most $2/3$. And by \cref{claim:vol-gt-23}, every bin with load at least $2/3$ has a weight of at least $3/2$.
Hence, every $SS$-bin (with an exception of at most one) will get converted into a bin of weight $3/2$.

There is one final detail, however. Consider two disjoint $S$-triplets $(S_1,S_2,S_3)$ and $(S_4,S_5,S_6)$. It can happen that
the former $S$-triplet resulted in an $SS$-bin $B$ and one of $S_4,S_5,S_6$ is packed in $B$, thus creating an $SSS$-bin
which has a weight of at least $3/2$. Hence, the bins of weight $3/2$ created by both the triplets are the same.
However, when this happens, note that any of the future non-tiny items---in particular, any of the items from the future $S$-triplets---cannot be packed in $B$.

Therefore, if there are $\varkappa$ number of mutually disjoint $S$-triplets after $t'_\sigma$, at least $\varkappa/2-O(1)$ number of bins with weight $3/2$
will be created after $t'_\sigma$.

%% file: other-omitted-proofs.tex
\subsection{Other Omitted Proofs}
\label{sec:other-omitted-proofs}
\begin{deferredproof}{\cref{prop:probability-fact}}
We have
\begin{align*}
\prob{X|Y}&=\frac{\prob{X\land Y}}{\prob{Y}}\\
        &=\frac{\prob{Y}-\prob{Y\land \overline X}}{\prob{Y}}\\
        &\ge\frac{\prob{Y}-\prob{\overline X}}{\prob{Y}}\\
        &=\frac{\prob{Y}-o(1)}{\prob{Y}}\\
        &=1-\frac{o(1)}{\prob{Y}}\\
        &=1-o(1)
\end{align*}
The last inequality follows because $\prob{Y}$ is at least a constant.
\end{deferredproof}
\begin{deferredproof}{\cref{wt-bf-atleast-one}}
Recall from the notations section (\cref{sec:notation}) that when specifying the type of a bin, we ignore
the tiny items in it.
By \cref{kenyon2}, with at most one exception, every bin is filled to a level at least $2/3$.
Consider any bin $B$ with load at least $2/3$.
If $B$ is of type $L/MM/MS/SS$, then it has a weight of at least one.
If $B$ is of type $LM/LS/SSS/MSS/MMS$, then it has a weight of at least $3/2$.
Otherwise, we consider three cases depending on the contents of $B$.
\begin{itemize}
\item If $B$ has only tiny items, $W(B) \geq 2$ since $\vol(B)\ge 2/3$ and the weight of tiny item is three times its size.
\item If $B$ had a medium item along with tiny items, $W(B) \geq \frac{1}{2} + (\frac{2}{3} - \frac{1}{2}) \cdot 3 =  1$ 
as a medium item has size $\leq \frac{1}{2}$ and $B$ is at least $2/3$ full.
\item If $B$ had a small item along with tiny items, $W(B) \geq \frac{1}{2} + (\frac{2}{3} - \frac{1}{3}) \cdot 3 =  3/2$ 
as a small item has size $\leq \frac{1}{3}$ and $B$ is at least $2/3$ full.
\end{itemize}
\end{deferredproof}
\begin{deferredproof}{\cref{prop:min-wt-bf}}
The lemma follows from the following string of inequalities. Let $\mathcal P$ denote the packing $\BF(I_\sigma(1,t_\sigma))$.
\begin{align*}
\BF(I_\sigma(1,t_\sigma))=\sum_{B\in \mathcal P}1 
    &\le\sum_{B\in \mathcal P}W(B)+1\tag{by \cref{wt-bf-atleast-one}}\\
    &=\sum_{B\in \mathcal P}\sum_{x\in B}W(x)+1\\
    &=\sum_{x\in I_\sigma(1,t_\sigma)}W(x)+1\\
    &=W(I_\sigma(1,t_\sigma))+1
\end{align*}
The lemma stands proved.
\end{deferredproof}
\begin{deferredproof}{\cref{opti'lowerbound}}
We first pack $\tilde I_\sigma(1,t_\sigma)$ in $\Opt(\tilde I_\sigma(1,t_\sigma))$ number of bins. Then we pack
$T(1,t_\sigma)$ using \nextfit{} \cite{johnson1973near}; each bin (with only the last bin being a possible exception) will be filled to a level greater than $3/4$. Therefore, the
total number of bins used is at most $\Opt(\tilde I_\sigma(1,t_\sigma))+\frac43 \vol(T(1,t_\sigma))+1$. Thus, we have
\begin{align*}\OPT(I_\sigma(1 , t_\sigma)) &\le \OPT(\tilde I_\sigma(1 , t_\sigma)) + \frac43 \cdot 12 \epsilon \vol(I_\sigma(1,t_\sigma)) + 1\\
 & = \OPT(\tilde I_\sigma(1 , t_\sigma)) + 16 \epsilon \OPT(I_\sigma(1 , t_\sigma))) + 1
\end{align*}
which gives the following lower bound on $ \OPT(\tilde I_\sigma(1 , t_\sigma))$:
\begin{align*}
 \OPT(\tilde I_\sigma(1 , t_\sigma))\ge \left(1-16 \epsilon\right)\Opt(I_\sigma(1,t_\sigma))-1
\end{align*}
\end{deferredproof}
\begin{deferredproof}{\cref{weightoptnotiny}}
We first upper bound the weight of tiny items in the time segment $(1,t_\sigma)$ in terms of the weight of the non-tiny items in 
$(1,t_\sigma)$ as follows.
\begin{align*}
W(T(1,t_\sigma)) &= 3 \vol(T(1,t_\sigma)) \\
& \leq 3 (12 \epsilon \vol(I_\sigma(1,t_\sigma)))\\
& = 3 (12 \epsilon \vol(\tilde I_\sigma(1,t_\sigma)) +   12\epsilon \vol(T(1,t_\sigma)) ) \\
& = 3 \left(12 \epsilon \vol(\tilde I_\sigma(1,t_\sigma)) +  4\epsilon W(T(1,t_\sigma))\right) 
\end{align*}
Rearranging terms, we obtain that
\begin{align*}
W(T(1,t_\sigma)) \leq \frac{36 \epsilon}{1 - 12 \epsilon }  \vol(\tilde I_\sigma(1,t_\sigma)) \leq  \frac{36 \epsilon}{1 - 12 \epsilon}  W(\tilde I_\sigma(1,t_\sigma)) 
\end{align*}
Then,
\begin{align*}
\frac{W(I_\sigma(1, t_\sigma))}{\OPT(I_\sigma(1, t_\sigma))} \leq \frac{W(\tilde I_\sigma(1, t_\sigma)) + W(T(1,t_\sigma)) }{\OPT(\tilde I_\sigma(1, t_\sigma))} 
&\leq \frac{W(\tilde I_\sigma(1, t_\sigma)) + \frac{36 \epsilon}{1 - 12 \epsilon}W(\tilde I_\sigma(1, t_\sigma))   }{\OPT(\tilde I_\sigma(1, t_\sigma) )} \nonumber\\ &\leq \frac{W(\tilde I_\sigma(1, t_\sigma) )}{\OPT(\tilde I_\sigma(1, t_\sigma) )} \left (\frac{1 + 24  \epsilon}{1 - 12 \epsilon} \right  )
\end{align*}
\end{deferredproof}

\begin{deferredproof}{\cref{weightfn_opt}}
We lower bound $\OPT(\tilde I_\sigma(1, t_\sigma))$ in terms of the weight $W(\tilde I_\sigma(1,t_\sigma))$ as follows
\begin{align} 
W(\tilde I_\sigma(1,t_\sigma))
=\sum_{B \in \OPT(\tilde I_\sigma(1, t_\sigma) )} W(B)&\leq \beta(\sigma) \OPT(\tilde I_\sigma(1, t_\sigma) ) + 
        (1 - \beta(\sigma))\OPT(\tilde I_\sigma(1, t_\sigma)) \frac{3}{2} + O(1) \nonumber\\ 
    &\leq \left(\frac{3}{2} - \frac{\beta(\sigma)}{2}\right)\OPT(\tilde I_\sigma(1, t_\sigma) ) + O(1)\label{eq:high-one-wt-bins}
\end{align}

Substituting \cref{eq:high-one-wt-bins} in \cref{weightoptnotiny}, we obtain a lower bound on $\Opt(I_\sigma(1,t_\sigma))$ in terms of $W(I_\sigma(1, t_\sigma)$.
\begin{align}
\frac{W(I_\sigma(1, t_\sigma))}{\OPT(I_\sigma(1, t_\sigma))} \leq \frac{W(\tilde I_\sigma(1, t_\sigma) )}{\OPT(\tilde I_\sigma(1, t_\sigma) )} \left (\frac{1 + 24  \epsilon}{1 - 12 \epsilon} \right  )\le \left (\frac{1 + 24  \epsilon}{1 - 12 \epsilon} \right  )\left(\frac32-\frac{\beta(\sigma)}{2}\right)+\frac{O(1)}{\OPT(\tilde I_\sigma(1, t_\sigma) )}\label{weight-by-opt}
\end{align}

Using \cref{prop:min-wt-bf}, and \cref{weight-by-opt}, we obtain
\begin{align*} \BF(I_\sigma(1,t_\sigma))&\le W(I_\sigma(1, t_\sigma))+1 \\
    &\leq \left(\frac{3}{2} - \frac{\beta(\sigma)}{2}\right)
\left(\frac{1 +24 \epsilon}{1 - 12 \epsilon} \right  ) \OPT(I_\sigma(1, t_\sigma))+\frac{O(1)\cdot\OPT(I_\sigma(1, t_\sigma) )}{\OPT(\tilde I_\sigma(1, t_\sigma) )}+1\\
&\le \left(\frac{3}{2} - \frac{\beta(\sigma)}{2}\right)\left(\frac{1 +24 \epsilon}{1 - 12 \epsilon} \right  ) \OPT(I_\sigma(1, t_\sigma))+O(1)
\end{align*}
where the last inequality follows from \cref{opti'lowerbound}. 
\end{deferredproof}
\begin{deferredproof}{\cref{claim:opt-prime-lb}}
Since $\hat b$ is the number of $LM$ bins in $\Opt(I'_\sigma(1,t_\sigma))$,
we have that $\hat b\le \Opt(I'_\sigma(1,t_\sigma))\le \hat\ell+\frac{\hat m-\hat b}{2}+1$ (from \cref{kdjfjkdkjfdkj}).
Rearranging terms, we obtain $\hat\ell+\frac{\hat m}{2}+1\ge \frac32\hat b$.
Adding $2\hat\ell+\hat m$ on both sides, we obtain
\begin{align*}
3\hat\ell+\frac32\hat m+1\ge \frac32\hat b+2\hat \ell+\hat m
\end{align*}
Rearranging terms, we obtain
\begin{align}
3\left(\hat\ell+\frac{\hat m-\hat b}{2}\right)+1\ge 2\hat\ell+\hat m\label{gfkjgjkfjkghjklfhjklh}
\end{align}
From \cref{kdjfjkdkjfdkj}, we have $\Opt(I'_\sigma(t_\sigma+1,n))\ge \hat\ell+\frac{\hat m-\hat b}{2}$. Hence
\begin{align*}
\frac32\Opt(I'_\sigma(t_\sigma+1,n))&\ge \frac32\left(\hat\ell+\frac{\hat m-\hat b}{2}\right)\\
                            &= \frac32\left(\hat\ell+\frac{\hat m-\hat b}{2}\right)\\
                            &\ge \hat\ell+\frac{\hat m}{2}-\frac12\tag{from \cref{gfkjgjkfjkghjklfhjklh}}
\end{align*}
Multiplying both sides by $2/3$ gives us the claim.
\end{deferredproof}

%% file: main.bbl
\newcommand{\etalchar}[1]{$^{#1}$}
\begin{thebibliography}{CJCG{\etalchar{+}}13}

\bibitem[ADKS22]{DBLP:conf/icalp/AyyadevaraD0S22}
Nikhil Ayyadevara, Rajni Dabas, Arindam Khan, and K.~V.~N. Sreenivas.
\newblock Near-optimal algorithms for stochastic online bin packing.
\newblock In {\em 49th International Colloquium on Automata, Languages, and
  Programming, {ICALP} 2022, July 4-8, 2022, Paris, France}, 2022.

\bibitem[AGJ23]{AlbersGJ23}
Susanne Albers, Waldo G{\'{a}}lvez, and Maximilian Janke.
\newblock Machine covering in the random-order model.
\newblock {\em Algorithmica}, 85(6):1560--1585, 2023.

\bibitem[AKL21a]{albers_et_al_MFCS}
Susanne Albers, Arindam Khan, and Leon Ladewig.
\newblock Best fit bin packing with random order revisited.
\newblock {\em Algorithmica}, 83(9):2833--2858, 2021.

\bibitem[AKL21b]{AlbersKL21}
Susanne Albers, Arindam Khan, and Leon Ladewig.
\newblock Improved online algorithms for knapsack and {GAP} in the random order
  model.
\newblock {\em Algorithmica}, 83(6):1750--1785, 2021.

\bibitem[BBD{\etalchar{+}}18]{BaloghBDEL18}
J{\'{a}}nos Balogh, J{\'{o}}zsef B{\'{e}}k{\'{e}}si, Gy{\"{o}}rgy D{\'{o}}sa,
  Leah Epstein, and Asaf Levin.
\newblock A new and improved algorithm for online bin packing.
\newblock In {\em European Symposium on Algorithms (ESA)}, volume 112, pages
  5:1--5:14, 2018.

\bibitem[BBD{\etalchar{+}}19]{BaloghBDEL19}
J{\'{a}}nos Balogh, J{\'{o}}zsef B{\'{e}}k{\'{e}}si, Gy{\"{o}}rgy D{\'{o}}sa,
  Leah Epstein, and Asaf Levin.
\newblock A new lower bound for classic online bin packing.
\newblock In {\em {WAOA}}, volume 11926, pages 18--28. Springer, 2019.

\bibitem[{Car}19]{fischer_thesis}
{Carsten Oliver Fischer}.
\newblock {\em New Results on the Probabilistic Analysis of Online Bin Packing
  and its Variants}.
\newblock PhD thesis, Rheinische Friedrich-Wilhelms-Universität Bonn, December
  2019.

\bibitem[CJCG{\etalchar{+}}13]{coffman2013bin}
Edward~G Coffman~Jr, J{\'{a}}nos Csirik, G{\'{a}}bor Galambos, Silvano
  Martello, and Daniele Vigo.
\newblock Bin packing approximation algorithms: survey and classification.
\newblock In {\em Handbook of combinatorial optimization}, pages 455--531.
  Springer New York, 2013.

\bibitem[CJGJ96]{coffman1997approximation}
Edward~G Coffman~Jr, Michael~R Garey, and David~S Johnson.
\newblock {\em Approximation Algorithms for Bin Packing: A Survey}, page
  46–93.
\newblock PWS Publishing Co., USA, 1996.

\bibitem[CJJLS93]{coffman1993probabilistic}
Edward~G Coffman~Jr, David~S Johnson, George~S Lueker, and Peter~W Shor.
\newblock Probabilistic analysis of packing and related partitioning problems.
\newblock {\em Statistical Science}, 8(1):40--47, 1993.

\bibitem[CJJSW93]{coffman1993markov}
Edward~G Coffman~Jr, David~S Johnson, Peter~W Shor, and Richard~R Weber.
\newblock Markov chains, computer proofs, and average-case analysis of best fit
  bin packing.
\newblock In {\em Proceedings of the twenty-fifth annual ACM symposium on
  theory of computing}, pages 412--421, 1993.

\bibitem[CJJSW97]{coffman1997bin}
Edward~G Coffman~Jr, David~S Johnson, Peter~W Shor, and Richard~R Weber.
\newblock Bin packing with discrete item sizes, part ii: Tight bounds on first
  fit.
\newblock {\em Random Structures \& Algorithms}, 10(1-2):69--101, 1997.

\bibitem[CKPT17]{ChristensenKPT17}
Henrik~I Christensen, Arindam Khan, Sebastian Pokutta, and Prasad Tetali.
\newblock Approximation and online algorithms for multidimensional bin packing:
  {A} survey.
\newblock {\em Computer Science Review}, 24:63--79, 2017.

\bibitem[CMS93]{clarkson1993four}
Kenneth~L Clarkson, Kurt Mehlhorn, and Raimund Seidel.
\newblock Four results on randomized incremental constructions.
\newblock {\em Computational Geometry}, 3(4):185--212, 1993.

\bibitem[dlVL81]{VegaL81}
W~Fernandez de~la Vega and George~S Lueker.
\newblock Bin packing can be solved within 1+epsilon in linear time.
\newblock {\em Combinatorica}, 1(4):349--355, 1981.

\bibitem[DS14]{opt-bestfit}
Gy{\"{o}}rgy D{\'{o}}sa and J~Sgall.
\newblock Optimal analysis of best fit bin packing.
\newblock In {\em ICALP}, pages 429--441, 2014.

\bibitem[Fre83]{freeman1983secretary}
PR~Freeman.
\newblock The secretary problem and its extensions: A review.
\newblock {\em International Statistical Review/Revue Internationale de
  Statistique}, pages 189--206, 1983.

\bibitem[GGJY76]{garey1976resource}
Michael~R Garey, Ronald~L Graham, David~S Johnson, and Andrew Chi-Chih Yao.
\newblock Resource constrained scheduling as generalized bin packing.
\newblock {\em Journal of Combinatorial Theory, Series A}, 21(3):257--298,
  1976.

\bibitem[GGU72]{DBLP:conf/stoc/GareyGU72}
Michael~R Garey, Ronald~L Graham, and Jeffrey~D Ullman.
\newblock Worst-case analysis of memory allocation algorithms.
\newblock In {\em STOC}, pages 143--150, 1972.

\bibitem[GJ78]{GareyJ78}
Michael~R Garey and David~S Johnson.
\newblock ``{S}trong'' {NP}-completeness results: Motivation, examples, and
  implications.
\newblock {\em J. {ACM}}, 25(3):499--508, 1978.

\bibitem[GKL21]{GuptaSCR}
Anupam Gupta, Gregory Kehne, and Roie Levin.
\newblock Random order online set cover is as easy as offline.
\newblock In {\em 62nd {IEEE} Annual Symposium on Foundations of Computer
  Science, {FOCS} 2021, Denver, CO, USA, February 7-10, 2022}, pages
  1253--1264. {IEEE}, 2021.

\bibitem[GR20]{GoemansR20}
Michel~X Goemans and Thomas Rothvoss.
\newblock Polynomiality for bin packing with a constant number of item types.
\newblock {\em J. {ACM}}, 67(6):38:1--38:21, 2020.

\bibitem[GS20]{Gupta020}
Anupam Gupta and Sahil Singla.
\newblock Random-order models.
\newblock In Tim Roughgarden, editor, {\em Beyond the Worst-Case Analysis of
  Algorithms}, pages 234--258. Cambridge University Press, 2020.

\bibitem[Hoe63]{Hoe63}
Wassily Hoeffding.
\newblock Probability inequalities for sums of bounded random variables.
\newblock {\em Journal of the American Statistical Association},
  58(301):13--30, 1963.

\bibitem[HR17]{DBLP:conf/soda/HobergR17}
Rebecca Hoberg and Thomas Rothvoss.
\newblock A logarithmic additive integrality gap for bin packing.
\newblock In {\em SODA}, pages 2616--2625, 2017.

\bibitem[JDU{\etalchar{+}}74]{johnson1974worst}
David~S Johnson, Alan Demers, Jeffrey~D Ullman, Michael~R Garey, and Ronald~L
  Graham.
\newblock Worst-case performance bounds for simple one-dimensional packing
  algorithms.
\newblock {\em SIAM Journal on computing}, 3(4):299--325, 1974.

\bibitem[JG85]{mffd}
David~S Johnson and Michael~R Garey.
\newblock A 71/60 theorem for bin packing.
\newblock {\em J. Complex.}, 1(1):65--106, 1985.

\bibitem[Joh73]{johnson1973near}
David~S Johnson.
\newblock {\em Near-optimal bin packing algorithms}.
\newblock PhD thesis, Massachusetts Institute of Technology, 1973.

\bibitem[Ken96]{DBLP:conf/soda/Kenyon96}
Claire Kenyon.
\newblock Best-fit bin-packing with random order.
\newblock In {\em SODA}, pages 359--364, 1996.

\bibitem[KK82]{KarmarkarK82}
Narendra Karmarkar and Richard~M Karp.
\newblock An efficient approximation scheme for the one-dimensional bin-packing
  problem.
\newblock In {\em FOCS}, pages 312--320, 1982.

\bibitem[KLMS84]{spaccamela}
Richard~M Karp, Michael Luby, and A~Marchetti-Spaccamela.
\newblock A probabilistic analysis of multidimensional bin packing problems.
\newblock In {\em Proceedings of the Sixteenth Annual ACM Symposium on Theory
  of Computing}, 1984.

\bibitem[KRTV18]{kesselheim2018primal}
Thomas Kesselheim, Klaus Radke, Andreas Tonnis, and Berthold Vocking.
\newblock Primal beats dual on online packing lps in the random-order model.
\newblock {\em SIAM Journal on Computing}, 47(5):1939--1964, 2018.

\bibitem[LL85]{lee-lee}
Chan~C Lee and Der-Tsai Lee.
\newblock A simple on-line bin-packing algorithm.
\newblock {\em J. ACM}, 32(3):562–572, July 1985.

\bibitem[Mey01]{meyerson2001online}
Adam Meyerson.
\newblock Online facility location.
\newblock In {\em Proceedings 42nd IEEE Symposium on Foundations of Computer
  Science}, pages 426--431. IEEE, 2001.

\bibitem[Mur88]{DBLP:journals/dam/Murgolo88}
Frank~D Murgolo.
\newblock Anomalous behavior in bin packing algorithms.
\newblock {\em Discret. Appl. Math.}, 21(3):229--243, 1988.

\bibitem[MY11]{MahdianY11}
Mohammad Mahdian and Qiqi Yan.
\newblock Online bipartite matching with random arrivals: an approach based on
  strongly factor-revealing lps.
\newblock In {\em STOC}, pages 597--606, 2011.

\bibitem[RT93a]{rhee1993lineB}
Wansoo~T Rhee and Michel Talagrand.
\newblock On-line bin packing of items of random sizes, ii.
\newblock {\em SIAM Journal on Computing}, 22(6):1251--1256, 1993.

\bibitem[RT93b]{rhee1993lineA}
Wansoo~T Rhee and Michel Talagrand.
\newblock On line bin packing with items of random size.
\newblock {\em Mathematics of Operations Research}, 18(2):438--445, 1993.

\bibitem[SL94]{simchi1994new}
David Simchi-Levi.
\newblock New worst-case results for the bin-packing problem.
\newblock {\em Naval Research Logistics (NRL)}, 41(4):579--585, 1994.

\bibitem[Ull71]{ullman1971performance}
Jeffrey~D Ullman.
\newblock The performance of a memory allocation algorithm.
\newblock {\em Technical Report}, 1971.

\bibitem[Wal12]{walsh2012knowing}
John~B Walsh.
\newblock {\em Knowing the odds: an introduction to probability}, volume 139.
\newblock American Mathematical Soc., 2012.

\end{thebibliography}
